%% file: main.tex
\documentclass[sigconf]{acmart}%\settopmatter{printfolios=false,printccs=false,printacmref=false}

\def\fullversion

\usepackage{caption}
\usepackage{bm}
\usepackage{titlecaps}

\usepackage{style}
\usepackage{bigstrut}
\usepackage{microtype}
\usepackage{tabularx}
\usepackage{multirow}
\usepackage{multicol}
\usepackage{rotating}
\usepackage{balance}
\usepackage{titlecaps}%\titlecap
\usepackage{soul}

\usepackage{amssymb}% http://ctan.org/pkg/amssymb
\usepackage{pifont}
\newcommand{\cmark}{\ding{51}}%
\input{macro}
\copyrightyear{2022}
\acmYear{2022}
\setcopyright{rightsretained}
\acmConference[SPAA '22]{Proceedings of the 34th ACM Symposium on
Parallelism in Algorithms and Architectures}{July 11--14,
2022}{Philadelphia, PA, USA}
\acmBooktitle{Proceedings of the 34th ACM Symposium on Parallelism in
Algorithms and Architectures (SPAA '22), July 11--14, 2022, Philadelphia,
PA, USA}
\acmDOI{10.1145/3490148.3538574}
\acmISBN{978-1-4503-9146-7/22/07}

\usepackage{array}
\begin{document}
\fancyhead{}
%\title[(Almost) Work-efficient Parallel Algorithm for Sequential Iterative Algorithms]{(Almost) Work-efficient Parallel Algorithms for Sequential Iterative Algorithms}
\title[Many Sequential Iterative Algorithms Can Be Parallel and (Nearly) Work-efficient]{Many Sequential Iterative Algorithms Can Be \\Parallel and (Nearly) Work-efficient}
  \author{Zheqi Shen}
  \affiliation{\institution{UC Riverside}\country{}}
  \email{zheqi.shen@email.ucr.edu}
  \author{Zijin Wan}
  \affiliation{\institution{UC Riverside}\country{}}
  \email{zwan018@ucr.edu}
  \author{Yan Gu}
  \affiliation{\institution{UC Riverside}\country{}}
  \email{ygu@cs.ucr.edu}
  \author{Yihan Sun}
  \affiliation{\institution{UC Riverside}\country{}}
  \email{yihans@cs.ucr.edu}

\hide{
}

%\setcopyright{none}

\input{abstract.tex}

%\ccsdesc[500]{Theory of computation~Shared memory algorithms}

%\keywords{a,b,c}

\maketitle
%\clearpage
%\newpage

%\pagestyle{plain}
%\pagenumbering{arabic}
%\setcounter{page}{0}

\makeatletter
\newcommand{\removelatexerror}{\let\@latex@error\@gobble}
\makeatother

%\newpage
\input{intro.tex}
\input{prelim.tex}

\input{datastructure.tex}
\input{prob.tex}

\input{type1.tex}
\input{type2.tex}
\input{exp.tex}
\section{Conclusion and Future Work}
In this paper, we used the \phaseparallel{} framework with general techniques to parallelize sequential iterative algorithms with certain dependences, and designed work-efficient and round-efficient algorithms for a variety of classic problems. Our results improved the previous theoretical bounds for many of them (e.g., the LIS and MIS algorithms). We also implemented these algorithms. Our results illustrated how work-efficiency and round-efficiency affect performance in practice (which matches our theory). For reasonable (not too large) input ranks,
our work-efficient algorithms achieved good parallelism and outperformed the sequential algorithms.
%Although our algorithm is the first parallel LIS algorithm with $\tilde{O}(n)$ work with $\tilde{O}(k)$ span ($k$ is the LIS size),
%it gets satisfactory performance only for very small $k$. %We believe reducing the $O(\log^2 n)$ overhead in work can improve the performance, which we leave as future work.
One interesting future direction is to reduce the $O(\log^2 n)$ overhead in the work of the LIS algorithm, which is also likely to improve its practical performance.

%\smallskip

%\myparagraph{Acknowledgement}.
\section*{Acknowledgement}
This work is supported by NSF grant CCF-2103483.
%\balance
\bibliographystyle{abbrv}
\bibliography{bib/strings,bib/main}

\appendix
\input{appendix.tex}

\end{document}

%% file: macro.tex
  \newcommand{\fname}[1]{\textit{#1}}
  \newcommand{\textfunc}[1]{\mf{#1}}
  \newcommand{\nodecircle}[1]{{\textcircled{\footnotesize{#1}}}}

  \newcommand{\bmath}[1]{\bm{#1}}

  \newcommand{\true}{\emph{true}\xspace}

  \newcommand{\zero}{\textit{zero}\xspace}
  \newcommand{\one}{\textit{one}\xspace}

  \newcommand{\parent}{\mathit{parent}\xspace}
  \newcommand{\degree}{\mathit{degree}\xspace}
  \newcommand{\flag}{\mathit{flag}\xspace}
  \newcommand{\frontier}{frontier\xspace}
  \newcommand{\dg}{dependence graph\xspace}
  \newcommand{\DG}{DG\xspace}
  \newcommand{\mathpivot}{\mathit{pivot}\xspace}
  \newcommand{\element}{element\xspace}
  \newcommand{\object}{object\xspace}

  \newcommand{\phaseparallel}{phase-parallel}
  \newcommand{\compatible}{compatible\xspace}
  \newcommand{\incompatible}{incompatible\xspace}
  \newcommand{\mis}{\mathit{MFS}}
  \newcommand{\mfs}{\mis}
  \newcommand{\rank}{\mathit{rank}}
  
  \newcommand{\feasible}{feasible\xspace}
  \newcommand{\ready}{ready\xspace}
  
  \newcommand{\finished}{finished\xspace}
  \newcommand{\unfinished}{unfinished\xspace}
  \newcommand{\obj}{\mathit{obj}}
  \newcommand{\setfam}{\mathcal{F}}
  \newcommand{\ind}{\mathcal{I}}
  \newcommand{\curset}{T\xspace}
  \newcommand{\dpvalue}{DP value\xspace}
  \newcommand{\DPvalue}{DP value\xspace}
  \newcommand{\MFS}{MFS\xspace}
  \newcommand{\block}{block\xspace}
  \newcommand{\pred}{\mathcal{P}}
  \newcommand{\rankbar}{\overline{\rank}}
  \newcommand{\relaxedrank}{relaxed rank}

  \newcommand{\mathdp}{\mathit{dp}\xspace}
  
  \newcommand{\mathrange}{\mathit{range}\xspace}

  \newcommand{\trange}{T_\mathrange}
  \newcommand{\tpivot}{T_\mathpivot}
  \newcommand{\todo}{\mathit{todo}}

  \newcommand{\pabst}{PA-BST\xspace}

  \newcommand{\timeend}{e}
  \newcommand{\timestart}{s}
  \newcommand{\ay}{A_{p_x}}
  \newcommand{\ax}{A_x}
  \newcommand{\leading}{pivot\xspace}

  \newcommand{\tastree}{TAS tree\xspace}
  \newcommand{\neighbor}{\mathcal{N}\xspace}
  \newcommand{\dmax}{d_{\mathit{max}}\xspace}

  \newcommand{\splitn}{\mf{Split}\xspace}

  \newcommand{\multifind}{\textfunc{Multi-Find}\xspace}
  \newcommand{\multiinsert}{\textfunc{Multi-Insert}\xspace}
  \newcommand{\multidelete}{\textfunc{Multi-Delete}\xspace}
  \newcommand{\multiupdate}{\textfunc{Multi-Update}\xspace}

  \newcommand{\forkins}{\texttt{fork}}

  \newcommand{\thread}{thread}

  \newcommand{\depth}{span\xspace}
  \newcommand{\TAS}{$\mf{TAS}$\xspace}
  
  \newcommand{\tas}{{\texttt{test\_and\_set}}\xspace}

  \newcommand{\Z}{\mathbb{Z}}
  \newcommand{\ifconference}{{{\ifx\fullversion\undefined}}}
  \newcommand{\iffullversion}{{{\unless\ifx\fullversion\undefined}}}

\makeatletter
\def\dfnt@space@setup{%
  \dfnt@preskip=\parskip
    \dfnt@postskip=0pt}
\makeatother

\newtheoremstyle{exampstyle}
  {.05in} % Space above
  {.05in} % Space below
  {} % Body font
  {.5em} % Indent amount
  {\sc \bfseries} % Theorem head font
  {.} % Punctuation after theorem head
  {.5em} % Space after theorem head
  {} % Theorem head spec (can be left empty, meaning `normal')
\theoremstyle{exampstyle} 
\theoremstyle{exampstyle} 

\crefname{section}{Sec.}{Sec.}
\crefname{theorem}{Thm.}{Thm.}
\crefname{lemma}{Lem.}{Lem.}
\crefname{corollary}{Col.}{Col.}
\crefname{table}{Tab.}{Tab.}

%% file: abstract.tex
\begin{abstract}
  To design efficient parallel algorithms, some recent papers showed that many sequential iterative algorithms can be directly parallelized,
  by identifying the \emph{dependences} between the input objects.
  This approach yields many simple and practical parallel algorithms, but there are still challenges in achieving work-efficiency and high-parallelism.
  Work-efficiency means that the number of operations is asymptotically the same as the best sequential solution.
  This can be hard for certain problems where the number of dependences between objects is asymptotically more than optimal sequential work bound,
  and we cannot even afford to generate them.
  To achieve high-parallelism, we want to avoid waiting for ``false dependences'' and process as many objects as possible in parallel.
  The goal is to achieve $\tilde{O}(D)$ span for a problem with the deepest dependence length $D$.
  We refer to this property as \emph{round-efficiency}.
  In this paper, we show work-efficient and round-efficient algorithms for a variety of classic problems and propose general approaches to do so.

  To efficiently parallelize many sequential iterative algorithms, we propose the \emph{\phaseparallel{} framework}.
  The framework assigns a \emph{rank} to each object and processes the objects based on the order of their ranks.
  All objects with the same rank can be processed in parallel.
  To enable work-efficiency and high parallelism, we use two types of general techniques.
  Type 1 algorithms aim to use range queries to extract all objects with the same rank, such that we avoid evaluating all the dependences.
  We discuss activity selection, unlimited knapsack, Dijkstra's algorithm, and Huffman tree construction using Type 1 framework.
  Type 2 algorithms aim to \emph{wake up} an object when the last object it depends on is finished.
  We discuss activity selection, longest increasing subsequence (LIS), greedy maximal independent set (MIS), and many other algorithms using Type 2 framework.

  All of our algorithms are (nearly) work-efficient and round-efficient.
  Many of them improve previous best bounds, and some of them (e.g., LIS) are the first to achieve work-efficiency with round-efficiency.
  We also implement many of them.
  On inputs with reasonable dependence depth, our algorithms are highly parallelized and significantly outperform their sequential counterparts.
\end{abstract}

  \hide{In parallel algorithm design, the parallel dependence relationships between objects in the problem,
  are often illustrated by a \emph{dependence graphs},
  where objects are represented as vertices and an edge from $u$ to $v$ means $v$ cannot be executed before $u$ finishes. }

%% file: intro.tex
%\vspace{-.05in}
\section{Introduction}
%\vspace{-.02in}
\label{sec:intro}
There are two goals in designing efficient parallel algorithms: to reduce work,
and to improve parallelism.
%For work, the goal is usually to achieve \emph{work-efficiency},
%which means the total number of operations is asymptotically the same as the best sequential solution.
\emph{Work-efficiency}, meaning that the \emph{work} (total number of operations) is asymptotically the same as the best sequential solution,
is crucial for practical parallel algorithms.
This is because nowadays and for the foreseeable future, the number of processors in a machine (up to thousands) is roughly polylogarithmic to input sizes.
Hence, a parallel algorithm is less practical if it blows up the work of the best sequential algorithm by a polynomial factor.
In this paper, we show (nearly) work-efficient parallel algorithms for a list of classic problems.
%Many of them are the first parallel solutions that are faster than the textbook sequential algorithms, to the best of our knowledge.

Our work is motivated by a list of recent papers that directly parallelize some sequential iterative algorithms
(i.e., sequential algorithms that iteratively process input objects)~\cite{blelloch2012internally,BFS12,fischer2018tight,hasenplaugh2014ordering,pan2015parallel,shun2015sequential,blelloch2016parallelism,blelloch2018geometry,blelloch2020randomized,blelloch2020optimal}.
%Most of these algorithms are shown to be pracical, since they are work-efficient, highly parallelized, and (at least conceptually) simple.
%, since they resemble the sequential algorithms.
%Many such parallel algorithms are work-efficient, highly-parallelized, and (conceptually) simple.
Their work-efficiency analysis usually directly follows the original sequential algorithm.
To achieve high parallelism from a sequential iterative algorithm,
the key is to identify the \emph{dependences} among ``objects'' (e.g., iterations, instructions, or input objects),
and process them in the proper order~\cite{jones1993parallel,gruevski2018laika,blelloch2012internally,BFS12,fischer2018tight,hasenplaugh2014ordering,pan2015parallel,shun2015sequential,blelloch2016parallelism,blelloch2018geometry,blelloch2020randomized,blelloch2020optimal}.
In particular, we want to avoid waiting for ``false dependences'' and to process as many objects as possible in parallel.
Such relationships can be modeled as a directed acyclic graph (DAG), referred to as the \defn{\dg{}} (\defn{\DG{}}).
Each object corresponds to a vertex in the \DG{},
and a directed edge from vertex $u$ to $v$ means $v$ can be executed only after $u$ is finished, and we say $v$ \defn{relies on} (or depends on) $u$.
%Dozens of practical parallel algorithms have been designed in the past decade
\hide{
These papers focused on showing that the \DG{s} are shallow, usually logarithmic to the input size, and proposed the general approaches such as the iterative \DG{} (IDG)~\cite{shun2015sequential} and the configuration \DG{} (CDG)~\cite{blelloch2016parallelism}.
However, these tools can only show the \DG{s} are shallow, but they do not directly yield work-efficient parallel algorithms.
In most cases, we have to look into each problem individually and design such an algorithm.
}

There are two existing general frameworks to design parallel algorithms using \DG{s}, but they both have limitations.
The first framework is \emph{deterministic reservations}~\cite{blelloch2012internally} (also used in~\cite{shun2015sequential}).
These algorithms run in rounds, and in each round, check the unfinished objects, execute the ``ready'' objects in parallel, and postpone the rest to later rounds.
%The ``span-efficiency'' holds trivially in that they only need $O(D)$ rounds
This framework provides good parallelism---for a \DG{} of depth $D$, we only need $O(D)$ rounds.
%It is easy to see that
In this paper, we define a computation as \defn{round-efficient} if it executes a \DG{} with depth $D$ in $\tilde{O}(D)$ span (longest dependence of instructions in the algorithm, formally defined in \cref{sec:prelim})\footnote{We note that round-efficiency does not guarantee optimal span, since round-efficiency is with respect to a given \DG{}. One can re-design a completely different algorithm that has a shallower \DG{} and a better span.}.
%This can be guaranteed by deterministic reservation, as it only needs $O(D)$ rounds.
%The essence of deterministic reservation guarantees round-optimality.
Despite the round-efficiency, deterministic reservations do not guarantee work-efficiency---the work in the worst case is $O(Dm)$, where $m$ is the number of edges of the \DG{} (using a topological sort sequentially only takes $O(m)$ work).
The second approach was proposed by Blelloch et al.~\cite{blelloch2020optimal} to prove work-efficiency of \DG{}-based algorithms,
but it only applies to when each vertex in the \DG{} has a constant in-degree.
As a result, most of the existing algorithms do not fit in these two frameworks,
and each~\cite{blelloch2016parallelism,blelloch2018geometry,blelloch2020randomized,blelloch2020optimal} uses specific design and analysis to get the work and span bounds (if any).
Moreover, previous approaches are \emph{edge-centric}---all edges in the \DG{s} are examined to find ready vertices/objects.
We observed that in many cases, even generating all edges in the \DG{} is work-inefficient.
Consider the classic longest increasing subsequence (LIS) problem that can be solved sequentially in $O(n\log n)$ work.
In the worst case, the \DG{} can contain $O(n^2)$ edges (see an example in \cref{fig:intro}).
It remains open whether we can design work-efficient parallel algorithms with non-trivial parallelism for many classic problems such as LIS, and even general approaches to achieve so.
%Our goal in this paper is to design (almost) work-efficient algorithms with round-optimality.

\myparagraph{Contributions in this paper.}
\hide{We propose \emph{a framework and general ideas to design (almost) work-efficient parallel algorithms with round-efficiency,
and algorithms in this framework for a variety of problems, most of them textbook greedy and dynamic programming (DP) algorithms. }
%Our approach is \emph{vertex-centric}, which overcomes the above-mentioned challenges.
%and
We call them \emph{sequential iterative algorithms} as their sequential solutions process all objects iteratively.}
We propose \emph{a general framework and algorithms in this framework to parallelize sequential iterative algorithms,
many of them textbook greedy and dynamic programming (DP) algorithms, that are (nearly) work-efficient and round-efficient.  }
Our approaches are \emph{vertex-centric}, which avoid examining all edges in the \dg{}.
%This involves several non-trivial algorithmic insights.
\hide{To do this, we first formalize the dependence of such problems using an algebraic structure (formally defined in \cref{sec:structure}).}
%illustrates possible correct executions for an iterative algorithm in parallel.
We define the \defn{rank} for each vertex in the \DG{} (an input object or a subproblem),
and we prove that rank fully captures the earliest
``phase'' that an object can be processed in a parallel algorithm.
We believe that defining rank simplifies parallel algorithm design for many problems.
%In a high-level, with properly defined rank function, all elements with the same rank can be executed in parallel.
Based on rank, we propose the \defn{\phaseparallel{} algorithm framework},
which processes all objects based on the ordering of ranks, and
round $i$ processes all objects with rank $i$ in parallel.
For example, the rank of an object in the LIS
%\zheqi{do we need to expand the abbreviation once?}
problem is the LIS length ending at this object.
The \phaseparallel{} algorithm will then process all objects with LIS size $i$ in round $i$, which finish in $k$ rounds for an input sequence with LIS length $k$.
%In addition, our framework can apply to parallelizing a list of sequential iterative algorithms.
We present the list of problems discussed in the paper, their ranks, and the cost bounds of our solutions in \cref{tab:overall}.

To achieve work-efficiency and round-efficiency, we propose two types of general ideas, referred to as Type 1 and Type 2
algorithms.
% several non-trivial algorithmic insights.
%Our key techniques can be sub in two types.

%We believe that the concept of \emph{rank} is crucial because it simplifies the designing of the entire parallel algorithm to only identify the rank of each object.
%We identified a few common patterns of the computation of the rank functions, and designed efficient and generic approaches to map them to efficient parallel algorithms.
\hide{In many cases, the rank functions can be computed by \defn{range queries} (1D or in higher dimensions), which can be efficiently answered by a line of recent work~\cite{blelloch2016justjoin,sun2018pam,sun2019parallel}.}
%Based on their rank functions, we categorize them into Type I and Type II algorithms.

\begin{figure}
  \centering
  \includegraphics[width=\columnwidth]{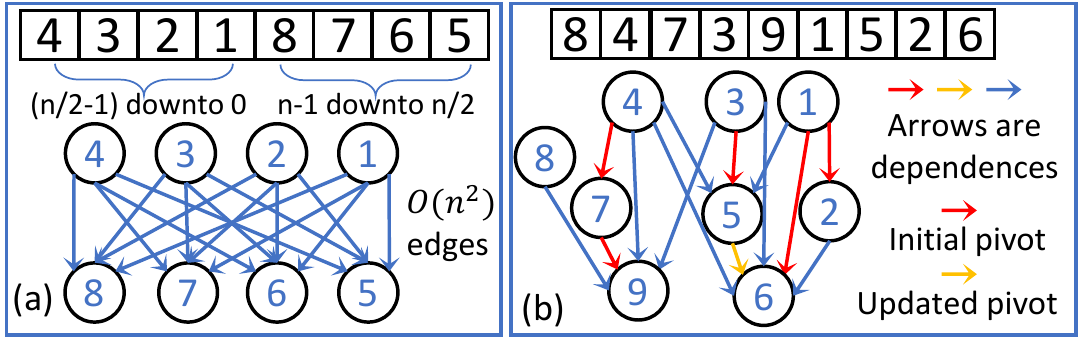}
  \caption{\small \textbf{Examples of Longest Increasing Subsequence (LIS).}
  (a) An example where we have $O(n^2)$ dependences between objects for input size $n$.
  (b) An example of how our algorithm processes LIS.
  Each object $x$ chooses a pivot among its predecessors (red arrows). We check the readiness
  of $x$ only when its pivot finishes. If $x$ is still not ready, we update the pivot to another unfinished object.
  This avoids checking all edges in the \DG{}.
  }\label{fig:intro}%\vspace{-.05in}
\end{figure}

%In Type I algorithms, the rank functions are piecewise monotonic constant functions.
%This indicates that if we maintain all objects by a parallel binary search tree, the set of ready objects can be acquired efficiently by 1D range queries from a dynamic map.
Type 1 algorithms aim to identify objects to be processed in round $i$ efficiently.
%Based on our \phaseparallel{} framework, this is equivalent to extract all objects with rank $i$.
To do this, we use \emph{range queries} based on parallel augmented trees~\cite{sun2018pam} (see \cref{sec:prelim}),
which take polylogarithmic work, instead of work proportional to the number of relevant edges in the \DG{}.
This idea applies to many greedy or DP algorithms, %that process all \object{s} or states in a certain order,
and our algorithms use range queries to find the maximal set of parallelizable objects.
Type 1 algorithms include activity selection, unlimited knapsack, Dijkstra's algorithm, and Huffman tree.
We note that some of these ideas may have been known. We do not claim them as the main contribution,
but use them to exemplify our framework.

%Type II algorithms are more interesting and involved.
Our Type 2 algorithms aim at \emph{waking up} the ready \object{s} at the right time.
%The readiness of each object in these problems has to be checked individually, and can rely on all $O(n)$ other objects.
%To avoid checking all dependence edge and work-efficiently process them, we not only need efficient rank functions to check readiness, but also cannot check each object in every round.
%This is to avoid the inefficiency in checking the readiness of all elements every round in previous work.
Instead of checking the readiness of all \object{s} in each round (as is done in previous work), we hope to touch an \object{} \emph{only when it is (almost surely) ready}.
%The rank computation for Type II algorithms are similar to what are computed in the original sequential algorithms (range queries in many cases), so they are also (nearly) work-efficient.
%We also propose on a novel \defn{wake-up strategy} to check an objects at appropriate times, ideally \emph{only when it is (almost surely) ready}.
%We propose two approaches for Type II problems.
One of our approaches is to assign a \emph{pivot} $p_x$ to each \object{} $x$, which is an \object{} $x$ relies on.
Only when $p_x$ finishes, we check if $x$ is ready. If $x$ is ready, we process it in the next round.
If not, we update $x$'s pivot to another unfinished \object{} it relies on.
\cref{fig:intro}(b) shows an example of LIS.
Although there are many dependences (all edges) in the \DG{}, each \object{} only picks one pivot (the red edges).
When the pivot is ready, the object itself is likely to be ready, but if not (e.g., \nodecircle{6} first picks \nodecircle{1}),
it selects a new pivot (the yellow edge from \nodecircle{5}). Therefore, only a small fraction of edges in the \DG{} (the red and yellow ones) are evaluated, which saves work. %We use this approach to solve activity selection and LIS problems.
%This guarantees that each object only has a constant or $O(\log n)$ edges evaluated by the algorithm, which avoids much extra work.
%We proved that each object is waked up at most $O(\log n)$ times, which incurs only a polylogarithmic overhead in work. In some algorithms we can even find exact pivots such that there is no failed wake up (e.g., the activity selection problem in Section \ref{sec:activity2}).
%The theoretical efficiency is guaranteed by the efficiency of the rank function plus additional randomness introduced carefully in this process.
%We also propose other approaches for other problems. For example, for the greedy maximal independent set (MIS) algorithm and other similar algorithms,
%we design a structure to identify the readiness of an object when its last predecessor in \DG{} finishes.
Another approach is based on a new structure to identify a ready object when its last predecessor in \DG{} finishes,
which we apply to the greedy maximal independent set (MIS) algorithm and other similar algorithms.
Our approach is based on a new data structure called \emph{\tastree{}}, which makes use of the atomic \emph{test-and-set} operation (formally defined in \cref{sec:prelim}).
For MIS, our algorithm is work-efficient, and improves existing span bound~\cite{BFS12} from $O(\log^3 n)$ to $O(\log^2 n)$.
%, such as identifying the readiness of an object when its last predecessor in \DG{} is finished.
%We use this idea to design the greedy maximal independent set (MIS) algorithm and graph coloring algorithm.

%(mention span).

We believe that our framework applies to a broad set of problems.
We picked the examples by reviewing the problems in Cormen, Leiserson, Rivest, and Stein~\cite{clrs}.
%, and noted that our framework applies to many of the greedy
%(all example problems in \S16: Greedy Algorithms, Dijkstra's algorithm, but not Prim's and Kruskal's MST algorithms),
%and dynamic programming algorithms.
%We believe that our approach is helpful for other similar problems.
All algorithms in this paper are (nearly) work-efficient (only the LIS algorithm has an $O(\log^2 n)$ factor of overhead), and round-efficient.

\hide{
Many of our new algorithms have improved bounds over the best previous results.
%Among them, the most interesting result should be the longest increasing sequence (LIS) algorithm, since LIS is a widely studied problem and has many applications.
For example, our LIS algorithm (\cref{algo:lis}) has $\tilde{O}(n)$ work %\zheqi{what is $\tilde{O}(n)$}
and $\tilde{O}(k)$ span for LIS length $k$.
We note that the best previous solution is either the sequential algorithm where the work and span are both $O(n\log n)$~\cite{Knuth69}, or parallel algorithm with polylogarithmic span but incurring significant overhead in work ($\Omega(n^2)$ work)~\cite{krusche2009parallel}. Our new algorithm greatly improves the existing theory results.
Our framework also applies to many existing algorithms in~\cite{blelloch2012internally,BFS12,fischer2018tight,hasenplaugh2014ordering,pan2015parallel,shun2015sequential,blelloch2016parallelism,blelloch2018geometry,blelloch2020randomized,blelloch2020optimal},
improves the bounds for many of them (e.g., the greedy MIS algorithm),
and provides a much simpler way to understand these algorithms.%, and the similarities and differences among them.
%Previous algorithms (citation) either use $\Omega(n^{1.5})$ work, or have $\Theta(n)$ span.
%Although when $k=O(n)$ our algorithm degenerates to the sequential algorithm (and practically slower due to the parallel overhead), we carefully engineered the algorithm and showed that on a reasonably large input parameter space, our algorithm can outperform the textbook LIS algorithm.
}

Our framework applies to many existing algorithms in~\cite{blelloch2012internally,BFS12,fischer2018tight,hasenplaugh2014ordering,pan2015parallel,shun2015sequential,blelloch2016parallelism,blelloch2018geometry,blelloch2020randomized,blelloch2020optimal},
improves the bounds for many of them (e.g., the greedy MIS algorithm),
and provides a much simpler way to understand these algorithms.
As another example, our LIS algorithm (\cref{algo:lis}) has $\tilde{O}(n)$ work
and $\tilde{O}(k)$ span for LIS length $k$.
Parallel LIS is widely studied~\cite{galil1994parallel,krusche2009parallel,krusche2010new,alam2013divide,seme2006cgm,thierry2001work,nakashima2002parallel,nakashima2006cost}.
Our algorithm is the first to achieve near-work-efficiency and round-efficiency,
which is advantageous especially for small output size.
We review the literature of parallel LIS in \cref{sec:lis}.
%in \cref{sec:lis} has better bounds when $k=o(n^{2/3})$.
Our algorithm is also simpler and more practical.
We are unaware of any implementations of these previous algorithms with
competitive performance to the standard sequential LIS algorithm with $O(n\log n)$ work.

We implement many of these algorithms, and test them as a proof-of-concept to show how work- and round-efficiency affect practical performance.
%Many of our implementations are the first known parallel version that can outperform the best sequential algorithms on certain inputs.
%For example, we are unaware of any practical LIS algorithms and implementations that can outperform the $O(n\log n)$ serial textbook algorithm.
Although there is parallel overhead and our worst-case span is $\tilde{O}(n)$, %\zheqi{worst span is $\Theta(n)$?},
our work-efficient algorithms achieve significant speedup over the sequential algorithm in a reasonably large input parameter space.
%We also compare Type 1 and Type 2 algorithms experimentally using activity selection since it fits in both frameworks.
%We summarize the contributions of this paper as follows:
Our contributions include:
%For LIS, we are unaware of any practical LIS algorithms and implementations that can outperform the $O(n\log n)$ serial textbook LIS algorithm.

%This is the first parallel LIS algorithm that can be faster than the textbook LIS algorithm on non-trivial cases.
%Meanwhile, the input range that can benefit from our parallel algorithm will monotonically increase in the future, since larger processor count $p$ will benefit \cref{alg:lis} but not the serial algorithm.

\begin{itemize}[wide, labelwidth=!, labelindent=0pt]
  \item The \phaseparallel{} framework to parallelize sequential iterative algorithms based on the concept of \emph{rank} defined in this paper.
  \item Two general techniques for \phaseparallel{} algorithms to achieve work-efficiency and round-efficiency, based on range queries and a wake-up strategy to identify ready objects, respectively.
  \item The first \textbf{nearly work-efficient} ($\tilde{O}(n)$ work) LIS algorithm with \textbf{round-efficiency} ($\tilde{O}(k)$ for LIS length $k$) and its implementation.
  %\zheqi{consistently using D or k to represent the LIS length?}
  %We also implemented the algorithm and show that the algorithm achieves reasonable speedup that matches our theory.
      %This is \textbf{the only parallel LIS implementation we know of that is faster than the work-efficient sequential version}.
  \item A \textbf{work-efficient} greedy maximal independent set (MIS) algorithm with \textbf{$\boldsymbol{O(\log n\log \dmax)}$ span} \whp{} in the binary-forking model (defined in \cref{sec:prelim}), where $n$ is the number of vertices and $\dmax$ is the maximum degree in the graph. This improves the previous best span bound of $O(\log^3 n)$.
  %\item We proposed two frameworks of parallel algorithms, based on the techniques of type 1 and 2 in this paper, respectively, for the activity selection problem. This is the first parallel algorithm for this problem we know of. Although not very complicated, this algorithm reveals the connections and intrinsic properties of the two types of algorithms in this paper. We also implemented the algorithms and experimentally evaluate them.
  \item Two algorithms for the activity selection problem. This is the first parallel algorithm for this problem we know of.
  %Although not complicated, this algorithm reveals the connections and intrinsic properties of the two types of algorithms in this paper.
  Although not complicated, they reveal the connections of the two types of algorithms in this paper.
      They are work-efficient and round-efficient. For the unweighted version, we also provide an algorithm with $O(\log n)$ span \whp.
      %We also implemented the algorithms and experimentally evaluate them.
  \item Many other simple and interesting algorithms using our framework, including Huffman tree, SSSP, unlimited knapsack, etc.
  %The idea to many of these problems may be known, but we show that they elegantly fit in and can be easily analyzed in our framework.
  \item Implementations and experimental studies of these algorithms.
\end{itemize}

\hide{
%, with non-trivial span (polylogarithmic or parameterized, depending on the depth of the \dg{}).
Our framework generalizes to a variety of classic problems, including Huffman tree (HT), single source shortest path (SSSP), activity selection (AS), longest increasing subsequence (LIS), greedy maximal independent set (MIS), list ranking, etc. Note that some of these algorithms are essentially similar to existing solutions, but properly fit in our new framework, which we believe provides new insights to better understand these algorithms. The others are newly proposed algorithms which improve
 existing work and/or span bounds.
Our key idea lies in two aspects corresponds to the two challenges above.
First, to reveal the parallel dependence in the \dg{}, we formalize these problems by assigning a \defn{rank} to each object in the computation based on the application.
In most of the algorithms, the set of objects with the same rank is the maximal set of objects that can be processed in parallel.
As such, we propose \defn{\phaseparallel{} algorithms} that processes objects based on their ranks in order (i.e., first process all rank 1 objects, then rank 2, etc.).
As an example, for LIS, the rank of each object is defined as the LIS length up to this object. Therefore, all objects with LIS length $i$ are ready to be processed in round $i$ in parallel.
We show that this enables algorithms with low span, which is always $\tilde{O}(D)$ for a \dg{} of depth $D$.
Secondly, to identify the appropriate time to process each object, many algorithms take advantage of \defn{range queries} to detect the readiness efficiently.
The range queries needed is to answer some aggregate value over all objects in a key range in a dynamic map (key-values). The the keys can also be in higher
dimensions (e.g., a 2D range query).
This can be supported efficiently by the parallel augmented binary search trees proposed by a line of recent work~\cite{blelloch2016justjoin,sun2018pam,sun2019parallel}.
This avoids exhausted check of readiness and enables work-efficiency.
As a result, the work bound can be sublinear to the number of dependencies in the problem.
We describe our algorithms in two types.

The first type of algorithms extract the ready set of objects by using range queries.
In particular, all objects with a certain rank can be extracted by a range query from a dynamic map, which can be
maintained efficiently by a parallel binary search tree.
This type of algorithms appear mostly in greedy algorithms where sequential algorithms tend to process a greedy choice one at a time,
and our algorithm will identify the maximal set of objects that can be parallelized by a range query based on the greedy strategy.
This technique applies to three applications including Huffman tree, single source shortest path, and activity selection.

The second type of algorithms rely on properly \defn{wake up} an object at the appropriate time.
Instead of checking the readiness of each object in each round, we wish the algorithms
would try to check an object \emph{only when it is (almost surely) ready}.
This again, can benefit from range queries (e.g., our new LIS algorithms use a 2D range query to quickly check
the readiness).
This type of algorithms include the greedy maximal independent set (MIS), activity selection,
list ranking, longest increasing subsequence and some other interesting problems.

Our contribution include:

\begin{enumerate}
  \item We propose a new, general framework for designing and analyzing parallel algorithms from sequential iterative algorithms based on the concept \emph{rank} defined in this paper.
  \item We propose a high-level idea of taking advantage of range queries using augmented binary search trees to enable work-efficiency of parallel algorithms.
  \item Our new algorithm for Longest Increasing Subsequence (LIS) is the first algorithm that is \textbf{almost work efficient} ($\tilde{O}(n)$ work) with \textbf{non-trivial span} ($\tilde{O}(D)$ for LIS length $D$). We note that the best previous solution is either the sequential algorithm where the work and span are both $O(n\log n)$, or parallel algorithm with polylogarithmic span but incurring significant overhead ($\Omega(n^2)$ work). Our new algorithm greatly improves the existing results in theory. We also implemented the algorithm and show that the algorithm achieves reasonable speedup that matches our theory. This is \textbf{the only parallel LIS implementation we know of that is faster than the work-efficient sequential version}.
  \item Our new algorithm for greedy Maximal Independent Set (MIS) is \textbf{work-efficient} and \textbf{$\boldsymbol{O(\log n\log \dmax)}$ span} in the binary-forking model (see definition in \cref{sec:prelim}), where $n$ is the number of vertices and $\dmax$ is the maximum degree in the graph. This improves previous best algorithm of $O(\log^3 n)$ span.
  \item We proposed two frameworks of parallel algorithms, based on the techniques of type 1 and 2 in this paper, respectively, for the activity selection problem. This is the first parallel algorithm for this problem we know of. Although not very complicated, this algorithm reveals the connections and intrinsic properties of the two types of algorithms in this paper. We also implemented the algorithms and experimentally evaluate them.
  \item We also apply our ideas to a variety of simple and interesting problems, including Huffman tree, single source shortest paths, list ranking, and Whac-a-mole problem. The parallel solutions to many of these problems are known, but we show that they elegantly fit in our framework and can be easily analyzed in our framework.
\end{enumerate}

One of the most important measurement of a parallel algorithm is its \emph{work-efficiency}, i.e., whether its work (the total number of operations performed) is the (almost) same as the best sequential algorithm.
Work-efficiency is also important for practical performance. Using a standard work-stealing scheduler (which is used in most modern parallel languages), the running time of
a parallel algorithm is $W/P+O(D)$ with high probability, where $W$ is the work, $D$ is the span and $P$ is the number of processors\footnote{Work, span, and with high probability are defined in \cref{sec:prelim}.}. On modern multicore machines where $P\ll W$, the running time is mostly dominated by the work. In this paper, we propose a framework for designing work-efficient parallel algorithms, and apply them to several classic problems to give work-efficient parallel solutions for them.

A line of recent work showed that, in many cases, we do not need to design brand-new parallel algorithms, but instead,
many sequential algorithms are already ``parallel'', in that the logical dependence of the operations in the algorithm does not form a serial chain.
Such examples include maximal independent set and maximal matching~\cite{BFS12,fischer2018tight}; graph coloring~\cite{hasenplaugh2014ordering}; correlation clustering~\cite{pan2015parallel}; random permutation, list/tree contraction~\cite{shun2015sequential}; comparison sorting, linear programming, smallest enclosing disk, and closest
pairs~\cite{blelloch2016parallelism}; Delaunay triangulation~\cite{blelloch2016parallelism,blelloch2018geometry} and convex hull~\cite{blelloch2020randomized}; least-object lists and strongly connected components~\cite{blelloch2016parallelism}.
}

\hide{
An effective tool in designing and analyzing such algorithms is \emph{\dg{}} (\DG{}), which models the objects and their dependence relationship in a directed acyclic graph (DAG).
In the \dg{}, each object in the algorithm, such as an iteration or an input object, is represented as a vertex.
We say $y$ \emph{depend on} $x$, if there is a directed arc from vertices $x$ to $y$, which indicates that $y$ can be processed only after $x$ is finished.
By observing the graph, when all predecessors of a vertex $x$ have been finished, $x$ is \defn{ready} to be processed, and multiple ready vertices can be executed in parallel.
As such, designing an \emph{efficient} parallel algorithm boils down to two challenges: 1) identifying
the \dg{} of objects in the algorithm, ideally a shallow one,
and 2) designing an scheme to process each objects at an appropriate time (e.g., as soon as it is ready) %without
%Designing an \emph{efficient} parallel algorithm involves identifying a shallow \dg{}
%asymptotically increasing the work (i.e., total number of operations),.
while bounding the overhead of parallelization, such as checking the readiness.
A shallow \dg{} indicates low span (i.e., parallel time), and bounding the overhead indicates work (i.e., total number of operations) efficiency.

Two important observations can be made from \dg{}. First of all, a set of vertices that do not depend each other can be processed in parallel, and thus in each round, we can process a maximal set of .
Secondly, if we process  the parallel span (longest dependence chain of operations in the algorithm) is lower bounded by the depth of the \dg{}

Inspired by this simple observation, a variety of previous work proposed parallel solutions either to an individual problem, or formalized reasonably general framework working for a range of problems. One of the successful framework among them is the \defn{deterministic reservation}~\cite{blelloch2012internally} and follow-up work using similar ideas \cite{blelloch2016parallelism,blelloch2020randomized,blelloch2020optimal,blelloch2018geometry,shun2015sequential,and,more}.
In a high-level, the algorithm runs in \emph{rounds}. In each round, the algorithm checks all (or a prefix of) unfinished objects and identify the ready ones,
and process them all in parallel.
We call all objects processed in the current round the \defn{frontier}.
Some algorithms also remove the processed objects from the unfinished object set at the end of each round, usually done by a parallel \fname{pack}.
Many algorithms are designed directly or indirectly based on this idea, many of them work efficient with low span.
Interestingly, the framework itself does not guarantee work-efficiency, but ``span-efficiency'' holds somehow trivially as it only needs $O(D)$ rounds, where $D$ is the depth of \dg{}, and this usually gives $\tilde{O}(D)$ span.
It is more challenging to enable work-efficiency, and the work bounds (if any) for each algorithm is shown specifically for each individual algorithm.
A naive implementation may require $O(nD)$ work, where $D$ is the depth of the \dg{}, and $n$ is the number of objects in the \dg{}.
Some algorithms achieve $O(m)$, where $m$ is the number of edges in the \dg{}, using some designs specific to each problem. We review these algorithms in more details in \cref{sec:related}.
This, however, does not provide satisfactory solution to some problems.
Consider one of the examples using in this paper, which is the longest increasing subsequence (LIS, see its formal definition in \cref{sec:prelim}).
The dependence can be defined by observing that an object can be processed only when all objects before it with a smaller value have been processed.
This implies that the possible number of dependencies between objects can be $O(n^2)$ (consider a strictly increasing sequence, where each object depend on all previous objects),
but the best known sequential solution costs $O(n\log n)$ work. Similar situation applies to many other problems in this paper.
To achieve work-efficiency, we cannot even afford exploring all dependences in the \dg{}.
%Some algorithms accommodate this by checking the successors after processing a object, which reduces work to $O(m)$, where $m$ is the number of edges in the \dg{}.
%Some more recent work focused on enabling asynchrony to avoid exhausted checking of readiness~\cite{blelloch2020optimal,blelloch2020randomized}, but
%such technique is only applicable to \dg{s} with constant fan-in.
The challenge lies in properly identifying the ready objects in the \dg{} efficiently, even without evaluating all edges in the \dg{}.
}

%many parallel solutions for a wide range of problems have been proposed, either individually or

\hide{
The key in parallel algorithm design is to identify the dependences
between objects, i.e., which objects are not parallelizable.
There are two goals in designing efficient parallel algorithms: work-efficiency (the total number of operations is asymptotically the same as the best sequential solution),
and high-parallelism (to avoid waiting for ``false dependences'' and to execute as many operations as possible in parallel).
\emph{Work-efficiency} is crucial for practical parallel algorithms.
% which means the work (the total number of operations performed) $W$ is asymptotically (almost) the same as the best sequential algorithm.
This is because nowadays and for foreseeable future, the number of processors $P$ in a machine (up to thousands) is roughly polylogarithmic to input sizes.
Hence, a parallel algorithm is less likely to be practical if it blows up the work of the best sequential algorithm by polynomial factor.
In this paper, we will show work-efficient or nearly work-efficient parallel algorithms for a list of classic problems.
Many of them are the first parallel solutions that are faster than the textbook sequential algorithms, to the best of our knowledge.

The second key aspect is to identify the \emph{dependence relationships} among objects, and process them in the right order.
Such relationships can be modeled as a DAG, referred to as a \emph{\dg{}} (\DG{}).
Vertices in a \DG{} can be iterations, instructions, or input objects, and a directed edge from vertex $u$ to $v$ means $v$ can be executed only after $u$ is finished.
%Dozens of practical parallel algorithms have been designed in the past decade
Based on the idea of \DG{}, many parallel algorithms have been designed~\cite{blelloch2012internally,BFS12,fischer2018tight,hasenplaugh2014ordering,pan2015parallel,shun2015sequential,blelloch2016parallelism,blelloch2018geometry,blelloch2020randomized,blelloch2020optimal}, by analyzing the parallel dependence in classic sequential algorithms.
We summarize them in \cref{sec:related}.
%Such examples include maximal independent set and maximal matching~\cite{BFS12,fischer2018tight}; graph coloring~\cite{hasenplaugh2014ordering}; correlation clustering~\cite{pan2015parallel}; random permutation, list/tree contraction~\cite{shun2015sequential,blelloch2020optimal} and range minimum queries~\cite{blelloch2020optimal}; comparison sorting, linear programming, smallest enclosing disk, and closest pairs~\cite{blelloch2016parallelism}; Delaunay triangulation~\cite{blelloch2016parallelism,blelloch2018geometry} and convex hull~\cite{blelloch2020randomized}; least-element lists and strongly connected components~\cite{blelloch2016parallelism}.
%These papers focused on showing that the \DG{s} are shallow, usually logarithmic to the input size~\cite{shun2015sequential,blelloch2016parallelism}.
There also exist general frameworks to design such parallel algorithms.
One successful framework is the \emph{deterministic reservations}~\cite{blelloch2012internally} and follow-up work~\cite{}.
The algorithms run in rounds, and in each round check all (or a prefix of) unfinished objects, execute those that are ``ready'', and postpone the rest to the next round.
Interestingly, instead of work-efficiency, ``span-efficiency'' holds somehow trivially in that they only need $O(D)$ rounds for a \DG{} of depth $D$.
Given a parallel \DG{} with depth $D$, we say a computation is \defn{round-optimal}\footnote{We note that round-optimal does not guarantee optimality in span. This is because
round optimality is with respect to a given \DG{}. However, one can re-design a completely different algorithm that have a shallower \DG{} and get better span.} when it has span $\tilde{O}(D)$.
%This can be guaranteed by deterministic reservation, as it only needs $O(D)$ rounds.
%The essence of deterministic reservation guarantees round-optimality.
However, they do not always guarantee work-efficiency---the work in the worst case is $O(Dm)$ where $D$ and $m$ are the depth and the number of edges of the \DG{}, respectively, and one has to prove work bounds (if any) for each problem separately.
Another approach by Blelloch et al.~\cite{blelloch2020optimal} showed a general way to prove work-efficiency, but only applies to when each vertex in the \DG{} has a constant degree.
In addition, all previous approaches are \emph{edge-centric}---all edges in the \DG{s} are examined to decide the readiness of the vertices/objects.
We observed that in many cases, we cannot even afford generating the edges in the DG work-efficiently.
Consider the classic longest increasing subsequence (LIS) problem can be solved sequentially in $O(n\log n)$ work. In the worst case the \DG{} can contain $O(n^2)$ edges (see an example in Figure \ref{fig:intro}).
It remains open on designing work-efficient parallel algorithms with non-trivial parallelism for many problems, or even general approaches to achieve so.
%Our goal in this paper is to design (almost) work-efficient algorithms with round-optimality.
} 

%% file: prelim.tex
\input{overalltable.tex}
\section{Preliminaries}
\label{sec:prelim}

\myparagraph{Notations.} For a sequence $s$, $s_i$ or $s[i]$ denotes the $i$-th
element, and $s_{i \ldots j}$ or $s[i\ldots j]$ denotes the $i$-th to the $j$-th elements 
in $s$. %We say $O(f(n))$ \defn{with high probability} (\whp{}) to indicate $O(cf(n))$ with probability at least $1-n^{-c}$ for $c \geq 1$. 
We use the term $O(f(n))$ \defn{with high probability} (\whp) in $n$ to indicate the bound
$O(k f(n))$ holds with probability at least $1-1/n^k$
for any $k \ge 1$. With clear context we drop ``in $n$''.

\myparagraph{Longest Increasing Subsequence (LIS).} Given a sequence $s_{1\ldots n}$, $s_{1\ldots m}'$ is a subsequence of $s$ if $s'_i=s_{k_i}$, where $k_1< k_2<\dots k_m$.
Given a sequence $s$, the \defn{longest increasing subsequence} (LIS) problem finds the longest subsequence $s^{*}$ of $s$ where $\forall i, s^{*}_i<s^{*}_{i+1}$.
\hide{
Sequentially, this problem can be solved using a dynamic programming (DP) algorithm. Let $\mathdp[i]$ (called the \defn{DP value} of element $i$) be the longest increasing subsequence of $s_{1\ldots i}$ ending with $s_i$.
Then
$$\mathdp[i]=\max(1,\max_{j<i,a[j]<a[i]}\mathdp[j]+1)$$
The problem and DP recurrence generalize to many applications where $\mathdp[i]=\max(1,\max_{j<i,a[j]<a[i]}\mathdp[j]+w(i))$, where $w(i)$ is a function of $i$. %In this paper we further use another classic problem of activity selection as an example to illustrate how our algorithm extends to more DP and greedy problems.
}

\myparagraph{Dependence Graph.}
A sequential iterative algorithm processes each input object in a given order.
The dependence graph (\DG{}) represents the processing dependence
of input objects.
Each vertex in the \DG{} denotes an input object.
An edge from $x$ to $y$ means that object $y$ can be processed only when $x$ has been finished, and we say $y$ \defn{relies on} $x$. %\zheqi{relies on $x$} in this case.
We say an object is \defn{\finished} if it has been processed, and \defn{\unfinished} otherwise.
We say an object is \defn{\ready} if all its predecessors in the \DG{} %\zheqi{is the dependence graph is the same as parallel dependence graph?}
have been finished. %and \defn{\unready} otherwise
For two objects $x$ and $y$, where $y$ relies on $x$ and $x$ is \unfinished{}, we say $x$ \defn{\block{s}} $y$.
%\zheqi{For two objects $x$ and $y$ where $y$ relies on $x$ and $x$ is \unfinished{}, we say $x$ \block{s} $y$.}
%We use $\pred(x)$ as the set of all predecessors of $x$ in the \DG{}.

\myparagraph{Parallel Computational Model.}
%We use the work-\depth{} model for fork-join parallelism with binary forking to analyze parallel
%algorithms~\cite{clrs,blelloch2020optimal} as is used many recent papers on parallel algorithms~\cite{agrawal2014batching,Acar02,blelloch2010low,BCGRCK08,BG04,Blelloch1998,blelloch1999pipelining,BlellochFiGi11,BST12,Cole17,CRSB13,BBFGGMS16,dinh2016extending,chowdhury2017provably,blelloch2018geometry,dhulipala2020semi,BBFGGMS18,Dhulipala2018,blelloch2020randomized,gu2021parallel}.
We use the work-\depth{} model on the binary-forking model (with \tas{}) to analyze parallel
algorithms~\cite{clrs,blelloch2020optimal} as is used in many recent papers~\cite{agrawal2014batching,Acar02,blelloch2010low,BCGRCK08,BG04,Blelloch1998,blelloch1999pipelining,BlellochFiGi11,BST12,Cole17,CRSB13,BBFGGMS16,dinh2016extending,chowdhury2017provably,blelloch2018geometry,dhulipala2020semi,BBFGGMS18,blelloch2020randomized,gu2021parallel}.
We assume a set of \thread{}s that share a memory.  Each \thread{} supports standard RAM instructions and a \forkins{}
instruction that forks two new child \thread{}s.
When a \thread{} performs a \forkins{}, the two child \thread{}s both start by running the next instruction, and the original \thread{} is suspended until both children terminate.
A computation starts with a single {root} \thread{} and finishes when that root \thread{} finishes.
We use atomic operation \tas{} (\TAS{}), which checks whether a memory location is of zero, sets it to one if so, and return its old value.
We say a \TAS{} is \defn{successful} if it changes zero to one and \defn{unsuccessful} otherwise.
One can use \TAS{} to implement a \emph{join} operation in the standard fork-join model~\cite{blelloch2020optimal}.
%It has been shown that binary fork-join model is at least as powerful as binary forking with \TAS{}~\cite{blelloch2020optimal}.
An algorithm's \defn{work} is the total number of instructions, and
the \defn{span} (depth) is the length of the longest sequence of dependent instructions in the computation.
Note that a parallel for-loop incurs $O(\log n)$ span because of binary-forking.
We can execute the computation efficiently using a randomized work-stealing scheduler both in theory and in practice~\cite{clrs,blelloch2020optimal,blumofe1999scheduling}.
%Given an algorithm of work $W$ and span $S$, the randomized work-stealing scheduler can execute the algorithm in $W/P+O(S)$ time \whp{}~\cite{BL98,ABP01}.
%\yihan{test and set}

%% file: overalltable.tex
\begin{table*}
\centering
\small
\renewcommand{\arraystretch}{1.1}
\begin{tabular}{lllccc}
  %\hline
   & \bf Feasible Condition for objects $x$ and $y$ &  \multicolumn{1}{c}{$\bmath{\rank(x)}$} & \multicolumn{1}{c}{\bf Type} & \multicolumn{1}{c}{\bf Work} & \multicolumn{1}{c}{\bf Span}\\
   \hline
  \bf Activity Selection & \multirow{2}{*}{$x$ and $y$ do not overlap} & The maximum number of non-  &\multirow{2}{*}{1\&2} &\multirow{2}{*}{$O(n\log n)$} &\multirow{2}{*}{$O(\rank(S)\log n)$}\\
  (general)&&overlapping activities ending at $x$&\\
  \hline
  \multirow{2}{*}{\bf Unlimited Knapsack} & Solution to weight $y$ can contain  & $x/w^{*}$, where $w^{*}$ is the minimum  &\multirow{2}{*}{1}&\multirow{2}{*}{$O(Wn)$} & \multirow{2}{*}{$O(\rank(S)\log (w^*n))$}\\
  &solution to the subproblem of weight $x$&weight\\
  \hline
  \bf Huffman Tree & $y$'s Huffman code is a prefix of $x$ & Subtree height of $x$  & 1 &$O(n\log n)$&$O(\rank(S)\log n)$\\

  \hline
  \multirow{2}{*}{\bf Dijkstra's} & \multirow{3}{*}{$x$ is on the shortest path to $y$} &Hop distance from $x$ to the source  & \multirow{3}{*}{1}&\multirow{3}{*}{$O(m\log n)$} & \multirow{3}{*}{$O\left(\frac{\max_{v\in V} d(v)}{\min_{e\in E} w(e)} \log n\right)$}\\
  \multirow{2}{*}{\bf Algorithm}&&on shortest path tree&&&\\
  &&$\rankbar(x)=d(x)/\min_{e\in E} w(e)$&\\
  \hline
  \bf LIS&$y>x$ & The length of LIS ending at $x$& 2&$O(n\log^3 n)^{\dagger}$&$O(\rank(S)\log^2 n)^{\dagger}$\\
  \hline
  \bf Activity Selection & \multirow{2}{*}{$x$ and $y$ do not overlap} & The maximum number of non-  &\multirow{2}{*}{2}&\multirow{2}{*}{$O(n\log n)$} &\multirow{2}{*}{$O(\log n)^{\dagger}$}\\
  (unweighted) &&overlapping activities ending at $x$&\\
  \hline
  \multirow{2}{*}{\bf MIS} & $\exists$ a path from $x$ to $y$ s.t.\ the priorities on & The the longest chain size ending at     &\multirow{2}{*}{2} &\multirow{2}{*}{$O(n+m)$} &\multirow{2}{*}{$O(\log^2 n)^{\dagger}$}\\
  & the path are monotonically increasing &vertex $x$ with increasing priorities &&&\\
  \hline
\end{tabular}
\vspace{-.1in}
\caption{\label{tab:overall} \small
\textbf{The problems, their definitions of rank, the work and span of our solutions for the given problems.}
In feasible conditions, we assume $y$ is later than $x$.
$n=|S|$ is the input size. For graphs, $n$ is the number of vertices and $m$ is the number of edges.
$d(v)$ in SSSP is the shortest distance of $v$ from the source and $w(e)$ is the weight of edge $e$. $W$ in the knapsack problem is the weight limit. $\rankbar$ means a relaxed rank (see \cref{sec:sssp}). $\dagger$: with high probability.
\vspace{-1em}
}
\end{table*} 

%% file: datastructure.tex
%\subsubsection*{\bf Parallel Data Structures}
\vspace{.3em}\noindent {\large\emph{\textbf{Parallel Data Structures}}}

%This section introduces the \emph{range-sum query} data structures used in our applications.
%Due to page limitation, we only present useful theorems here, and give more details in Appendix \label{app:tree}.
%The data structures are mostly used for various versions of \emph{range queries}.
We now present useful theorems of data structures for range-sum queries used in this paper.  More details are given in \cref{app:datastructure}.
The data structures maintain a map of entries (\emph{key-values}) sorted by the keys, where keys can be either one or two dimensions.
A range sum query is defined by a key range (an interval in 1D or a rectangle in 2D) and \emph{augmentation},
%The range sum queries are defined on a map of entries (\emph{key-values}) sorted by the keys, and an \emph{augmentation}.
%The keys can be either one or two dimensions (points on 2D planar). The ranges are defined by an interval in 1D or a rectangle on 2D planar.
which defines how the ``sum'', called the \emph{augmented value}, should be computed.
One example is to report the sum (or min/max) of values in the given key range.
Formally, suppose the map maintains key-value pairs of type $K\times V$, we define the augmented value of type $A$ by an augmented structure consisting of two functions and the identity of $A$.

\begin{itemize}[leftmargin=*]
  \item Base function $g:K\times V\mapsto A$, which maps an entry (key-value) to an augmented value.
  \item (Associative) Combine function $f:A\times A\mapsto A$, which combines (adds) two augmented values into a new augmented value.
  \item The identity $I_A\in A$ of $f$ on $A$. $(A,f,I_A)$ is a monoid.
\end{itemize}
In the value-sum example above, the base function $g=(k,v)\mapsto v$, the combine function $f=(a_1,a_2)\mapsto a_1+a_2$, and $I_A=0$.
We first present a useful theorem about range sum queries.

\begin{theorem}
  For $k\in \{1,2\}$, there exist data structures that can answer $k$-D range sum query in $O(\log^k n)$ time,
  can be constructed by (or be flattened into) a sorted sequence of entries in $O(n\log^{k-1} n)$ work and $O(\log^k n)$ span,
  and allow for batch update (e.g., insertion, deletion and value updates) in $O(m\log^k n)$ work and $O(\log^k n)$ span,
  where $n$ is the number of input entries, and $m\le n$ is the batch size.
  We assume constant cost for the base and combine functions.
\end{theorem}
%\zheqi{why not generalize theorem to arbitrary k dims}
%For example, for $n$ 2D points, we can query the maximum value in a rectangle in $O(\log^2 n)$ work,
%and update the values of a batch of $m$ entries in $O(m\log^2 n)$ work and $O(\log^2 n)$ span.

This can be achieved using parallel augmented balanced binary search trees (\pabst{})~\cite{sun2018pam}
with algorithms in \cite{sun2018pam,blelloch2020optimal,blelloch2016just,dhulipala2022pac}.
For 2D range sum queries, we use a 2D range tree using \pabst{s}~\cite{sun2018pam,sun2019parallel}.

%Another type of data structure used in our paper is an index on a key pairs, where the data are ordered first by a \emph{primary key}, and
%then the \emph{secondary key} when ties occur. This can be achieved by using a two level \pabst{} and the following bound holds.
We also use multi-map, where multiple entries can have the same key, and a search on a key will return all values with this key\footnote{Note that this query is different from the 1D or 2D range query defined above, so this uses a different data structure from range trees stated above.}.
%The following theorem holds also by using \pabst{}.
By using \pabst{}, we have the following theorem~\cite{blelloch2016just}.
\begin{theorem}
  For $n$ key-values, there exists a data structure that can search or update (insert or delete) a batch of entries in $O(m\log n)$ work and $O(\log m \log n)$ span, where $m\le n$ is the total number of elements found or updated in the batch.
\end{theorem}
%\zheqi{we do not consider the dims here?}

\hide{
\myparagraph{Parallel Augmented BST (\pabst{}).}
The base data structure we use for range queries is a parallel augmented balanced binary search trees (\pabst{})~\cite{sun2018pam} proposed by Sun et al.
A \pabst{} stores an ordered map of entries of \emph{key-values}, sorted by the keys.
Each tree node maintains an \emph{augmented value} to support fast range sum queries,
which is some aggregate information (an abstract ``sum'') about all entries in its subtree.
%The augmented value can be viewed as an abstract sum of the entries in the subtree based on any associative operations.
%This is maintained to support fast range sum queries.
One example is to store the sum (or min/max) of values in the subtree, such that
reporting the sum of values (or min/max) in a certain key range can be performed in $O(\log n)$ time by combining $O(\log n)$ relevant
tree nodes and subtree augmented values.
%To formally define augmentation, we use the idea proposed in Sun et al.~\cite{sun2018pam}.
Formally, suppose the map maintains key-value pairs of type $K\times V$, we define the augmented value of type $A$ by two functions and an identity of $A$.

\begin{itemize}[leftmargin=*]
  \item Base function $g:K\times V\mapsto A$, which maps an entry (key-value) to an augmented value.
  \item (Associative) Combine function $f:A\times A\mapsto A$, which combines (adds) two augmented values into a new augmented value.
  \item The identity $I_A\in A$ of $f$ on $A$. $(A,f,I_A)$ is a monoid.
\end{itemize}

In the value-sum example above, the based function $g=(k,v)\mapsto v$, the combine function $f=(a_1,a_2)\mapsto a_1+a_2$, and $I_A=0$.

A \pabst{} can report the augmented value of any key range in $O(\log n)$ work.
\pabst{} supports parallel batch operations~\cite{blelloch2016justjoin,sun2018pam}. One can construct a \pabst{} of size $n$ in $O(n \log n)$ work and $O(\log^2 n)$ span, and flatten a \pabst{} into a sorted array in $O(n)$ work and $O(\log n)$ span.
One can also perform batch operations including find (\multifind), insert (\multiinsert), deletion (\multidelete), and update (\multiupdate) in $O(m\log (n/m))$ work and $O(\log n\log m)$ span on a tree of size $n$ and sorted batch of size $m\le n$~\cite{blelloch2016justjoin,sun2018pam}.
Taking union on two \pabst{} of size $n$ and $m\le n$ costs $O(m\log (n/m))$ work and $O(\log n\log m)$ span.
%The algorithms also properly update the augmented values if needed.
%All these algorithms use a simple divide-and-conquer scheme to parallelize the operations, and the details can be found in~\cite{}.

\myparagraph{Parallel Nested BSTs.} In some applications a two-level index is needed for entries with two \emph{keys}. The entries need be ordered first by a \emph{primary key}, and %all entries with the same primary key are ordered by the \emph{secondary key}.
then the \emph{secondary key} when ties occur.
We can first build a \pabst{} keyed on the primary keys. All entries with the same primary key will be organized as another BST (keyed on the secondary key) associated with the corresponding primary key in the outer tree as its value. %Both the inner tree and the outer tree can be further augmented.
The space needed to store $n$ \element{s} in a nested BST $O(n)$.
We can extend the parallel algorithms on BSTs to the nested BSTs and get the same bound as a regular \pabst{} mentioned above.
%. A batch of find (\multifind), insert (\multiinsert), deletion (\multidelete), and update (\multiupdate) can be performed in $O(m\log (n/m))$ work and $O(\log n\log m)$ span on a tree of size $n$ and sorted batch of size $m$~\cite{blelloch2016justjoin,sun2018pam}.

\myparagraph{Parallel Augmented 2D Range Trees.} Some algorithms in this paper needs a \defn{2D (augmented) range query}. For a set of points on a 2D planar, the 2D augmented range query asks for the ``augmented value'' of all points in a given rectangle (e.g., maximum value).
%, where the ``augmented values'' are defined based on base and combine functions similarly to \pabst{}
An augmented range tree is defined on the comparison function for the two dimensions ($<_x$ and $<_y$), respectively, and the \emph{base} and \emph{combine} functions similar to \pabst{}.
This query can be answered by a \defn{2D Range Tree} with an appropriate augmentation.  %using a 2-level \pabst{} (see more details in \cite{sun2018parallel}).
A 2D range tree is a two-level BST where the \emph{outer tree} is an index of the x-coordinates of the points.
Each tree node maintains an \emph{inner tree} storing the same set of points in its subtree, but keyed on the y-coordinates.
%This essentially can be understood as a \pabst{} where the inner trees are the augmented values, and thus the combine function is a parallel union on the two inner trees~\cite{}.
Each inner tree can be further augmented based on the application.
Then the augmented value of any rectangle can be reported in $O(\log^2 n)$ time, for a dataset of $n$ points.
For example, if the inner tree is augmented with the maximum value of each subtree,
then the range tree can answer any rectangle-max query (report the maximum value for all points in a given rectangle) in $O(\log^2 n)$ time.

%Then the augmented value of any rectangle can be reported in $O(\log^2 n)$ time, for a dataset of $n$ points.

Note that the 2D range tree is \emph{not} a 2-level nested tree as defined above. The space usage of a range tree on $n$ points is $O(n\log n)$.
%In particular, a 2D Range Tree
}

\hide{\subsection*{Parallel Data Structures}
This section introduces the data structures used in our applications, mostly for \emph{range queries}.
%The data structures are mostly used for various versions of \emph{range queries}.

\myparagraph{Parallel Augmented BST (\pabst{}).}
The base data structure we use for range queries is a parallel augmented balanced binary search trees (\pabst{})~\cite{sun2018pam} proposed by Sun et al.
A \pabst{} stores an ordered map of entries of \emph{key-values}, sorted by the keys.
Each tree node maintains an \emph{augmented value} to support fast range sum queries,
which is some aggregate information (an abstract ``sum'') about all entries in its subtree.
%The augmented value can be viewed as an abstract sum of the entries in the subtree based on any associative operations.
%This is maintained to support fast range sum queries.
One example is to store the sum (or min/max) of values in the subtree, such that
reporting the sum of values (or min/max) in a certain key range can be performed in $O(\log n)$ time by combining $O(\log n)$ relevant
tree nodes and subtree augmented values.
%To formally define augmentation, we use the idea proposed in Sun et al.~\cite{sun2018pam}.
Formally, suppose the map maintains key-value pairs of type $K\times V$, we define the augmented value of type $A$ by two functions and an identity of $A$.

\begin{itemize}[leftmargin=*]
  \item Base function $g:K\times V\mapsto A$, which maps an entry (key-value) to an augmented value.
  \item (Associative) Combine function $f:A\times A\mapsto A$, which combines (adds) two augmented values into a new augmented value.
  \item The identity $I_A\in A$ of $f$ on $A$. $(A,f,I_A)$ is a monoid.
\end{itemize}

In the value-sum example above, the based function $g=(k,v)\mapsto v$, the combine function $f=(a_1,a_2)\mapsto a_1+a_2$, and $I_A=0$.

A \pabst{} can report the augmented value of any key range in $O(\log n)$ work.
\pabst{} supports parallel batch operations~\cite{blelloch2016justjoin,sun2018pam}. One can construct a \pabst{} of size $n$ in $O(n \log n)$ work and $O(\log^2 n)$ span, and flatten a \pabst{} into a sorted array in $O(n)$ work and $O(\log n)$ span.
One can also perform batch operations including find (\multifind), insert (\multiinsert), deletion (\multidelete), and update (\multiupdate) in $O(m\log (n/m))$ work and $O(\log n\log m)$ span on a tree of size $n$ and sorted batch of size $m\le n$~\cite{blelloch2016justjoin,sun2018pam}.
Taking union on two \pabst{} of size $n$ and $m\le n$ costs $O(m\log (n/m))$ work and $O(\log n\log m)$ span.
%The algorithms also properly update the augmented values if needed.
%All these algorithms use a simple divide-and-conquer scheme to parallelize the operations, and the details can be found in~\cite{}.

\myparagraph{Parallel Nested BSTs.} In some applications a two-level index is needed for entries with two \emph{keys}. The entries need be ordered first by a \emph{primary key}, and %all entries with the same primary key are ordered by the \emph{secondary key}.
then the \emph{secondary key} when ties occur.
We can first build a \pabst{} keyed on the primary keys. All entries with the same primary key will be organized as another BST (keyed on the secondary key) associated with the corresponding primary key in the outer tree as its value. %Both the inner tree and the outer tree can be further augmented.
The space needed to store $n$ \element{s} in a nested BST $O(n)$.
We can extend the parallel algorithms on BSTs to the nested BSTs and get the same bound as a regular \pabst{} mentioned above.
%. A batch of find (\multifind), insert (\multiinsert), deletion (\multidelete), and update (\multiupdate) can be performed in $O(m\log (n/m))$ work and $O(\log n\log m)$ span on a tree of size $n$ and sorted batch of size $m$~\cite{blelloch2016justjoin,sun2018pam}.

\myparagraph{Parallel Augmented 2D Range Trees.} Some algorithms in this paper needs a \defn{2D (augmented) range query}. For a set of points on a 2D planar, the 2D augmented range query asks for the ``augmented value'' of all points in a given rectangle (e.g., maximum value).
%, where the ``augmented values'' are defined based on base and combine functions similarly to \pabst{}
An augmented range tree is defined on the comparison function for the two dimensions ($<_x$ and $<_y$), respectively, and the \emph{base} and \emph{combine} functions similar to \pabst{}.
This query can be answered by a \defn{2D Range Tree} with an appropriate augmentation.  %using a 2-level \pabst{} (see more details in \cite{sun2018parallel}).
A 2D range tree is a two-level BST where the \emph{outer tree} is an index of the x-coordinates of the points.
Each tree node maintains an \emph{inner tree} storing the same set of points in its subtree, but keyed on the y-coordinates.
%This essentially can be understood as a \pabst{} where the inner trees are the augmented values, and thus the combine function is a parallel union on the two inner trees~\cite{}.
Each inner tree can be further augmented based on the application.
Then the augmented value of any rectangle can be reported in $O(\log^2 n)$ time, for a dataset of $n$ points.
For example, if the inner tree is augmented with the maximum value of each subtree,
then the range tree can answer any rectangle-max query (report the maximum value for all points in a given rectangle) in $O(\log^2 n)$ time.

%Then the augmented value of any rectangle can be reported in $O(\log^2 n)$ time, for a dataset of $n$ points.

Note that the 2D range tree is \emph{not} a 2-level nested tree as defined above. The space usage of a range tree on $n$ points is $O(n\log n)$.
%In particular, a 2D Range Tree
} 

%% file: prob.tex
\section{\titlecap{\phaseparallel{}} Algorithms}\label{sec:structure}

%Given an independence system $(S,\setfam)$,
%A set system is a finite set $S$ and a family of subsets $\setfam$ over $S$, called \feasible{} sets.
In this section, we introduce our key concept: \defn{\phaseparallel{} algorithms},
and show a general approach to design \phaseparallel{} algorithms
to maximize parallelism based on the \defn{rank} function.
Since our idea is sophisticated, we first show the pseudocode in \cref{alg:framework} and describe the high-level idea.

To seek parallelism in many sequential iterative algorithms, we define the $\rank(\cdot)$
of each object to capture the dependences among them, which indicates the earliest phase an object can be ready.
With a properly defined rank function,
Algorithm \ref{alg:framework}
%processes the objects based on rank ordering, and
processes all objects (in parallel) of rank $i$ in round $i$.
We call the set of objects processed in round $i$ ($\curset_i$ in Algorithm \ref{alg:framework}) the \defn{\frontier{}} of round $i$.
Table \ref{tab:overall} shows how rank is defined for the problems in this paper.
For example, in the LIS problem, an object's rank is the size of the LIS ending at this object.
Therefore, \cref{alg:framework} finds and processes all objects with LIS size 1 in parallel, then those with LIS size 2, etc.
Throughout the section, we use the LIS problem as an example to help understand the abstract concepts.
%We note that although the generalized formalization is sophisticated,
%the algorithms for each individual problem are reasonably simple.
Next, we formalize the \phaseparallel{} algorithms.
We note that all of our algorithms are reasonably simple.
The goal of the formalization is to
extend our idea to general independence systems,
which generalizes to more DP and greedy algorithms.

\begin{algorithm}
\caption{The \phaseparallel{} algorithm\label{alg:framework}}\small
\KwIn{$S$, and $\rank(x)$ that implies $\setfam$}
\DontPrintSemicolon
  $i\gets 1$\\
  \While{$S\ne \emptyset$\label{step:frameworkwhile}} {
    Find the set $\curset_i$ that contains all objects with rank $i$\label{line:find-frontier}\\
    Process all objects in $\curset_i$ in parallel\\
    $S\gets S\setminus\curset_i$\\
    Update the status of objects in $S$ if necessary\\
    $i\gets i+1$
  }
\end{algorithm}

An independence system is a pair $(S, \setfam)$, where $S$ is a finite set and $\setfam$ is a collection of subsets of $S$ (called the \defn{independent sets} or \defn{feasible sets})  with the following properties:

\begin{enumerate}
  \item The empty set is feasible, i.e., $\emptyset \in \setfam$.
  \item (Hereditary property) A subset of a feasible set is feasible, i.e., for each $Y\subseteq X$, we have $X\in \setfam \implies Y\in \setfam$.
\end{enumerate}
A feasible set for the LIS problem is any increasing subsequence.

Given an independence system $(S,\setfam)$,
a \defn{sequential order} of it is a permutation of all objects in $S$, usually specified by the input.
For an object $x\in S$, let $\bmath{\ind_S(x)}$ be the index of $x$ w.r.t. its sequential order.
We say an object $x$ is \defn{earlier} than $y$ if $\ind_S(x)<\ind_S(y)$, and \defn{later} otherwise.
Let $\bmath{x^{\downarrow S}}=\{y\in S: \ind_S(y)\le \ind_S(x)\}$ be the downward closure of $x$,
i.e., all objects no later than $x$.
With clear context, we drop the superscripts and use $\bmath{\ind(x)}$ and $\bmath{x^{\downarrow}}$. We use $S_i$ as the object in $S$ with index $i$.
In LIS, the index $\ind(x)$ of an object $x$ is its position in the input sequence $S$, and $x^{\downarrow}$ is the prefix of $S$ up to $x$.
%For LIS, this order is from the first object to the last, but it can be more complicated for other problems.

We say two objects $x$ and $y$ are \defn{\incompatible{}} if $\nexists E\in \setfam$, s.t. $x\in E$ and $y\in E$, and \defn{\compatible} otherwise.
We say an object $x$ is compatible with a set $E\subseteq S$ if $E\cup \{x\} \in \setfam$.
Mapping this to LIS, two objects $x$ and $y$ (later than $x$) are \compatible{} iff $x<y$.

Given an object $x$, we use $\setfam(x)=\{E\in \setfam: E\subseteq x^{\downarrow}, x\in E\}$ to denote all feasible sets \emph{with the last object as $x$},
and the \defn{Maximum Feasible Set (MFS)}\footnote{This is also known as the maximum independent set (MIS). In this paper, to avoid confusion with the greedy MIS algorithm in \cref{sec:mis}, we use the term MFS.} $\mis(x)=\arg\max_{E\in \setfam(x)}|E|$ as the largest set among $\setfam(x)$.
For many DP problems, MFS is usually related to the \emph{DP value} of the object.
For example, in LIS, the $\mis(x)$ is the LIS ending at $x\in S$.
For a set $S$, we also define the MFS as the largest feasible subset of $S$.
%\zheqi{is it saying that $\mfs(x)=\mfs(x^{\downarrow})?$}
We define the \defn{rank} of a set or an object as $\rank(\cdot)=|\mis(\cdot)|$.
In LIS, $\setfam(x)$ refers to all increasing subsequences ending at object $x$.
The \MFS{} of an input sequence is the LIS of the sequence.
$\rank(x)$ is the size of LIS ending with $x$.

%an algorithm $\mathcal{A}$ is a sequential iterative algorithm on $S$ if $\mathcal{A}$ processes all objects in $S$ with the goal of
%optimizing an objective $\obj(E):E\in \setfam$, or finding all possible feasible sets in $\setfam$.
Given an independence system $(S,\setfam)$, a \defn{sequential iterative algorithm} $\mathcal{A}$ on $S$ processes each object $S_i$ in $S$ iteratively based
on the sequential order, with the goal to optimize some value of all (or some) feasible sets.
%Usually processing $S_i$ involves solving a subproblem on $S_i^{\downarrow}$,
%(possibly) by enumerating all $S'\in \setfam(S_i)$.
Since a processed object usually corresponds to a subproblem on $S_i^{\downarrow}$,
they are sometimes called the \defn{states} in dynamic programming problems.
For example, in LIS, processing object $S_i$ is to compute the LIS up to (and including) object $S_i$.
%In other words,
%processing $x$ usually involves finding all feasible sets containing $x$ and only objects before $x$.

To parallelize a sequential iterative algorithm,
note that an object does not need to wait for \emph{all} earlier objects to finish,
but only a subset of them.
Let $\pred(x)$ be all objects that $x$ rely on, i.e., all predecessors of $x$ in the \DG{}.
For LIS, $\pred(x)=\{y: \ind(y)<\ind(x), y<x\}$.
When all objects in $\pred(x)$ finish, $x$ is \ready{}.
Also, if two objects do not rely on each other in the \DG{}, they can be processed in parallel.
These two simple observations have been used in existing parallel algorithms and frameworks (e.g.,~\cite{BFS12,hasenplaugh2014ordering,shun2015sequential,blelloch2016parallelism,blelloch2018geometry,blelloch2018geometry,blelloch2020randomized}).
In this paper, we formalize the problem for a class of algorithms based on an independence system and point out that identifying the ready objects can be captured by the \emph{ranks} of the objects. We define \emph{\phaseparallel{}} as follows.

%We say two objects $x$ and $y$ are \defn{compatible} if $\{x,y\}\in \setfam$,

\begin{definition}\label{def:pp-algo}
Given an independence system $(S,\setfam)$,
a sequential iterative algorithm on $S$ is \defn{{\phaseparallel}} if it has the following property:
an object $x\in S$ relies on $y\in S$ in the parallel \dg{} if and only if:
\begin{enumerate}[label=(\arabic*).]
  \item (Ordering) $\ind(y)<\ind(x)$.
  \item (Compatibility)
  $\forall E\subseteq \setfam(y)$, we have $E\cup \{x\} \in \setfam$, i.e., any feasible set containing $y$ and only objects up to $y$ are also compatible with $x$.
  %$\forall E\subseteq y^{\downarrow}$ where $y\in E, E\in \setfam$, we have $E\cup \{x\} \in \setfam$. In other words, any feasible set with maximum index of $\ind(y)$ is also compatible with $x$.
\end{enumerate}
\end{definition}
This means computing the state (processing an object) $x$ only relies on previous states in $x^{\downarrow}$ compatible with $x$.
This indicates \emph{optimal substructure} property~\cite{clrs}, where the best solution at $x$ can be obtained by optimal solutions before $x$.

%We define the \defn{Maximum Independent Set (MIS) of a set $E\subseteq S$}, noted as $\mis(E)$, as the largest subset in $E$ that is feasible (independent), i.e., $\mis(E)=\arg\max_{E\in \setfam}|E|$.

To achieve maximum parallelism, our goal is to find the largest possible set of objects to process in parallel.
We first show that all objects with the same rank (MFS size) can be processed in parallel.

\begin{theorem}\label{thm:phase-parallel}
  Given a \phaseparallel{} algorithm $\mathcal{A}$ on the independence system $(S,\setfam)$, if $\rank(x)=\rank(y)$, then $x$ and $y$ cannot rely on each other in the parallel \dg{}.
\end{theorem}
\begin{proof}
  Assume to the contrary that $y$ relies on $x$ (the other case is symmetric).
  Consider the MFS of $y$. By definition $\mis(y)\cup \{x\}$ is also feasible, which means $\rank(x)\ge \rank(y)+1$.
\end{proof}

\hide{
\begin{corollary}\yan{check}
If two objects are \incompatible{}, they can be processed in parallel.
\end{corollary}
}

\cref{thm:phase-parallel} leads to the following conclusion, based on which we propose \cref{alg:framework}. %\zheqi{on algo which we proposed}.

\begin{corollary}
In a \dg{}, if $x$ relies on $y$, $\rank(x)>\rank(y)$. All objects with the same rank can be processed in parallel. 
\end{corollary}

\begin{theorem}
\label{thm:rankdepth}
  The rank of an object in a \phaseparallel{} algorithm is its depth in the \DG{}.
\end{theorem}
\begin{proof}
  (Sketch) From the compatibility property of the \phaseparallel{} algorithm, we know the rank (MFS size) of an object must be 1 plus the maximum MFS size of its predecessors. By induction, we can prove the given theorem.
\end{proof}

Theorem \ref{thm:rankdepth} verifies the strategy of Algorithm \ref{alg:framework}, which means we are just processing the objects based on their depth in the \DG{}.

%In this algorithmic framework, we will use introduced in \cref{sec:prelim} to maintain the set $S$, so we can get an instantiation (the algorithm for a specific application) once we  supplement \cref{line:find-frontier}---how to find the frontier to process in each round.
In this paper, we propose novel and efficient ways to process \phaseparallel{} algorithms.
The challenges lie in achieving work-efficiency with non-trivial parallelism.
As mentioned, although the high-level idea of processing all ready objects (thus achieving round-efficiency)
in a round has been used in existing work,
most of them need to check all edges in the \DG{}.
In many cases, the number of edges can be asymptotically more than efficient work.
%This harms work-efficiency since in many cases we cannot afford checking all edges (or even vertices) in the \DG{}, let alone checking each of them multiple times.
%Instead of checking all possible dependences in the \DG{},
%we show a much broader solution framework that can apply to much more applications.
%In this paper, % we can alternatively analyze the rank function (or an estimation of the rank) to design better algorithms.
%The specific solution is mainly decided by how the rank function is defined and calculated, and based on them, we categorize them into two types, and are introduced in \cref{sec:type1} and \cref{sec:type2} respectively.
We propose two general ideas to reduce work in \phaseparallel{} algorithms.
Type 1 algorithms (in \cref{sec:type1}) find the frontier in each round efficiently using a \emph{range query} in polylogarithmic cost.
Type 2 algorithms (in \cref{sec:type2}) try to \emph{wake up} all ready objects by the finished ones at the right time.
Both cases avoid checking all edges in the \DG{}.
%We note that the previous approach of checking downward closures can be considered as a special case of Type 2 algorithms, and we use MIS in \cref{sec:mis} as an example to show algorithms based on the new framework can achieve better span bounds for several problems.

\hide{
Algorithm outline:

\medskip
%\rule[0cm]{8cm}{0.4pt}

\begin{enumerate}[label=\arabic*.]
\hrule
  \item $i\gets 1$;
  \item \label{step:frameworkwhile} While $S\ne \emptyset$:
  \begin{enumerate}[itemindent=.25cm,label=\ref{step:frameworkwhile}\arabic*.]
    \item Find all objects with rank $i$, noted as $\curset_i$;
    \item Process all objects in $\curset_i$;
    \item $S\gets S-\curset_i$;
    \item Update the status of objects in $S$ if necessary;
    \item $i\gets i+1$;
  \end{enumerate}
\hrule
\end{enumerate}
%\rule[0cm]{8cm}{0.4pt}

\medskip

Type 1: When each frontier can be identified by a range query.

%Type 2: $\forall x\in S$, the set of objects $x$ relies on can be identified by a range query.
Type 2: after an object has been processed, it can wake up possible candidates of the next frontier.

}

\hide{
\begin{definition}
Given an set system $(S,\setfam)$,
a sequential iterative algorithm on $S$ is \defn{\titlecap{\phaseparallel}} with the following property:
\begin{enumerate}[label=(\arabic*).]
  \item \label{prop:empty}$\emptyset\in \setfam$.
  \item \label{prop:hereditary}(Hereditary property) $\forall Y\subseteq X$, we have $X\in \setfam \implies Y\in \setfam$. \ref{prop:empty} and \ref{prop:hereditary} imply that $(S,\setfam)$ is an \defn{independence system}.
  \item An object $x\in S$ relies on (dependent on) $y\in S$ in the parallel \dg{} iff.
  $\forall E\subseteq y^{\downarrow}$ where $y\in E, E\in \setfam$, we have $E$
\end{enumerate}
\end{definition}

 with the objective
of $\min_{E\in \setfam}\obj(E)$ (or $\max$) (e.g., for our LIS or activity selection algorithm), or finding certain feasible sets in $\setfam$ (e.g., SSSP or Huffman Code).

}

%% file: type1.tex
\section{Type 1 Algorithms using Efficient Frontier Identifying}\label{sec:type1}

%In Type 1 algorithms, the special property of the rank function is a piecewise monotonic constant function.
%It takes some property of an object $x$, say $x.k$, which is a real number.
%Then $\rank(x)$ maps it to an integer, and satisfies $\rank(x)\le \rank(y)$ if $x.k<y.k$.
Type 1 algorithms exhibit the special property
where each object maintains a value, and all objects with the same rank have their values in a contiguous range.
In this case, we can use \pabst{} to maintain these values and use a range search to efficiently find the frontier. %($T_i$ on \cref{line:find-frontier} in \cref{alg:framework}).
%We show three Type I algorithms: Huffman tree, single-source shortest paths (SSSP), and activity selection.
\ifconference{
We show two Type 1 algorithms in the paper (activity selection, and Dijkstra's algorithm) and more in the Appendix (unlimited knapsack in Appendix \ref{app:knapsack}, and Huffman tree in Appendix \ref{app:huffman}).
}
\else{
We show four Type 1 algorithms in the paper (activity selection, unlimited knapsack, Dijkstra's algorithm, and Huffman tree).
}
\fi
%We note that similar versions of some of these algorithms may be known.
Many of them are straightforward.
We do not claim all of them as the main contributions, but use them as simple examples to understand our framework.

\subsection{Activity Selection}

%\begin{figure}
%  \centering
%  \includegraphics[width=\columnwidth]{figures/actsel.pdf}
%  \caption{\textbf{Activity selection.}}\label{fig:activity}
%\end{figure}
\emph{Activity selection} is a textbook example of greedy or dynamic programming algorithms~\cite{clrs}.
Given a set of activities $S=\{A_i\}$ defined by their start time $s_i$, end time $e_i$ and weight $w_i$, the \defn{activity selection} problem is to find a \feasible{} subset of non-overlapping activities to maximize the total weight.
When all activities have a unit weight, a simple \emph{earliest-end} greedy strategy can solve the problem~\cite{clrs},
i.e., repeatedly selecting the earliest ending activity and removing all \incompatible{} (overlapping) activities.
The general version (arbitrary weight) can be solved by the dynamic programming recurrence:

%$$D[i]=\max_{\substack{j<i\\ \timeend{}_j\le \timestart{}_i}}D[j]+w_i$$
\vspace{-.05in}
\begin{equation}\label{eqn:activitydp}
  \mathdp[i]=\max_{\timeend_j\le \timestart_i}\mathdp[j]+w_i
\end{equation}
\vspace{-.03in}

The sequential order (and the \emph{index}) is defined by the end time.
We assume all activities are pre-sorted by their end time.
A \feasible{} set is a set of non-overlapping activities.
$\mathdp[i]$, or the \defn{\DPvalue{}} of activity $i$,
means the highest possible weight
by using the first $i$ activities, which must include $A_i$.
%$w_i$ is the weight gained by selecting activity $i$.
Naively computing \cref{eqn:activitydp} needs $O(n^2)$ work.
Since the condition in \cref{eqn:activitydp} is a range of end time,
sequentially, the work can be reduced to $O(n\log n)$ using augmented range queries.
In parallel, an activity $A_i$ depends on all activities ending before $A_i$ starts (see an illustration in \cref{fig:activity}), which leads to $O(n^2)$ dependences in the worst case.
Note that the rank of $A_i$ by definition is the maximum size of the \feasible{} set containing $A_i$ and only the first $i$ activities.
%maximum size of \feasible{} set containing activity $i$ and only using the first $i$ activities. We can easily verify that this problem is \phaseparallel{}.
\ifconference{
We first present the following lemma and prove it in Appendix~\ref{app:activityproof}
}
\else{
We first present the following lemma.
}
\fi

\begin{lemma}
\label{lem:activity}
All activities overlapping the earliest-end activity $A_1$ has rank 1.
After removing all activities with rank no more than $k$, suppose the earliest-end activity is $A_j$, then all remaining activities that overlaps with $A_j$ has rank $k+1$.
\end{lemma}
\iffullversion{
\input{act-proof.tex}
}
\fi

\begin{algorithm}
\SetKwFor{parForEach}{parallel\_for\_each}{do}{endfor}
\caption{Type-1 activity selection algorithm \label{alg:activity1}}\small
\KwIn{All activities' start time $s_i$, end time $e_i$ and weight $w_i$}
\DontPrintSemicolon
  Build PA-BSTs $T_{\mathit{time}}$ on key-values $(s_i,e_i)$, augmented on the minimum end time, and
  $T_{\mathit{DP}}$ on key-values $(e_i,dp[i])$, augmented on the maximum DP value\\
  \While{$T_{\mathit{time}}\ne \emptyset$\label{step:frameworkwhile}} {
    Find the earliest-end activity $x$ by the augmented value of $T_{\mathit{time}}$\\
    $\langle T,T'\rangle \gets \mathit{split}(T_{\mathit{time}},e_x)$\tcp{All activities starting before $e_x$ form the current frontier T}
    \parForEach{activity $i\in T$}{
      $dp[i]=w_i+T_{\mathit{DP}}.\mathit{range}(-\infty, s_i)$
    }
    Update all DP values of activities in $T$ in $T_{\mathit{DP}}$ in parallel\\
    $T_{\mathit{time}}\gets T'$\tcp{remove finished objects}
  }
\end{algorithm}

Based on the lemma and the \phaseparallel{} framework, we can design an algorithm (Algorithm \ref{alg:activity1}).
To enable work-efficiency, we use a range query to find the largest parallelizable \frontier{}. The algorithm uses two \pabst{s} $T_{\mathit{time}}$ and $T_{\mathit{DP}}$.
$T_{\mathit{time}}$ maintains all unprocessed activities sorted by their start time and augmented on the minimum end time, which is used to identify the \frontier{s}.
%$T_{\mathit{DP}}$ maintains all activities with their \DPvalue{s} and is augmented on the largest \dpvalue{}, which is used to determine $\max_{e_j\le s_i}D[j]$ in the DP recurrence.
$T_{\mathit{DP}}$ maintains all activities sorted by their end time and augmented on the largest DP value, which is used to determine $\max_{e_j\le s_i}dp[j]$ in the DP recurrence.
In each round, we find the earliest-end activity $x$ by reading the augmented value of $T_{\mathit{time}}$. Then we split $T_{\mathit{time}}$ based on $e_x$. Those starting no later than $e_x$ will be split out as the frontier and will be processed in parallel.
Since $T_{\mathit{time}}$ is indexed on start time, \splitn{} takes $O(\log n)$ work.
When processing activity $i$, we use $T_{\mathit{DP}}$ to extract the highest \dpvalue{} among all activities with end time in range $(-\infty,s_i]$,
and use it to update the \dpvalue{} of $i$ in $T_{\mathit{DP}}$.
The work for processing $m$ objects in the frontier is $O(m\log n)$ for the augmented range query, and $O(m\log n)$ for updating the \dpvalue{s} in $T_{\mathit{DP}}$. This leads to the following theorem.

\begin{theorem}\label{them:activity_1}
  Type 1 activity selection algorithm takes $O(n\log n)$ work and $O(\rank(S)\log n)$ span, where $S$ is the input set and $n=|S|$.
\end{theorem}

\iffullversion{
\subsection{Unlimited Knapsack}
\input{knapsack.tex}
}\fi

\subsection{Algorithms with Relaxed Rank}
\label{sec:sssp}
In some problems, it is hard to find (or use) the exact rank of the objects, in which case we use a \relaxedrank{}, defined as follows.

\begin{definition}
Given a \phaseparallel{} algorithm $\mathcal{A}$ on the independence system $(S,\setfam)$, a function $\rankbar(x)$ on $x\in S$ is a \defn{\relaxedrank{}} on object $x$ if
\begin{itemize}
  \item $\forall x\in S, \rank(x) \le \rankbar(x)$.
  \item For $x,y\in S$ where $x$ relies on $y$ in the \dg{}, $\rank(x)>\rank(y)$.
\end{itemize}
\end{definition}

Then Algorithm \ref{alg:framework} can process all objects with the same \relaxedrank{} in each round.
Note that the trivial \relaxedrank{} is the index of each object $\ind(x)$, which gives no parallelism in Algorithm \ref{alg:framework}.
Therefore, when we use the \relaxedrank{}, we need careful analysis to show non-trivial parallelism.
%In this section, we discuss some problems that can be solved by a similar framework as Type 1, but using \relaxedrank{} functions.

\myparagraph{Dijkstra's Algorithm.}
%We can parallelize Dijkstra's algorithm in the similar way to Huffman tree construction.
Dijkstra's algorithm~\cite{dijkstra1959} solves the single-source shortest paths (SSSP) problem on a weighted graph.
As a sequential iterative algorithm, Dijkstra processes the vertices in the order of their distances to the source and relaxes their neighbors.
%starts by assigning each vertex a \emph{tentative distance} $\delta[v]=\infty$ except for the source $s$ which has $\delta[s]=0$. The algorithm repeatedly finds the smallest tentative distance among unsettled vertices, marks it as settled, and relaxes its neighbors.
\ifconference{
SSSP is very challenging in the parallel setting.
Dijkstra is work-efficient and relaxes each edge exactly once. However, it is hard to parallelize because each round only processes one vertex.
%One could get better parallelism using Bellman-Ford, but it has significantly more work.
Bellman-Ford has better parallelism but significantly more work.
%Even so, parallel Bellman-Ford and its variants are widely used in practice because of better parallelism. There are other algorithms, such as $\Delta$-stepping~\cite{deltastepping} and $\rho$-stepping~\cite{rhostepping} that use heuristics to achieve tradeoff between extra work and parallelism.
Almost all state-of-the-art parallel SSSP algorithms (e.g., $\Delta$-stepping~\cite{meyer2003delta} and $\rho$-stepping~\cite{rhostepping}) achieve high parallelism by using more work. %We review some background of parallel SSSP algorithms in Appendix \ref{app:sssp}.
}
\else{
\input{sssp-bg.tex}
}
\fi

%As a sequential iterative algorithm, Dijkstra processes the vertices in $V$ in the order of their actual distance to the source.
The \DG{} of the SSSP problem is conceptually the shortest path tree.
%Two elements are \compatible{} when one is on the shortest path to the other.
The rank of a vertex $v$ is the hop distance from $v$ to $s$ in the shortest path tree.
However, the algorithm itself is unaware of the explicit structure of \DG{} before the shortest paths are computed.
Hence, the exact rank of each object is hard to acquire.
%To overcome this, we define $\rankbar(\cdot)$ as an overestimation of rank, which guarantees $\rankbar(y)<\rankbar(x)$ when $x$ relies on $y$.
%In this case, we can process the elements based on the values of $\rankbar(\cdot)$, and achieve reasonable parallelism.
Let $w^*$ be the smallest edge weight in the graph, and $d(v)$ the actual distance of $v$.
We define a relaxed rank for a vertex $\rankbar(v)=\lceil d(v)/w^*\rceil$. 
This is because %when the closest tentative distance is $d^*$, not only this vertex, but all vertices with $\delta[\cdot]=d^*+w^*$ should have been settled. This is because any updated distance from then on must have (actual) distance more than $d^*+w^*$ by relaxation.
distances within a window of $w^*$ cannot rely on each other (relaxation increases the distance by at least $w^*$).
Therefore, each frontier can be extracted using a range query.
Interestingly, we observe that this is (conceptually) similar to using $\Delta=w^*$ in $\Delta$-stepping~\cite{meyer2003delta}.
Using \pabst{} to maintain the distances of all vertices, we have the following result.%this gives an algorithm of $O(|E|\log |V|)$ total work and $O(\log|V|\max_{v\in V} d(v)/w^*)$ span.

\begin{theorem}
  There exists a parallel algorithm that solves SSSP problem on a graph $G=(V,E)$ using $O(|E|\log |V|)$ work and \newline $O(\rankbar(V)\log |V|)$ span, where $\rankbar(V)$ is the ratio of the maximum shortest path in the graph and the smallest edge weight.
\end{theorem}

We note that there can be other ways to define the relaxed rank of the Dijkstra's algorithm~\cite{crauser1998parallelization,kainer2019more}, which enable different bounds to the \phaseparallel{} algorithms. In the paper we simply discuss the version based on the smallest edge weight, since it can be easily tested using a $\Delta$-stepping-based implementation.

%In our experiments in \cref{sec:exp}, we use the $\Delta$-stepping implementation in~\cite{rhostepping} with $\Delta=w^*$ to test the performance of this idea. On low-diameter graphs with reasonably large $w^*$, setting $\Delta=w^*$ gives the best performance among all choices of parameter $\Delta$ because of efficiency in work.
In our experiments in \cref{sec:exp}, we use the $\Delta$-stepping implementation in~\cite{rhostepping} with $\Delta=w^*$ to test the performance of this idea. On low-diameter graphs with reasonably large $w^*$, setting $\Delta=w^*$ gives the best performance among all choices of parameter $\Delta$ because of the work-efficiency.

\input{act-lis-fig.tex}
\ifconference{
\myparagraph{Huffman Tree.} Due to page limitation, we discuss the problem in the full version of this paper~\cite{iterativefull}, and present the result here.
}
\fi
%%We then apply our idea to constructing Huffman tree.
%\begin{theorem}
%\label{thm:huffman}
%  There exist a parallel algorithm that constructs a Huffman tree for $n$ input objects in $O(n\log n)$ work and $O(H\log n)$ span, where $H$ is the Huffman tree height.
%\end{theorem}
%

\def\huffmanthm{
\begin{theorem}
	\label{thm:huffman}
	There exist a parallel algorithm that constructs a Huffman tree for $n$ input objects in $O(n\log n)$ work and $O(H\log n)$ span, where $H$ is the Huffman tree height.
\end{theorem}
}

\ifconference{
\huffmanthm
}
\else{
\myparagraph{Huffman Tree}\label{sec:huffman}
\input{huffmantree}
}
\fi 

%% file: act-proof.tex
\begin{proof}
 We first prove all activities overlapping the earliest-end activity $A_1$ must have rank 1.
 %By definition of rank, we just need to prove all such activities have maximum feasible set up to thi
 Recall the rank of an activity $x$ is the maximum number of compatible activities we can select from $x$ and earlier activities, which must include $x$.
 For an activity overlapping $A_1$, if its rank is larger than 1, there must exist another activity $y$ that is earlier than $x$ and compatible with $x$.
 This means $y$ ends before $x$ starts, which contradicts the assumption that the earliest-end activity overlaps with $x$.

 We next prove that after removing all activities of rank no more than $k$, if the current earliest-end activity is $A_j$, all activities overlapping $A_j$ should have rank $k+1$.
 Let $S'$ be the set of activities removed, which have rank no more than $k$.
 Assume to the contrary that one of such activity $A_{i}$ has rank $r>k+1$. Then there must be an activity $A_{i'}$ finishing earlier than $A_i$, and have rank $r'>k$.
 Therefore, $A_{i'}\notin S'$ since we only remove activities with rank no more than $k$.
 However, this contradicts the assumption that $A_j$ is the earliest-end activity in $S\backslash S'$, and $A_j$ already overlaps $A_i$.
\end{proof} 

%% file: knapsack.tex
Given a weight limit $W$ and a set of items, where item $i$ has integer weight $w_i$ and value $v_i$, the unlimited knapsack problem is to select some items (each can be selected multiple times) with total weight within $W$, and to maximize the total value of them. Sequentially, this can be solved by processing each possible weight $j= 0,1,\cdots W$, and computing the optimal solution for the corresponding subproblem. Let $\mathdp[j]$, called the \dpvalue{}, be the maximum value we can achieve using weight limit $j$. The DP recurrence is
\vspace{-.03in}
\begin{align}
  \mathdp[j]=\max\{0, \max_{w_i\le j} \mathdp[j-w_i]+v_i\}
\end{align}

\vspace{-.02in}
%Clearly, a state $j$ relies on all states $j-w_i$ for all $w_i\le j$.
Two states $x$ and $y>x$ are \compatible{} if forming a solution of weight $y$ can contain the solution of the subproblem of weight $x$.
We can easily verify that the problem is \phaseparallel{}. A state $j$ relies on all states $j-w_i$, for all $i$ with $w_i\le j$. The rank of a state $x$
is the maximum number of items one can select with weight limit $x$.
Let $w^{*}=\min w_i$. The rank of a state $x$ is $\rank(x)=\lfloor x/w^*\rfloor$ because to get the maximum number of items, we should always choose the lightest item.
Applying the \phaseparallel{} framework, the frontier of round $i$ should be all states in range $[(i-1)w^*,iw^*)$.%, which can also be obtained by a range query.

\begin{theorem}
  There exists a parallel algorithm that can solve the unlimited knapsack problem in $O(nW)$ work and $O(\rank(W)\log n)$ span, where $n$ is the number of items, $W$ is the weight limit, $\rank(W)=W/w^*$ and $w^*$ is the minimum weight of the items.
\end{theorem}

%% file: sssp-bg.tex
SSSP is a big open problem in the parallel setting.
Dijkstra is work-efficient in that each edge is relaxed exactly once. However it is hard to parallelize because in each round, only the closest unsettled vertex can be processed.
One could easily get better parallelism using Bellman-Ford, where in each round, all vertices relax their neighbors in parallel, but it has significantly more work.
Even so, parallel Bellman-Ford and its variants are widely used in practice because of better parallelism. There are other algorithms, such as $\Delta$-stepping~\cite{meyer2003delta} and $\rho$-stepping~\cite{rhostepping} that use heuristics to achieve tradeoff between extra work and parallelism.
For example, the $\Delta$-stepping algorithm is a hybrid of Dijkstra and Bellman-Ford.
It determines the correct shortest distances in increments of $\Delta$.
In step $i$, the algorithm will find and settle down all the vertices with distances in $[i\Delta, (i+1)\Delta]$.
Within each step, the algorithm runs Bellman-Ford as substeps, until no more vertices in the given distance range get updated.

%% file: act-lis-fig.tex
\begin{figure*}
  \centering
  \vspace{-.1in}
  \begin{minipage}[h]{0.4\textwidth}
  \vspace{-.25in}
  \includegraphics[width=\columnwidth]{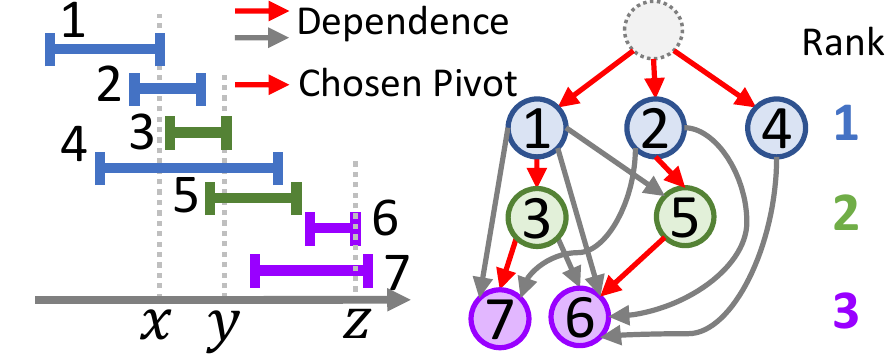}
  \caption{\small \textbf{Illustration of the activity selection problem.}
  Left: The start and end time of 7 activities (ordered by end time).
  Right: The dependences between activities.
  Rank 1 activities start before $x$ (shown in blue).
  Rank 2 activities start before $y$ (shown in green).
  Rank 3 activities start before $z$ (shown in purple).
  Red dependences are the pivots chosen in the Type 2 algorithm.
  The pivot of an object is the compatible activity before it with the latest start time.
  A rank-$r$ object has a pivot with rank $r-1$.
  %\zheqi{fig 2 is not referred in anywhere else}
  \label{fig:activity}}
  \end{minipage}\hfill
  \begin{minipage}[t]{0.58\textwidth}
      \includegraphics[width=\columnwidth]{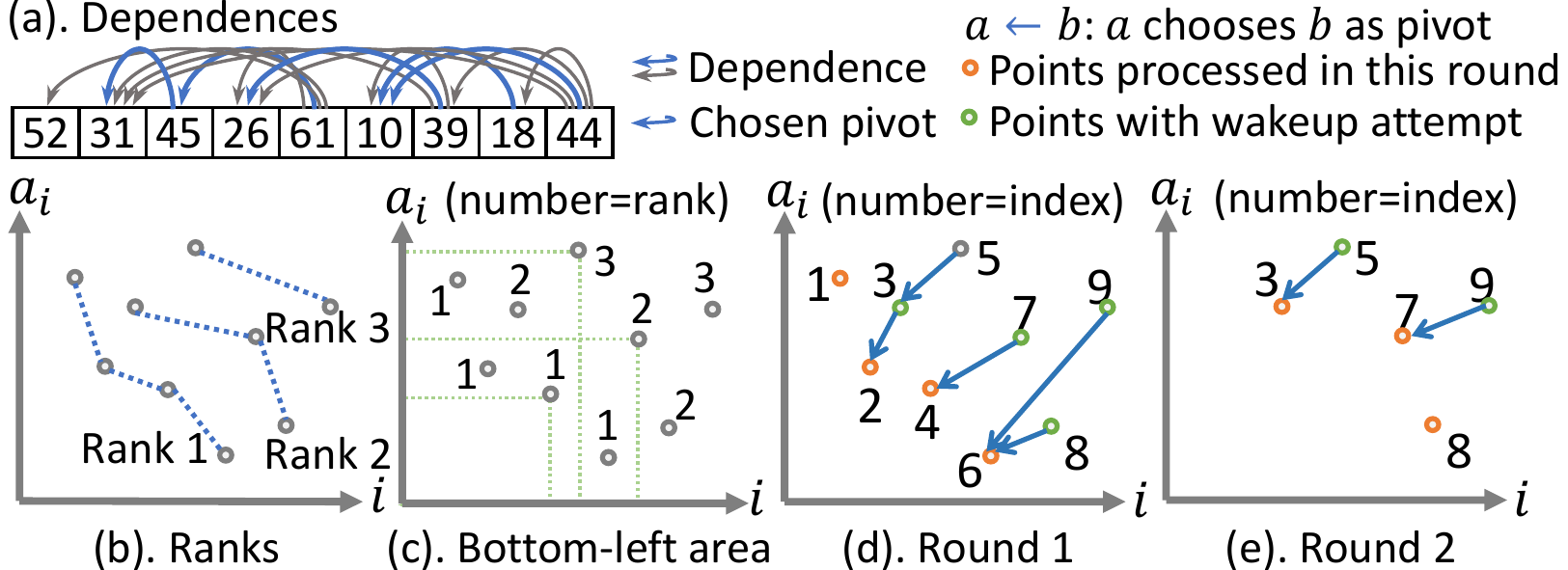}
  \caption{\small \textbf{Illustration of the LIS algorithm.}
  (a). Objects, their dependences, and the random pivot chosen.
  (b). Ranks of the objects represented as 2D points $(i,a_i)$.
  (c). Objects (points) and their bottom-left area. A point with rank $i$ only has points with rank $<i$ in its bottom-left area.
  (d). Points processed (\nodecircle{1}, \nodecircle{2}, \nodecircle{4}, \nodecircle{6}) and waking-up attempts in round 1. \nodecircle{2} wakes up \nodecircle{3}. \nodecircle{4} wakes up \nodecircle{7}. \nodecircle{6} attempts to wake up \nodecircle{8} and \nodecircle{9}, but \nodecircle{9} is not ready.
  (e). Points processed (\nodecircle{3}, \nodecircle{7}, \nodecircle{8}) and waking-up attempts in round 2. %\nodecircle{5} and \nodecircle{9} will be waken up in the next round.
  \label{fig:lis}
  }
  \end{minipage}
  \vspace{-.2in}
\end{figure*} 

%% file: huffmantree.tex
The Huffman tree organizes a set of objects $S_i$ with their \emph{frequency} $f_i$.
Constructing a Huffman tree is to iteratively extract two objects $S_i$ and $S_j$ with the smallest frequencies,
removes them, and re-inserts a new object $S_k$ with frequency $f_k=f_i+f_j$ to the object set.
This can be represented using the \emph{Huffman tree} where $S_k$ is the parent of $S_i$ and $S_j$.
Each iteration in the algorithm processes (generates) internal nodes in order. %The output of the algorithm constructs the tree.
The iteration of creating node $v$ depends on the creation of all its descendants.

To parallelize this, note that the Huffman tree itself is the \dg{} of the problem. Two objects are \compatible{} if there is a path between them on the Huffman tree (i.e., one is the prefix of the other). The rank of each object is its subtree height. Clearly, we can achieve the best parallelism if we first process all objects with height 1, then height 2, 3, etc. However, it is again hard to follow the \emph{exact} set of objects with a certain height. We here present a parallel algorithm below and show it corresponds to a valid \relaxedrank{}.
The high-level idea is that, in every round, we can first combine the two smallest frequencies and get their sum of $f_m$, and then all objects with frequency smaller than $f_m$ can be processed in parallel.
We note that this idea is known\footnote{We are aware of a similar algorithm used in some parallel algorithm courses, but did not find a formal reference.}. Here we just use it as an additional example of Type 1 algorithms.

Because of the greedy essence of the algorithm, if the sum of the current smallest two frequencies is $f_m$, all objects with frequency smaller than $f_m$ are clearly ready because no further objects will have frequency smaller than $f_m$.
%all nodes that are ready to be merged (except for the possibly last one if there an odd number of such nodes) are those with frequency less than $f_m$, where $f_m$ is the sum of the two current smallest frequencies.
%Therefore, we could maintain a \pabst{} $T$ with all elements sorted by their frequency.
%In each round, we compute $f_m$ as the sum of the two smallest frequency among unused elements.
Therefore, we can split all characters using frequency $f_m$, and those with smaller frequencies will be the frontier $T$ (if there are an odd number of such objects, we leave out the last one).
%Then we pair these objects in order, and combine each pair's frequency to generate $|T|/2$ new objects (internal nodes) and combine them with the other unused objects.
Then we pair these objects in order, and combine each pair's frequency to generate $|T|/2$ new objects (internal nodes) and merge them with the objects whose frequency is larger than $f_m$.
We repeat this process until the last object is created.
The work for extracting minimum and splitting is $O(\log n)$.
Combining the new objects with the old unused ones takes time $O(|T|\log n)$.
Considering the total size of all set $T$ in the algorithm is $O(n)$, the algorithm has work $O(n\log n)$, which is asymptotically the same as the preprocessing sorting time. The span of the algorithm is $O(\rank(S)\log n)=O(H\log n)$ where $H$ is the height of the Huffman tree. Note that although here the objects are not processed strictly based on the order of rank (the special case is usually caused by having $|T|$ as an odd value and having to postpone one to the next round),
the final number of rounds is still $O(\rank(S))$, which is round-efficient.

\begin{definition}
\label{def:relaxedrank}
The \relaxedrank{} of the parallel Huffman tree algorithm can be defined as follows. Assume the object with the least frequency is $a_0$ with frequency $f_0$, the path from $a_0$ to the root contains
nodes with frequencies $f^*_0=f_0,f^*_1, f^*_2,\cdots, f^*_H=1$, where $H$ is the Huffman tree height.
An object with frequency $f'$ has \relaxedrank{} $i$ if $f^*_i\le f' < f^*_{i+1}$.
\end{definition}

We can easily verify that this is a valid \relaxedrank{} as defined in Definition \ref{def:relaxedrank}.
We then show that our algorithm will process the objects based on their \relaxedrank{}.
Note that the objects on the path from $a_0$ to the root will be processed in this order. i.e., $a_0$'s parent will be processed in the first round,
$a_0$'s grandparent will be processed in the second round, so on so forth. This is because in each round, the ancestors of $a_0$ must be one of the two objects
with the smallest frequency.
Therefore, if $f^*_i\le f' < f^*_{i+1}$, the object with frequency $f'$ must be processed in the same round as $f^*_i$, which is consistent with its \relaxedrank{} $i$.

It is easy to see that the algorithm finishes in $H$ rounds, since the largest \relaxedrank{} is $H$. This means that, even though we are not processing the objects based on the exact ranks, the algorithm is still round-efficient and finishes in $O(H)=O(\rank(S))$ rounds. This proves \cref{thm:huffman}.
\iffullversion{
\huffmanthm
}
\fi
We first note that the work of the algorithm is dominated by sorting all input frequencies. In fact, one could reduce the work of the later process (other than sorting) to linear, which is also a well-known trick in the sequential Huffman tree construction algorithm. We did not especially show the linear-time version because this version fits directly in our framework.

%For $n$ elements,
%this algorithm runs in $O(n\log n)$ work and $O(H\log n)$ span, where $H=\rank(S)$ is the huffman tree height.

Note that the span can be linear in the worst case, but it is less likely in realistic settings. This is because a tree of height $H$ indicates that $\max f_i/\min f_i$ is $O(2^H)$. Considering the precision supported by the machine, $H=O(\log W)$, where $W$ is machine word size.

\hide{all internal nodes $v$ with height $h(v)=1$ can be first generated in parallel.
This is because all such nodes depend on the creation of two leaves, and their information is available in the input.
None of the other nodes can be processed, because they all depend on some other internal nodes that have not been created yet.
Similarly, then all internal nodes with height 2 are ready, and so on.
We say two nodes are compatible if one node is the ancestor of the other (the corresponding Huffman code is a prefix of the other), and the rank of a node is its height in the tree.
One could easily verify that this algorithm is \phaseparallel{}.
} 

%% file: type2.tex
\section{Type 2 Algorithms with Wake-up Strategies}\label{sec:type2}
\label{sec:lis}

%In \phaseparallel{} algorithms, the rank of $x\in S$ is
%$\max\{\rank(y)~|~y\in \mathcal{P}(x)\}+1$,
%which means if all predecessors of object $x$ are ready before or at round $i$, then $x$ will be ready in $(i+1)$-th round.
In \phaseparallel{} algorithms, an object is ready when all its predecessors $\mathcal{P}(x)$ finish.
%If we can precisely identify the moment where the last finished object in $\mathcal{P}(x)$, we can avoid checking the readiness of each object
%exhaustedly.
Previous approaches require explicitly generating $\mathcal{P}(x)$ for each $x\in S$~\cite{blelloch2020optimal,blelloch2020randomized,blelloch2018geometry,blelloch2016parallelism,shun2015sequential} (achieving work-efficiency only when $|\pred(x)|=O(1)$),
checking the readiness of all objects every round~\cite{blelloch2012internally} (not necessarily work-efficient),
or based on dual-binary search~\cite{BFS12,hasenplaugh2014ordering} (incurs overhead in span).
To avoid exhaustedly checking the readiness of every object,
Type 2 algorithms aim to \emph{wake up} an object $x$ when the \emph{last} object in $\pred(x)$ finishes.
%Unlike Type I algorithms that directly extract each frontier from the intrinsic properties of the objects,
%the rank function for each object in Type II algorithms is evaluated independently.
%Type II algorithms aim at \emph{waking up} an object $x$ when the last object in $\pred(x)$ finishes.
%All previous approaches require explicitly generating $\mathcal{P}(x)$ for each $x\in S$, and the readiness of objects is either checked every round~\cite{blelloch2012internally,shun2015sequential} (not necessarily work-efficient), or based on dual-binary search~\cite{BFS12,hasenplaugh2014ordering} (not span-efficient),
%only applies to constant size $\mathcal{P}(x)$~\cite{blelloch2020optimal,blelloch2020randomized,blelloch2018geometry,blelloch2016parallelism}.
%only applies to constant size $\mathcal{P}(x)$.

%In this paper, we give two general approaches to parallelize iterative algorithms that yield better work and span bounds on a list of classic problems.
%The key components in both approaches are the \emph{wake-up strategies} that find out the readiness of the objects efficiently.
We propose two wake-up strategies. %to efficiently identify the readiness of the objects.
Our first approach, which we believe is very interesting, is to avoid explicitly generating $\mathcal{P}(x)$,
%and to use a range query to directly check if an object is ready.
and check the readiness of an object $x$ when $x$ is likely to be ready.
%Unfortunately, doing so every round is costly if the check fails too many times, so we want to perform the check only when we are somehow confident that $x$ is ready.
%However, checking the readiness is still costly, so we want to perform the check only when we are somehow confident that $x$ is ready.
To do so, we attach each object $x$ to an unfinished object $p_x\in\pred(x)$, called the \emph{pivot}, which blocks $x$.
We redo the check only when $p_x$ is ready, which bounds the number of total checks to be $O(\log |\mathcal{P}(x)|)$ \whp{}.
%\yan{work and span}
\ifconference{
We show activity selection and longest increasing sequence (LIS) as examples, and more applications in the full version.
}
\else{
We show activity selection and longest increasing sequence (LIS) as examples, and more applications in the appendix.
}
\fi

The second approach applies to algorithms that can afford to generate $\mathcal{P}(x)$ for each $x\in S$. %---our new solution can asynchronously execute the computation in parallel, and applies to most of the existing algorithms.
Our idea is to build an asynchronous structure using \tas{} to precisely identify when the last object in $\pred(x)$ is ready, and achieve better span.
\ifconference{
We present the greedy MIS algorithm as an example in \cref{sec:mis}, and more discussions in the full version of the paper.
}
\else{
We present the greedy MIS algorithm as an example and other similar problems in \cref{sec:mis}.
\fi
%and also some new applications such as parallel edit distance (string matching).

\subsection{Activity Selection}

%Another example of using wake-up scheme to avoid exhausted checking is also the activity selection problem.
%Recall from \cref{sec:as-type1} that the recurrence is
%$$D[i]=\max_{\timeend(A_j)\le \timestart(A_i)}D[j]+w_i$$
%$D[i]=\max_{\timeend_j\le \timestart_i}D[j]+w_i$.

We now revisit the activity selection problem and present an algorithm using Type 2 framework. Recall the DP recurrence $\mathdp[i]=\max_{\timeend_j\le \timestart_i}\mathdp[j]+w_i$.
Therefore, $A_i$ is ready when all other activities with end time before $\timestart{}_i$ have been processed.
%We can let each activity hook to a random pivot activity with $\timeend{}_j\le \timestart{}_i$. When the DP value of an activity has been determined, it wakes up
%all other activities using itself as the pivot to check if they are ready. If so, these activities will be processed in the next round. If not, they hook to another unready activity as the new pivot. This can be done by using a \pabst{} of all activities, which set keys as end time, and augmented value as a pair of the number of unready points and the maximum DP value.
%Since only a 1D range query is involved, the total work is $O(n\log^2 n)$ \whp{} using a similar analysis of LIS for $n$ activities.
%The span of the algorithm is $O(S\log n)$, where $S$ is the rank of the input set (the maximum number of \compatible{} activities in the input).
%The question boils down to let each activity find a \emph{pivot}, which
Our idea is to let each activity $x$ find a \emph{pivot} $p_x$, where finishing processing $p_x$ indicates the readiness of $x$.
%We can optimize the problem to work-efficiency based on the following observation: instead of waiting for \emph{all} possible pivots an activity $A_i$ depends on, we could find \emph{one specific} pivot activity $\ay$ that \block{s} $\ax$, such that $\ax$ is ready once $\ay$ is ready. In particular,
We prove the following lemma.

\begin{lemma}
\label{lem:lateststart}
  Given activity $\ax$, let activity $\ay=\arg \max_{A_i:\timeend_i \le \timestart_x}\timestart_i$ be the latest-start activity among all activities ending before $\ax$ starts (i.e., those earlier than $\ax$ and compatible with $\ax$), and call $\ay$ the \defn{\leading{} activity} of $\ax$. Then $\rank(\ax)=\rank(\ay)+1$.
\end{lemma}

\begin{proof}
  By definition, $\ax$ depends on $\ay$ because $\ax$ is \compatible{} with all \MFS{} ending with $\ay$.
  This indicates $|\mis(\ay)|=\rank(\ay)< \rank(\ax)=|\mis(\ax)|$.
  We then show $\mfs(\ay)\cup \{\ax\}$ is an MFS of $\ax$.
  Assume to the contrary that $t=|\mis(\ax)|>|\mis(\ay)|+1$.
  Consider such an \MFS{} $T=\mis(\ax)$ where $|T|=t$.
  Let activity $A_k$ be the activity with the latest starting time in $T-\{\ax\}$.
  We first prove $\ay\notin T-\{\ax{},A_k\}$. This is because if $\ay\in T-\{\ax{},A_k\}$,
  $\ay$ should have an earlier starting time than $A_k$ (by definition of $A_k$), which contradicts the definition of $\ay$.
  %Note that the problem exhibit optimal substructure, and therefore $T-\{\ax\}$ is an \MFS{} of $A_k$.
  Similarly, all other activities in $T-\{\ax\}$ have earlier end time than $\timestart_k$.
  By the definition of $\ay$, $A_k$ starts no later than $\ay$ ($\ay$ is the latest starting activity before $\ax$ and compatible with $\ax$).
  This means that $\ay$ is also compatible with (and later than) $T-\{\ax,A_k\}$.
  $\mis(\ay)\ge |(T-\{\ax{},A_k\})\cup \{\ay\}|=t-1$, contradicting the assumption of $t=|\mis(\ax)|>|\mis(\ay)|+1$.
\end{proof}

We show an example of \leading{} activities in \cref{fig:activity}. \cref{lem:lateststart} implies that the \leading{} activity $p_x$ of any activity $x$ must be processed in the previous round of when $x$ is processed.
%In other words, the \leading{} activity of an activity $\ax$ is the last activity that \block{s} $A_x$. In this case, we can simply set
%the pivot of $\ax$ to $\ay$. Once $\ay$ has been processed, $\ax$ must be ready and can be waked up by $\ay$.
In other words, once $p_x$ is finished, we can wake up $x$ and process it in the next round.
In this case, we can first let all activities find their pivot via binary searches,
which is $O(\log n)$ work per activity.
We use a tree $\tpivot$ as a multi-map to store all pairs $(p_x,x)$.
We start with processing all activities with rank 1.
For each activity $y$ in the current \frontier{}, after processing them,
we find all pairs $(y,z)\in \tpivot$ and put all such $z$ in the next frontier (they will be wakened up).
An activity can be processed (computing its DP value) similarly as in Type 1 by using a \pabst{} $T_{\mathit{DP}}$.
We have the following theorem.

\begin{theorem}\label{them:activity_2}
  Type 2 activity selection algorithm takes $O(n\log n)$ work and $O(\rank(S)\log n)$ span, where $S$ is the input set and $n=|S|$.
\end{theorem}

\myparagraph{An $O(\log n)$ span algorithm for unweighted activity selection.}
Based on \cref{lem:lateststart}, we can further design a parallel algorithm with better span for the \emph{unweighted} activity selection problem, where each activity has a unit weight ($w_i=1$).
Note that this is equivalent to computing the $\rank$ of each activity.
Based on \cref{lem:lateststart}, we can rewrite the DP recurrence for the unweighted version as:

\vspace{-.1in}
$$\mathdp[i]=\mathdp[j]+1:A_j \text{ is the \leading{} activity of }A_i$$
\vspace{-.1in}

This simplifies the \dg{} to a tree structure, where each activity only relies on its \leading{}. The rank of each activity is also its depth in this tree, which can be computed using a standard tree contraction~\cite{blelloch2020randomized} in $O(n)$ work and $O(\log n)$ span \whp.

\begin{theorem}
  The unweighted activity selection problem can be solved in $O(n\log n)$ work and $O(\log n)$ span \whp.
\end{theorem}

\input{lis}

\input{mis.tex} 

%% file: lis.tex
\subsection{Longest Increasing Subsequence (LIS)}

We propose a parallel algorithm for the longest increasing subsequence (LIS) problem using our \phaseparallel{} framework.
Given a sequence~$a$, LIS asks for the longest subsequence in $a$ that is strictly increasing.
LIS is widely studied, and its parallel solutions have been studied in~\cite{galil1994parallel,im2017efficient,krusche2009parallel,krusche2010new,alam2013divide,seme2006cgm,thierry2001work,nakashima2002parallel,nakashima2006cost}.
Most of these algorithms~\cite{galil1994parallel,krusche2009parallel,seme2006cgm,thierry2001work,nakashima2002parallel,nakashima2006cost} introduced polynomial overhead in work, and Alem and Rahman's algorithm~\cite{alam2013divide} has $\tilde{\Theta}(n)$ span.
Krusche and Tiskin's BSP algorithm~\cite{krusche2010new} translates to $\tilde{O}(n)$ work and $\tilde{O}(n^{2/3})$ span, which is the only nearly-work-efficient algorithm with sublinear span.
This algorithm relies on complicated techniques from~\cite{tiskin2015fast}, and has no implementation.
In fact, we are unaware of any previous parallel LIS implementation with
competitive performance to the standard sequential $O(n\log n)$ LIS algorithm.

Sequentially, LIS can be computed using the dynamic programming (DP) recurrence as follows. Let $\mathdp[i]$, called the DP value, be the LIS length of $a_{1..i}$ ending with $a_i$.
Then

\vspace{-.05in}
\begin{equation}\label{eqn:lisdp}
  \mathdp[i]=\max(1,\max_{j<i,a_j<a_i}\mathdp[j]+1)
\end{equation}
\vspace{-.05in}

We can iteratively compute $\mathdp[i]$ and maintain any search structure to find $\max_{j<i,a_j<a_i}\mathdp[j]$ in $O(\log n)$ work.
Our new \phaseparallel{} LIS algorithm parallelizes this sequential algorithm and achieves near work-efficiency ($\tilde{O}(n)$ expected work) and round-efficiency.
Moreover, our algorithm is implementable, and we show experimental study in \cref{sec:exp-lis}.
Here we maximize LIS length, but our algorithm can be generalized to the weighted case where objects have different weights.

An object $a_i$ is ready once all objects $a_j$ with $j<i,a_j<a_i$ are ready. %We say $a_j$ \defn{\block{s}} $a_i$ in this case.
%When there are multiple objects ready, they can all be processed in parallel.
%The main challenge here is that we cannot afford processing all dependencies since there can be $\Omega(n^2)$ of them.
%Even generating all such dependencies can blow up the work.
In our \phaseparallel{} framework,
the rank of an object $a_i$ is the size of the LIS ending with $a_i$ (its DP value).
%start with defining a the \emph{rank} for each object. Inspired by the DP recurrence, let $\rank(a_i)$ be the LIS length ending with $a_i$.
An object $a_i$ only depends on objects with a smaller rank.
After all objects with rank~$r$ have finished, all objects with rank~$r+1$ must be ready.
%However, the difficulty lies in that we cannot quickly find all elements with a certain rank quickly: this boils down to the original problem of computing the LIS of the each element.
The main challenge is to avoid processing all dependencies since there can be $\Theta(n^2)$ of them.
Interestingly, the problem exhibits a nice geometric property.
If we draw each object as a 2D coordinate $(i,a_i)$,
all the objects that $a_i$ relies on are the points to its lower-left area (see \cref{fig:lis}).
Therefore we can determine if an object $x$ is ready using a 2D range query on the number of unfinished objects to its lower-left area.
Similarly, the DP value of $x$ can be obtained by querying the maximum DP value among all objects in its lower-left area (and plus one).
Both queries can be done by a augmented 2D range tree (see \cref{sec:prelim}).
This gives a simple algorithm, where in each round, we can run a 2D range query to every unfinished object to identify the ready ones,
and then compute their DP values. However, this can still incur $\Omega(n^2)$ work in the worst case.

To achieve work-efficiency, we would like to wake up an object only when it is (almost) ready.
Unlike activity selection, we found it hard to find an \emph{exact} pivot for each object $x$ that has rank $\rank(x)-1$.
Instead, our strategy is to randomly pick an unfinished object in $\pred(x)$ as $x$'s pivot, which can be efficiently supported by an augmented 2D range tree.
When the pivot of $x$ is processed, we \emph{attempt} to wake up $x$ by checking whether all objects in $\pred(x)$ (those in its lower-left corner) are finished.
If so, $x$ is ready, and we query the maximum DP value in $\pred(x)$ to compute $x$'s DP value.
Otherwise, $x$ is not waked up successfully, and selects another random unfinished object in its lower-left corner as the new pivot.
In each round, all ready objects will attempt to wake up all objects using $x$ as the pivot.
This inductively guarantees that objects with rank $i$ (LIS length $i$) are waked up and processed in round $i$.

Our algorithm is in \cref{algo:lis}.
Each object corresponds to a \emph{point}, defined on its $x$-coordinate (its index $i$), and $y$-coordinate ($a_i$).
We also maintain its DP value $\mathdp$ (initialized to $+\infty$).
We create a virtual point $p[0]$ as a starting point with index $0$ and value $-\infty$.

\input{lis-code.tex}

We use a range tree $\trange$ to maintain all the points in the 2D planar, augmenting on a triple $\langle n_{\infty},\mathdp^*,x^* \rangle$, which records for the current subtree,
the number of unfinished points $n_{\infty}$, the maximum DP value $\mathdp^*$, and an $x$-coordinate $x^{*}$ (an index).
%, which is either a potential new pivot or the index to achieve the maximum DP value (see details below)
If the maximum DP value $\mathdp^{*}$ is $\infty$,
which means that there exist unfinished elements in this subtree, then the index $x^{*}$ is selected uniformly at random from the unfinished objects.
Otherwise, $x^{*}$ is used to record the index to achieve the maximum DP value, which can be used to reconstruct the LIS if needed.
To maintain such augmented values, the combine function simply adds up $n_{\infty}$ (Lines \ref{line:combinereturn1} and \ref{line:maxdpx}),
and takes a maximum on $\mathdp^*$ (Line \ref{line:getmax}) on the two augmented values $a_1$ and $a_2$.
If $\mathdp^*$ is not $+\infty$, $x^{*}$ can be simply set to be the argument to achieve the highest DP value (Line \ref{line:maxdpx}).
Otherwise, $x^*$ is selected from the $x^{*}$ value of either $a_1$ or $a_2$, and the probability is decided by $t_1:t_2$, where
$t_1$ and $t_2$ are the number of unfinished objects (the $n_{\infty}$ values) in $a_1$ and $a_2$, respectively (Line \ref{line:selectrandom}).
%and $t_2$ is the number of unfinished elements in $a_2$ (Line \ref{line:selectrandom}).
By doing this, $x^{*}$ is selected uniformly at random from both $a_1$ and $a_2$.
We also use a multi-map $\tpivot$ to maintain the pivot-object pairs.

The algorithm starts from a frontier of the virtual point $p[0]$.
Initially, $\tpivot$ stores pairs $(0,i)$ for all $i$,
since $p[0]$ is the initial pivot of all objects.
In each round, the algorithm processes each object $x$ in the frontier in parallel.
We first find all objects $q$ such that $\langle x, q\rangle \in \tpivot$,
which are all objects with pivots in the frontier (Line \ref{line:findtodo}).
%For each such object $q$ (Line \ref{line:forq}),
%we attempt to wake it up by searching in $\trange$ the half-open rectangle with top-right point as $q$ (Line \ref{line:rangesearch}),
We attempt to wake up each such object $q$ by searching in $\trange$ the half-open rectangle with top-right point as $q$ (Line \ref{line:rangesearch}),
getting triple $\langle \_, k, i\rangle$, where $k$ is maximum $\mathdp^{*}$ in the query range and $i$ is the $x^{*}$ value.
If $k\ne +\infty$ (Line \ref{line:setdpif}), meaning
that there is no unfinished object in $q$'s lower-left area,
then $q$ is ready and we set the DP value of $q$ as $1+k$ (Line \ref{line:getdp}).
%Also, the previous element $\mathpre$ of $q$ can be set using the returned value of $i$ (Line \ref{line:getpre}).
Otherwise ($n_{\infty}\ne 0$), there are still unfinished objects in $q$'s lower-left area, and we reset $q$'s pivot as $i$,
which is selected uniformly at random from all unfinished objects in the queried range.

At the end of a round, all newly-generated pivot-object pairs are inserted into $\tpivot$ in parallel (Line \ref{line:getnewpivots}--\ref{line:insertpivots}).
All newly finished objects are packed into next frontier (Line \ref{line:getnextfrontier}).
The DP values of the newly finished objects are updated in $\trange$ (Line \ref{line:updatedp}).
\ifconference{
We use the following lemma to prove the work of the algorithm and prove it in the full version of this paper~\cite{iterativefull}.
}
\else{
We use the following lemma to prove the work of the algorithm.
}
\fi

\begin{lemma}
\label{lem:backwardlogn}
%Let $\mathit{rand}(a,b)$ be a function generating a uniformly random integer in range $[a,b]$.
%Given a sequence $a$ where $a_0=0$, $a_i=\mathit{rand}(a_i-1,n)$, and let $k=\min_{a_i=n}i$, then $k=O(\log n)$ \whp{}.
%In other words, for a sequence of elements, in step $0$ we select the element
Construct a sequence $x_i$ as follows.
%Let $x_0=1$, and $x_i=\mathit{rand}(x_{i-1},n)$, which means $x_i$ is a uniformly random number selected from $x_{i-1}$ to $n$. Let the sequence stop at $x_k$, which is the first element s.t. $x_k=n$. Then $k=O(\log n)$ \whp{}.
Let $x_0=1$, and $x_i$ be a uniformly random number selected from $x_{i-1}$ to $n$.
Let $k$ be the first element s.t. $x_k=n$. Then $k=O(\log n)$ \whp{}.
\end{lemma}
\iffullversion{
\input{lis-proof}
}
\fi

\begin{lemma}
\label{lem:lislogn}
For each object $x$ in the input sequence of size $n$, Algorithm \ref{algo:lis} will attempt to wake up $x$ for $O(\log n)$ times \whp{}.
\end{lemma}
\begin{proof}
  %This can be derived directly from Lemma \ref{lem:backwardlogn}.
  For an object $x$, let $S$ be the sequence of objects before $x$ and smaller than $x$ sorted by rank.
  Let the first pivot of $x$ be the $p$-th object in $S$, which is selected uniformly at random from $S$.
  When $S_p$ finishes, all objects in $S$ with rank smaller than $S_p$ must have finished,
  and the next pivot is selected uniformly at random from $S_{p'..n}$, where $p'>p$. Similar process applies to later pivot selections.
  Therefore, the number of pivots selected for $x$ is no more than the length of sequence in Lemma \ref{lem:backwardlogn}, which is $O(\log n)$ \whp{}.
\end{proof}

\begin{theorem}\label{them:lis}
  %The work of the parallel LIS algorithm in Algorithm \ref{algo:lis} is $O(n\log^3 n)$ \whp{} on a sequence of size $n$.
  Algorithm \ref{algo:lis} computes the LIS of a sequence of size $n$ in $O(n\log^3 n)$ work and $O(r\log^2 n)$ span \whp{}, where $r$ is the rank (LIS length) of the input sequence.
\end{theorem}
\begin{proof}
  First, all initialization including constructing the range tree has work $O(n\log n)$ and span $O(\log^2 n)$.
  Assume in round $i$ there are $n_i$ objects in the $\todo$ list, which are the objects that we attempt to wake up.
  Line \ref{line:findtodo} finds at most $n_i$ objects from $\tpivot$ with $O(n_i\log^2 n)$ work.
  Each range query costs $O(\log^2 n)$, and thus the total cost in the parallel for-loop in Line \ref{line:forq} is $O(n_i\log^2 n)$.
  Line \ref{line:getnextfrontier} packs at most $n_i$ objects with work $O(n_i)$.
  Line \ref{line:getnewpivots} and \ref{line:insertpivots} pack and insert at most $n_i$ elements to $\tpivot$, and thus costs $O(n_i\log^2 n)$.
  Line \ref{line:updatedp} updates at most $n_i$ objects to $\trange$, and thus costs $O(n_i\log^2 n)$.
  In summary, each round of function $\mathit{wakeup}$ has work $O(n_i\log^2 n)$.
  From Lemma \ref{lem:lislogn}, we know that $\sum n_i=O(n\log n)$. This proves that the total work of Algorithm \ref{algo:lis} is $O(n\log^3 n)$ \whp{}.

  By definition of rank, the number of rounds in Algorithm \ref{algo:lis}
  is $r$. In each round, the span bounds of range queries (Line \ref{line:rangesearch}), pack (Lines \ref{line:getnextfrontier}--\ref{line:getnewpivots}), and update $\tpivot$ and $\trange$ (Lines \ref{line:insertpivots} and \ref{line:updatedp}) are all $O(\log^2 n)$. Therefore, the total span is $O(r\log^2 n)$.
\end{proof}

\hide{
\begin{theorem}
  The span of the parallel LIS algorithm in Algorithm \ref{algo:lis} is $O(r\log^2 n)$, where $n$ is the input sequence size and $r$ is the rank (LIS length) of the input sequence.
\end{theorem}
}

\iffullversion{
\myparagraph{Other similar algorithms.} We discuss another interesting algorithm similar to LIS, the \emph{Whac-a-mole} problem, in Appendix \ref{app:mole}.
}\fi

%% file: lis-code.tex
\begin{algorithm}[t]
\fontsize{8pt}{10pt}\selectfont
\caption{The parallel LIS algorithm\label{algo:lis}}
\SetKwFor{parForEach}{parallel\_for\_each}{do}{endfor}
\KwIn{A sequence $a[1..n]$ with comparison function $<$}
\KwOut{The LIS length of $a[\cdot]$}
    \vspace{0.5em}
\SetKwProg{MyStruct}{Struct}{ contains}{end}
\SetKwProg{RangeTree}{RangeTree<Point>}{ with}{end}
\SetKwProg{NestedTree}{NestedTree<pair<int>{}>}{ with}{end}
\SetKwProg{AugFunc}{}{}{end}
\SetKwProg{MyFunc}{Function}{}{end}
\newcommand{\Int}{\KwSty{int}}
\newcommand{\A}{\KwSty{A}}
\DontPrintSemicolon
\MyStruct{Point}{
  \Int{} $x,y$ \tcp*{x = index, y = a[x]}
  %\A{} $y$ \tcp*{a[x]}
  \Int{} $\mathdp$ \tcp*{The DP value}
  %\Int{} $\mathpre$ \tcp*{The previous position of LIS}
}
Point $p[1..n]$\\
$p[0] \gets \langle 0,-\infty,0\rangle$\tcp*{insert virtual point 0}
\lparForEach{$a[i]\in a$}{
  $p[i] = \langle i, a[i],+\infty\rangle$
}
\RangeTree{$\trange$}{
  $<_x(p_1,p_2)$: \Return{$p_1.x<p_2.x$}\\
  $<_y(p_1,p_2)$: \Return{$p_1.y<p_2.y$}\\
  \tcp{$n_{\infty}$: \# of unfinished points (dp value $\infty$)}
  \tcp{$\mathit{dp}^{*}$: max dp value in subtree}
  \tcp{$x^{*}$: index (x) of max dp value}
  \KwSty{augmented value}: $\langle n_{\infty},\mathit{dp}^{*},x^{*}\rangle$\\
  \AugFunc{\upshape\KwSty{base}($p$):}{
    \lIf{$p.\mathdp{}=+\infty$}{\Return{$\langle 1,p.x,p.\mathdp{}\rangle$}}
    \lElse{\Return{$\langle 0,p.x,p.\mathdp{}\rangle$}}
  }
  \AugFunc{\label{line:combine1}\KwSty{combine}($a_1,a_2$):\tcp*[f]{combine two aug values}}{
    $m\gets \arg\max_{i\in\{1,2\}} a_i.\mathdp^{*}$\label{line:getmax}\\
    \tcc{When max dp value is $+\infty$, choose a uniformly random one as the potential pivot}
    \If{$a_m.\mathdp^{*}=+\infty$}{
      $t\gets$random$(1,2)$ with probability $a_1.n_{\infty}:a_2.n_{\infty}$\label{line:selectrandom}\\
      \Return{$\langle a_1.n_{\infty}+a_2.n_{\infty},+\infty,a_t.x^{*} \rangle $}\label{line:combinereturn1}
    }
    \Return{$\langle 0,a_m.\mathdp^{*}, a_m.x\rangle $}\label{line:maxdpx}
  }
}
Construct $\trange$ from $p[\cdot]$\\
%multi-map{$\tpivot$}{
  %primary$(p)$: $p.\mathfirst$\\
  %secondary$(p)$: $p.\mathsecond$
%}
\textbf{Multi-map} $\tpivot=\{(0,i):i=1..n\}$ \tcp*{pivot pairs $(p_x,x)$}
%Insert $\langle 0, i\rangle$ into $\tpivot$ for all $i=1..n$\\
frontier = \{0\}\\
\While {frontier $\ne \emptyset$} {
  frontier = \mf{WakeUp}(\mb{frontier})
}
\Return{$\max_i(p[i].\mathdp)$}

\medskip
\MyFunc{\upshape\mf{WakeUp}(frontier)}{
  todo $\gets \tpivot.$multi\_find(\mb{frontier})\label{line:findtodo}\\
  \parForEach{\upshape$q\in$ todo\label{line:forq}}{
    $\langle \_,k,i\rangle\gets\trange.$range$(q.x,q.y)$\label{line:rangesearch}\\
    \If{$k\ne +\infty$\label{line:setdpif}}{
      $q.\mathdp\gets k+1$\label{line:getdp}\\
      %$q.\mathpre\gets i$\label{line:getpre}\\
      mark $q$ as next\_frontier
    } \lElse(\hfill\texttt{// not ready yet}){
      mark $\langle i, q\rangle$ as new\_pivot\_pair
    }
  }
  Pack points marked as next\_frontier into frontier$^*$\label{line:getnextfrontier}\\
  Pack points marked as new\_pivot\_pair into pivot$^*$\label{line:getnewpivots}\\
  $\tpivot.$multi\_insert(pivot$^*$)\label{line:insertpivots}\\
  Update $\mathdp$ values for $q\in$ frontier$^*$ in $\trange$\label{line:updatedp}\\
  \Return{frontier$^*$}
}
\end{algorithm} 

%% file: lis-proof.tex
\begin{proof}
	We will show that $k$ is equivalent to the recursion depth in a quicksort on $n$ input objects of $[1..n]$. Fix the largest object $n$ and consider all the subproblems it belongs to during quicksort.
	In the first round, the pivot $p_1$ is selected uniformly at random, and the subproblem $n$ belongs to will be the subsequence from $p_1$ to $n$.
	In round $i$ of recursion, the subproblem $e$ belongs to is from $p_i$ to $n$, where $p_i$ is selected uniformly at random from $p_{i-1}$ to $n$. This means that the sequence $x_i$ constructed in Lemma \ref{lem:backwardlogn} is equivalent to all the pivots in $n$'s subproblems. From the fact that the recursion depth of a quicksort is $O(\log n)$ \whp{}, we can prove Lemma \ref{lem:backwardlogn}.
\end{proof}

%% file: mis.tex
\input{mis-fig.tex}
\subsection{Greedy MIS and Related Applications}\label{sec:mis}

In this section, we propose a new parallel MIS (maximal independent set) algorithm.
Parallel MIS is widely-studied~\cite{chatterjee2020sleeping,afek2013beeping,daum2013maximal,fineman2014fair,dahlum2016accelerating,Luby86,BFS12,Cook85,calkin1990probabilistic,coppersmith1987parallel,fischer2018tight}.
Given a graph $G=(V,E)$, an independent set $A\subseteq V$ is a subset of vertices where $\forall u,v\in A$, $(u,v)\notin E$.
An MIS is an independent set $A$ where $\forall v\in V,v\notin A$, $A\cup\{v\}$ is not an independent set.
The widely-adopted greedy MIS algorithm~\cite{Luby86,BFS12,Cook85,calkin1990probabilistic,coppersmith1987parallel,fischer2018tight} starts with assigning each vertex a random priority and greedily selecting the vertices based on the priority (highest to lowest).
The MIS is initialized as an empty set.
When processing vertex $v$, we add $v$ to the MIS if none of $v$'s neighbor is selected in the MIS, and skip $v$ otherwise.
We say a vertex is \defn{available} if none of its neighbors are selected in the current MIS, and \defn{unavailable} otherwise.
Blelloch et al. parallelized the algorithm~\cite{BFS12}.
Note that a vertex is ready when it has a higher priority than all its available neighbors.
We say a neighbor of $v$ is a \defn{\block{ing} neighbor} if it has a higher priority than $v$.
When all $v$'s \block{ing} neighbors become unavailable, $v$ is ready.
In each round, the parallel algorithm processes all ready vertices in parallel by adding them to the MIS, and removing their neighbors (marking as unavailable).
An illustration is shown in \cref{fig:mis}(a).
%We say a vertex is \defn{disabled} when it is removed by its neighbor.

The main challenge in the algorithm is to find ready vertices efficiently.
A previous approach \cite{BFS12} hooks each unfinished vertex $v$ to an unfinished neighbor with the highest priority as the pivot.
Only when $v$'s pivot is finished, it applies a dual binary search, which either finds the next pivot or decides that $v$ is ready.
Combining the results in~\cite{BFS12} and~\cite{fischer2018tight}, the work is $O(m)$ and the span is $O(\log^3 n)$ \whp{}.
The high span comes from the $O(\log^2n)$ dual binary search that can apply in each round ($O(\log n)$ steps each requiring an $O(\log n)$-span min-reduce).
There exist other algorithms such as Prism~\cite{kaler2016executing} which is work-efficient with $O(\log^2 n)$ span. However, it assumes to know the number of processors in the algorithm, and a strong constant-time atomic operation \texttt{fetch-and-decrease}.

In this paper, we show a new, asynchronous approach that improves span, while maintaining work-efficiency for parallel MIS.
The key idea is to wake up a vertex only when the last \block{ing} neighbor is finished.
We use the atomic operation \tas{} (\TAS{}), and build a complete binary tree, called the \tastree{}, for each vertex, and use it to check if all the \block{ing} neighbors are finished.
Let the \tastree{} of $v$ be $T_v$.
Each leaf in $T_v$ corresponds to a \block{ing} neighbor $u$ of $v$ (i.e., with a higher priority than $v$).
We set a flag (initialized to \zero{}) in each leaf to show if $u$ has been unavailable.
Each internal node in the \tastree{} also maintains a flag: %it is \false{} if both of its children are \false, and \true{} otherwise.
it is \one{} when at least one subtree is fully unavailable, and \zero{} otherwise.
For example, in \cref{fig:mis}(b),
\nodecircle{14} maintains four \block{ing} neighbors in its \tastree{}.
%uses a tomb bit to represent whether all leaves in its subtree are tombs, which can be obtained by taking a logical-AND on its two children.
%All nodes in the \tastree{} initialize the bit to be \false{}.
Note that when all leaves in the $T_v$ are \one{} (unavailable), $v$ is ready to be added to the MIS.
%We want the bit in the internal node to precisely reflect the unavailability the subtree.
%In other words,%
%when all vertices in a subtree are marked unavailable, we want to inform its parent.
We want to use the flag to reflect the unavailability of the subtree.
More precisely, for each subtree~$t$, when the last flag in~$t$ becomes \one{}, the information should be carried to $t$'s parent.

\input{mis-code.tex}

The pseudocode of our MIS algorithm is given in \cref{alg:mis}.
The algorithm starts with constructing the \tastree{s} (initialized to \zero{} for all tree nodes), and finding all vertices with an empty \tastree{} (the ready ones).
%It then finds all vertices with an empty \tastree{}, indicating that they are ready. %indicating that they have the highest priority among their neighbors.
%Therefore these vertices have rank 1 and will be put in the initial frontier.
%It then marks all their neighbors as unavailable.
%For any such vertex $x$, we update its neighbors in the \tastree{s} as unavailable.
%When processing a vertex in the frontier, we will mark each of its neighbors $u$ as unavailable.
We start with processing these vertices in parallel.
When processing a vertex, we will mark each of its neighbors $u$ as unavailable.
We further need to notify all the \tastree{s} containing $u$ that $u$ is now unavailable.
For each of the \tastree{s}, we set $u$'s flag in the leaf to be \one{} (Line \ref{line:mis-for-tastree}).
This information is prorogated up along the path to set the flag of its parent, call it $p$, to be \one{} by \TAS{} (Line \ref{line:mis-tas-while}).
If the \TAS{} succeeds, it means that the other branch of $p$'s subtree is not fully finished yet, and we can just quit.
For example, in \cref{fig:mis}, after we process \nodecircle{2} in round 1, we set \nodecircle{7} and \nodecircle{13} as unavailable in \nodecircle{14}'s \tastree{}.
Both of them will \TAS{} their parent and succeed, so they quit.
%In this case, the last \block{ing} vertex must be in the other branch, and thus the function terminates and will let the other branch inform $p$'s parent.
When the \TAS{} fails, it means that $p$'s flag is already \one{}, so the other branch has been fully unavailable.
Since the current subtree at $p$ is also fully unavailable, we will continue to attempt mark $p$'s parent using \TAS{} recursively.
For example, in \cref{fig:mis}, when we mark \nodecircle{12} as unavailable in \nodecircle{14}'s \tastree{}, we first \TAS{} its parent. As it was set by \nodecircle{13} previously, the \TAS{} fails.
We then continue to its parent (the root). This \TAS{} succeeds.
When a \TAS{} at the root of any \tastree{} $T_v$ fails, the entire subtree is now unavailable, so $v$ is ready to be waked up (\cref{line:mis-wakeupv}),
and we will repeat this process for $v$.
For example, when we mark \nodecircle{11} as unavailable, we \TAS{} its parent and fail, and continue to the parent, which is the root. This \TAS{} also fails, which means the entire tree is unavailable. Therefore, \nodecircle{14} can be waked up.

We note that our new MIS algorithm is fully asynchronous: no round-based synchronization is used.
Although this does not directly follow our \phaseparallel{} algorithm,
it also uses the same idea of our Type 2 framework, where we wish to identify the last finished object in $\pred(x)$ and wake up $x$ at that time.
The rank of each vertex $v$ can be viewed as the longest chain with decreasing priority starting from $v$.
We now analyze the cost of this algorithm.

\begin{theorem}
\label{thm:miscost}
  \cref{alg:mis} generates the greedy maximal independent set of $G=(V,E)$ in $O(m)$ work and $O(\log n \log \dmax)$ span \whp{}, where $n=|V|$, $m=|E|$, and $\dmax$ is the maximum degree in $G$.
\end{theorem}

\ifconference{
We show the proof in Appendix \ref{app:mis:proof}.
}
\else{
\input{mis-proof}
}
\fi
In the worst case when $\dmax=n$, the span of \cref{alg:mis} is $O(\log^2n)$, which improves previous result by a factor of $O(\log n)$.
Comparing to Prism~\cite{kaler2016executing}, our approach requires no knowledge on the number of processors and is completely in the fork-join model.
The efficiency of our algorithm comes from the simple idea \tastree{} data structure. We note that using tree-like structure to do counting is used in previous work~\cite{dwork1997contention,goodman1989efficient}, but our algorithm is different in that we make the observation that in the MIS algorithm, the information needed is \emph{whether} all blocking neighbors all finish, instead of the number of unfinished blocking neighbors. This simplifies the problem and enables better span bound.

%\myparagraph{Graph Coloring, Matching and Other Algorithms}.
%Many other algorithms can be improved using the same technique, which we discuss in Appendix \ref{app:graphcoloring}.

\ifconference{
\myparagraph{Other Algorithms}. We discuss other similar algorithms such as graph coloring and matching~\cite{ahmadi2020distributed,hasenplaugh2014ordering,azad2015distributed} in Appendix \ref{app:graphcoloring}.
}
\else{
\input{graphcoloring}
}
\fi
%Many other algorithms can be improved using the same technique, which we discuss in Appendix \ref{app:graphcoloring}. 

%% file: mis-fig.tex
\begin{figure*}
  \centering
  \vspace{-1em}
  \includegraphics[width=\textwidth]{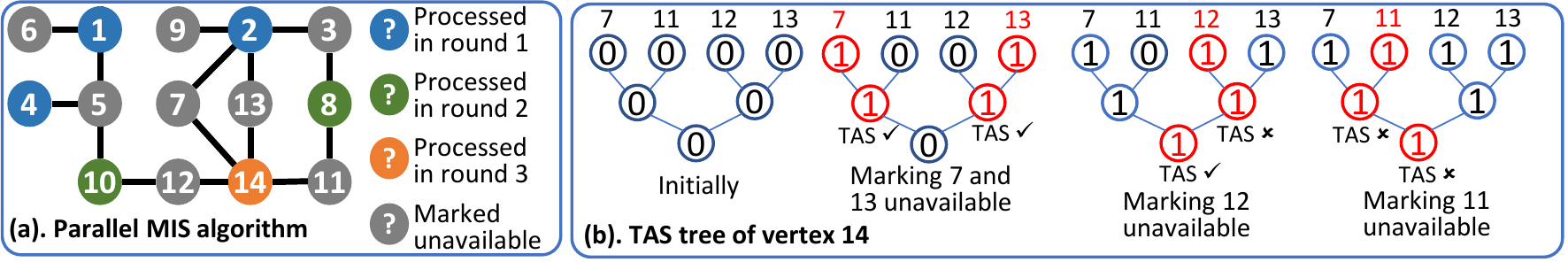}
  \caption{\small\textbf{Illustration of the greedy maximal independent set (MIS) algorithm.}
  (a). The input graph (numbers are priorities) and the round in which each vertex is processed in the algorithm.
  (b). The \tastree{} of \nodecircle{14} in the graph in (a) from the initial status to when marking some leaves unavailable.
  %$T=\true{}, F=\false{}$. \TAS{}=\tas{}.
  \cmark $=$ successful \TAS{}.
  \ding{55} $=$ unsuccessful \TAS{}. When there is an unsuccessful \TAS{} at the root,
  the entire \tastree{} is finished.
  \label{fig:mis}}
  \vspace{-.2in}
\end{figure*} 

%% file: mis-code.tex
\begin{algorithm}[t]
\small
\caption{The parallel MIS algorithm.\label{alg:mis}}
\SetKwFor{parForEach}{parallel\_for\_each}{do}{endfor}
\KwIn{A graph $G=(V,E)$ with priority $p:V\mapsto \Z$.}
\KwOut{Maximal Independent Set of $G$}
\SetKwProg{MyStruct}{Struct}{ contains}{end}
\SetKwProg{MyFunc}{Func}{}{end}
\SetKw{Break}{Break}
\newcommand{\Int}{\KwSty{int}}
\newcommand{\A}{\KwSty{A}}
\DontPrintSemicolon
\smallskip
\parForEach{$v\in V$}{
  $T_v\gets $ the \tastree{} for $v$\\
  Maintain the list of \tastree{s} that contains each vertex\\
  status[$v$] $\gets$ \texttt{undecided}
}
\parForEach{$v\in V$}{
  \lIf {$T_v$ is empty}{\mf{WakeUp}($v$)\tcp*[f]{no blocking neighbor}}
}
\Return {all $v\in V$ marked as \texttt{selected}}

\smallskip
\MyFunc{\upshape\mf{WakeUp}($v$)}{
  %Mark $v$ as in MIS;\\
  status[$v$] $\gets$ \texttt{selected}\\
  \parForEach{$u\in \neighbor(v)$\label{line:mis-for-neighbor}}{
      status[$u$] $\gets$ \texttt{removed}\\
      \parForEach{\tastree{} $T_{v}$ contains $u$\label{line:mis-for-tastree}}{
        \If {status[$v$] $\ne$ \texttt{removed}} {
            $x\gets $ the leaf of $u$ in $T_{v}$\\
            $x.\flag\gets$ \true\\
            $p\gets \parent(x)$\\
            \While {\texttt{test\_and\_set}$(p.\flag)$ successful\label{line:mis-tas-while}} {
              \If{$p$ is the root of $T_u$\label{line:mis-to-root}} {
                \mf{WakeUp}($v$)\label{line:mis-wakeupv}\\
                \Break{}
              }
              $p\gets \parent(p)$ \\
            }
        }
      }
  }
}
\end{algorithm}

%% file: mis-proof.tex
\begin{proof}
 Fischer and Noever~\cite{fischer2018tight} showed that if the priorities are assigned uniformly at random, the longest increasing chain of priority in the graph is $O(\log n)$ \whp{}.
 Generating the priorities takes $O(n)$ work and $O(\log n)$ span \whp~\cite{blelloch2020optimal}.
 When processing each vertex, each parallel for-loop (\cref{line:mis-for-neighbor}) forks off $O(\dmax)$ subtasks and incurs $O(\log \dmax)$ span.
 When removing $u$ from a \tastree{} (the loop body from Line \ref{line:mis-for-neighbor}),
 both the parallel-for-loop and the while-loop incur $O(\log \dmax)$ span.
 Combining the costs together gives the $O(\log n \log \dmax)$ span bound \whp{}.

 For work, we focus on the part to update the flags in the \tastree{s}, since other parts trivially uses $O(m)$ work.
 Each \tastree{} $T_v$ has $\degree(v)$ nodes.
 At most two \tas attempts apply to each node, one from each branch.
 Hence, the total cost on the \tastree{} $T(v)$ is $O(\degree(v))$, which add up to $O(m)$ for all $v\in V$.

 %If we keep two copies of each undirected edge, we assume
 Note that our algorithm assumes we store an undirected edge for each end point and know the correspondence
 (for $v\in \neighbor(u)$, we need to locate $u$ in $v$'s \tastree{}).
 If not, we can use semisort or hash table~\cite{GSSB15,blelloch2020optimal}, but that makes work bound $O(m)$ in expectation.
\end{proof}

%% file: graphcoloring.tex
\myparagraph{Graph Coloring and Matching}.
Several iterative graph algorithms that share the similar approach can be improved using the same technique.
For graph coloring, Jones and Plassmann~\cite{jones1993parallel} showed the greedy algorithm that can be parallelized using the similar approach in~\cite{BFS12}.
Hasenplaugh et al.~\cite{hasenplaugh2014ordering} analyzed a list of heuristics that vary the greedy order of the vertices and can lead to different span bounds and output quality.
For graph matching, Blelloch et al.~\cite{BFS12} showed a parallel greedy algorithm that is very similar to the MIS algorithm in the same paper.
%By replacing the original wake-up strategy in~\cite{BFS12} in these algorithms, we can improve the span by $O(\log n)$.
By replacing the original wake-up strategy in~\cite{BFS12} with our new approach, we can improve the span by $O(\log n)$.
We note that the parallel graph-matching algorithm cannot be fully asynchronous since each edge's readiness relies on two vertices, which needs to be checked after synchronization.
Hence, a synchronization is required between the rounds, but that does not change the span bound.
The analysis by Hasenplaugh et al.~\cite{hasenplaugh2014ordering} assumes atomic decrement-and-fetch operations which is usually not included on the family of the binary-forking models~\cite{blelloch2020optimal}.
Using the technique in this paper can achieve the same bounds without this assumption.

\myparagraph{Other Algorithms}.
Many algorithms in~\cite{blelloch2016parallelism,blelloch2018geometry,blelloch2020optimal,blelloch2020randomized} have constant size $\mathcal{P}(x)$ for all $x\in S$.
Blelloch et al.~\cite{blelloch2020optimal} showed that \tas can be used to check the readiness in this case.
Here we note that this can be considered as a special case for our \tastree{s}, just with constant sizes.
These applications include random permutation, list ranking, and tree contraction~\cite{shun2015sequential,blelloch2020optimal}; convex hull and Delaunay triangulation~\cite{blelloch2016parallelism,blelloch2018geometry,blelloch2020randomized}.
Although we do not improve the bounds of these algorithms, we show the interconnections among all these algorithms, and an additional angle to review these problems and algorithms.

%% file: exp.tex
\section{Experiments}
\label{sec:exp}
\begin{figure*}
	\centering
\small
\vspace{-1.5em}
	\begin{minipage}[h]{0.48\textwidth}
    \begin{tabular}{c@{}c}
      \includegraphics[width=.5\columnwidth]{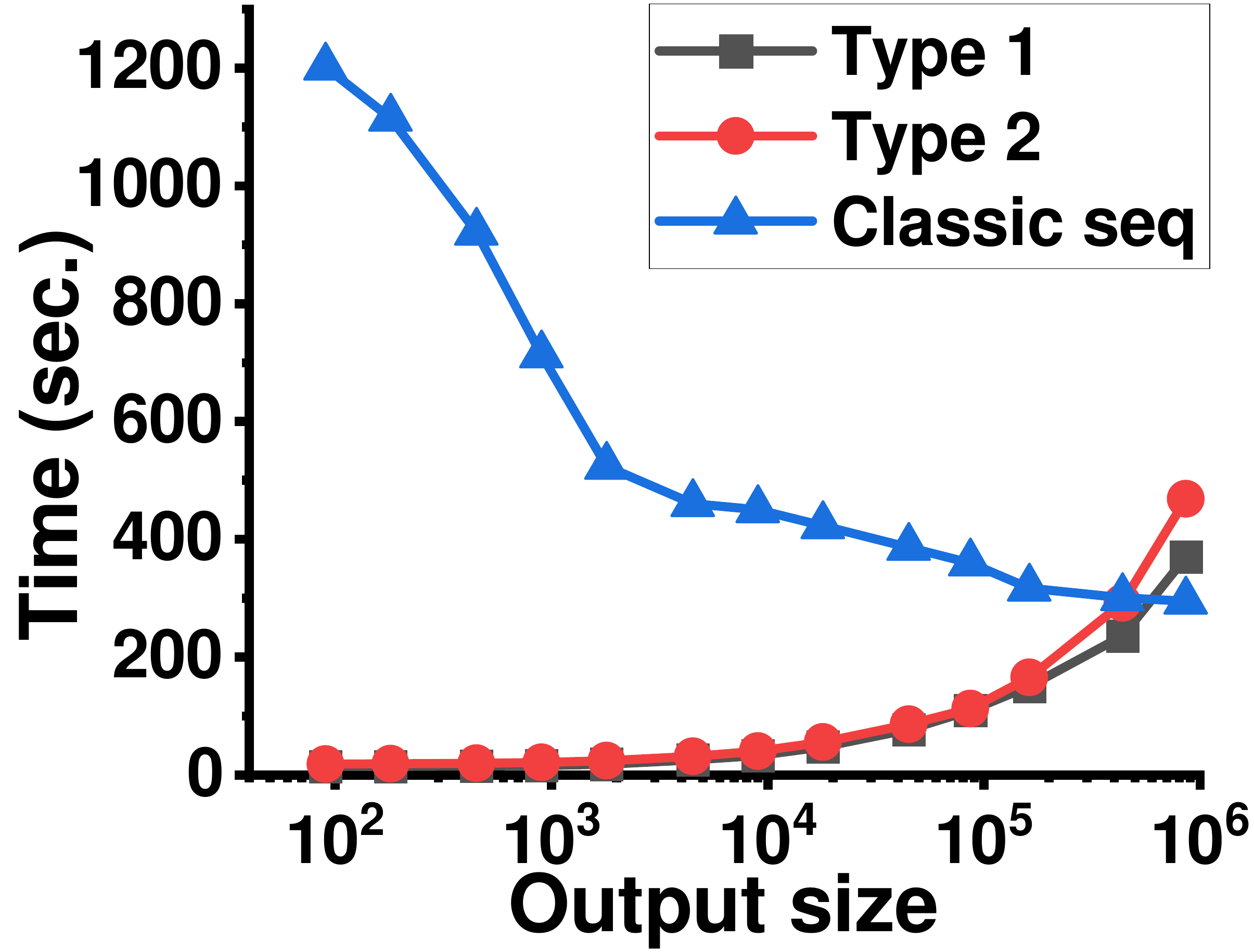} &
       \includegraphics[width=.5\columnwidth]{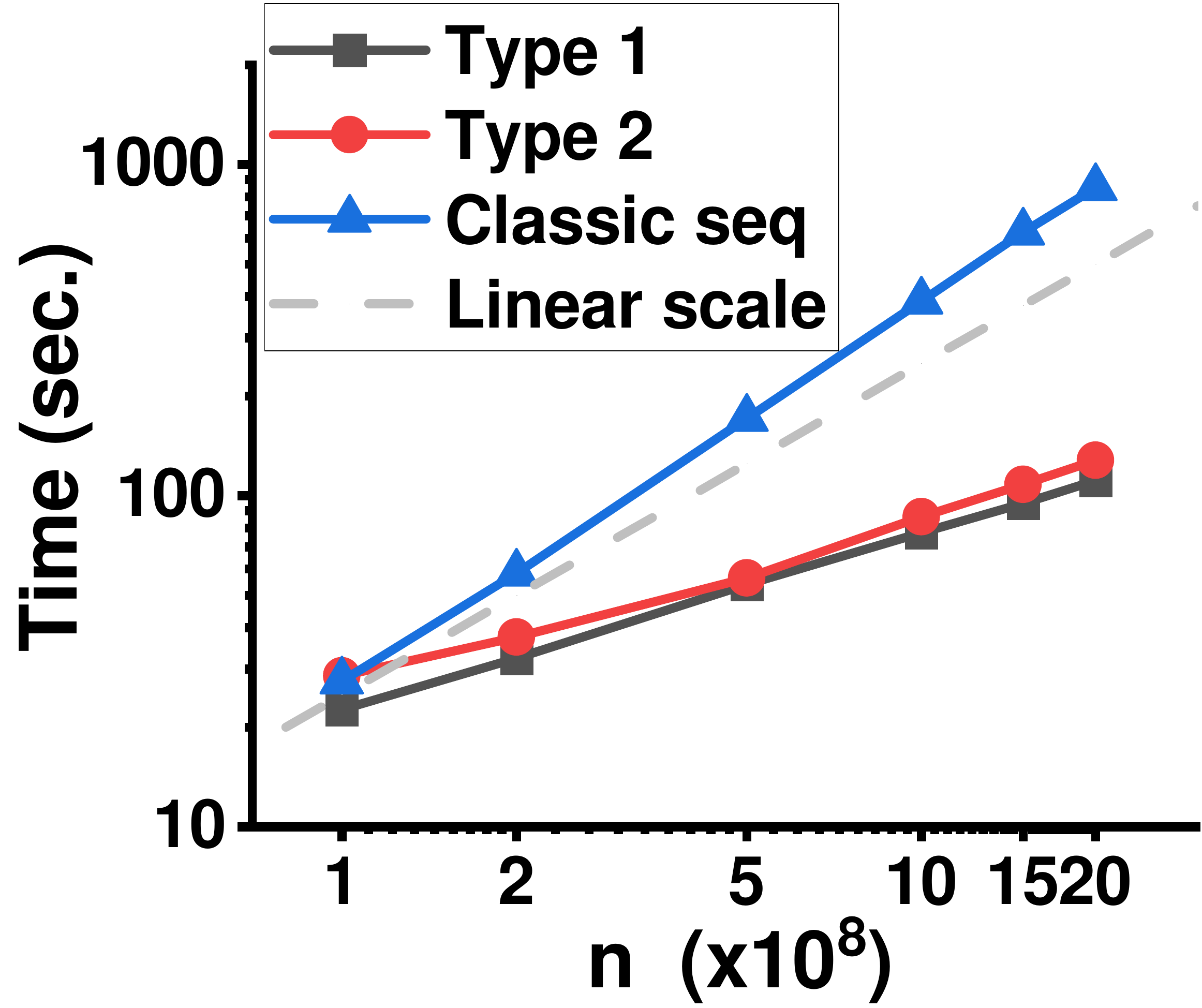}\\
      (a). $n=10^9$, varying rank & (b). $r=45000$, varying input size \\
    \end{tabular}
    \caption{\small\textbf{Experiments on activity selection.}
    (a). fix input size $n=10^9$ and vary the rank of the input.
    (b). fix rank $r=45000$ and vary the input size.
    ``Classic seq'' is the classic sequential DP algorithm.
    %``Linear scale'' line with $y=kx$ showing that both types of our parallel algorithm outperform the sequential DP as the input size grows.
    ``Linear scale'' shows the slope of linearly growing line.
    \label{fig:exp:activity}
    }
	\end{minipage}\hfill
	\begin{minipage}[h]{0.48\textwidth}
    \includegraphics[width=\columnwidth]{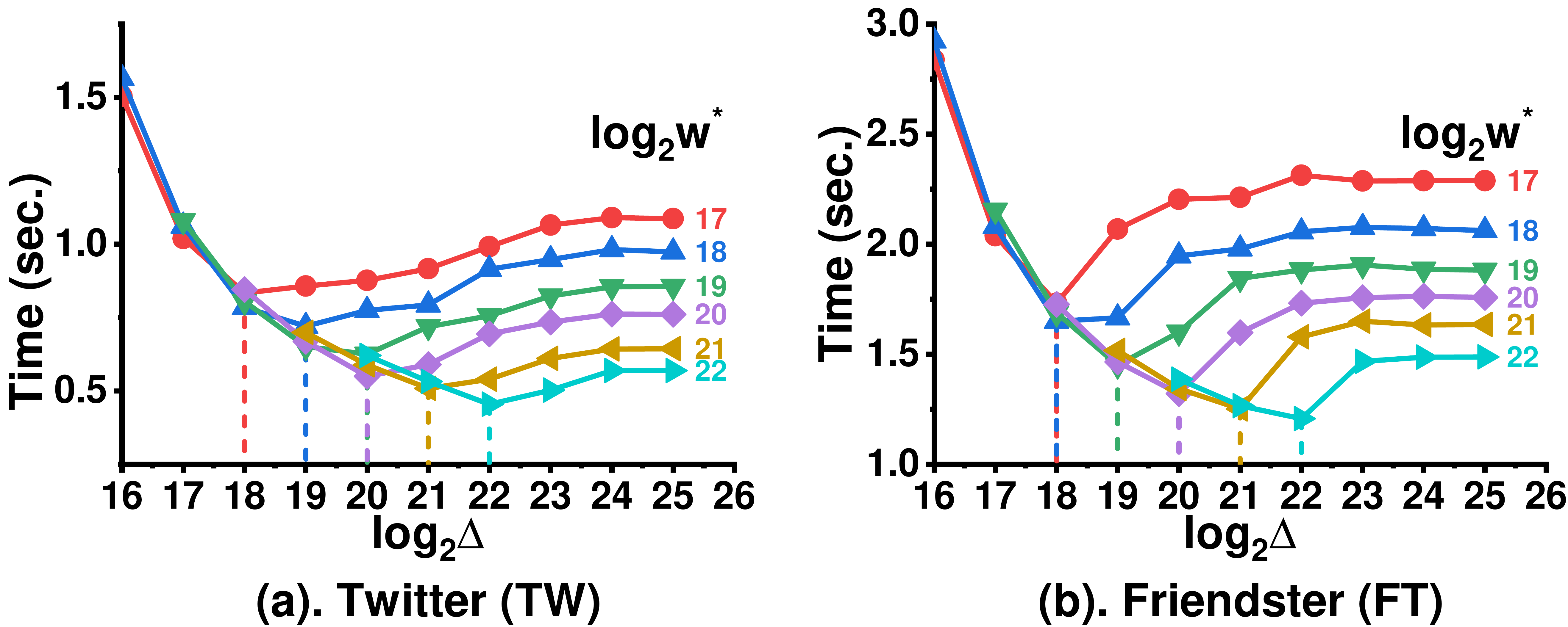}
    \caption{\small\textbf{Experiments on parallel SSSP.}
    (a). Running time on graph Twitter (41.7 millions of vertices and 1.47 billions of directed edges).
    (b). Running time on graph Friendster (65.6 millions of vertices and 3.61 billions of undirected edges).
    We use $\Delta$-stepping implementation from~\cite{rhostepping} with different values of $\Delta$ and $w^*$ (the smallest edge weight in the graph).
    Values on the right are the values of $\log_2 w^*$. The maximum edge weight is $2^{23}$.
    \label{fig:exp:sssp}
    }
	\end{minipage}
\vspace{-.2in}
\end{figure*}

%\subsection{Experimental Setup}
%\myparagraph{Setup.}
In addition to the new theoretical results,
we also implemented many proposed algorithms based on our \phaseparallel{} frameworks, including activity selection (both Type 1 and Type 2), Huffman tree, SSSP, and LIS.
We use our experiments as proofs-of-concept to show how work-efficiency and round-efficiency affect performance in practice.
We run our experiments on a 96-core (192 hyperthreads) machine with four Intel Xeon Gold 6252 CPUs, and 1.5 TiB of main memory.
Our implementation is in C++ with Cilk Plus~\cite{Cilk11}.
%, and compiled with flags \textsf{-O3 -fcilkplus -march=native}.
For the parallel results, we use all cores and fully interleave the memory among NUMA areas using \textsf{numactl -i all}.
We use $r$ as the input rank.
%All implementations are based on our phase parallel frameworks.
%For all experiments we did 6 repeated tests. We report the average of the last five tests.
All reported numbers are the averages of the last five runs among six repeated experiments.
\ifconference{
Due to page limitation, we put experiments and discussions about LIS and Huffman tree in Appendix \ref{app:exp}.
}
\fi

%We run all experiments on a single server with 96 cores (192 threads) and 1536GiB DRAM. The detailed system configuration is listed in the Table \ref{eval:config}.

\hide{\begin{table}[hbtp]
\caption{System configuration}
\centering
\footnotesize
\begin{tabular}{| c | c |}
%\hline
%Component & Specification \\
\hline
CPU & Four-way Intel$^\circledR$ Xeon$^\circledR$ 6252 24c48t@2.80GHz boost on all cores\\
\hline
 & 32KiB L1-d/i, 1MiB L2, 35.75MiB L3 cache in each socket\\
\hline
Memory & 48*32GiB, DDR4 2933MT/s\\
\hline
% SSD & 960GB, ???\\
% \hline
OS & CentOS 8 (Linux 4.18.0-240.22.1.el8\_3.x86\_64)\\
\hline
CompilerToolKit & GCC 7.5.0 (with builtin Cilk Plus)\\
\hline
\end{tabular}
\label{eval:config}
\end{table}
}

%\zheqi{As shown in Theorem \ref{them:activity_1}, \ref{them:activity_2} and \ref{them:lis}, the rank of input set relates to the span and it affects the performance of our implementation. In our experiments, the input data is generated as follows in order to control the rank and size respectively and we could thus show the results under the different settings.}

\subsection{Activity Selection}
%We implement the parallel activity selection in both Type 1 and 2 and the parallel LIS in Type 2.
We implement Type 1 and Type 2 algorithms for activity selection.
We also implement a sequential version based on the DP recurrence in \cref{eqn:activitydp} for comparison.
%We test input size from $10^8$ to $10^9$.
For each activity, we set a random start time and a length based on a truncated normal distribution.
%To understand the performance for different rank of input data, we control the rank (maximum number of non-overlapping activities) of the input set.
%Since the span is affected by the rank of the input,
We control the mean and standard deviation of this distribution to control the rank of the input data. The weights are generated uniformly at random in $[1,2^{32})$.

\cref{fig:exp:activity}(a) shows the running time of all tested algorithms on $10^9$ input activities with various ranks.
Both our Type 1 and Type 2 algorithms have very similar performance and outperform the classic sequential algorithm up to rank of $4\times 10^6$.
Note that the running time of our algorithm increases as the rank increases because our algorithms have span proportional to the rank.
Also, when the rank is large, each round in the algorithm only deals with a small number of objects, which also harms parallelism.
Even so, the running time of our algorithm seems to grow sublinearly with the rank.
This is because with about 200 threads,
the work should still dominate the cost, and both our algorithms are work-efficient.
%our algorithm is work-efficient,
%and the ratio of work to the number of processors
Interestingly, the performance of the classic sequential algorithm improves with increasing input rank.
This is because in the sequential algorithm, when the rank is large,
the range query of an activity $x$ (see \cref{eqn:activitydp}) will likely find an activity close to $x$, which exhibits better cache locality.
For small ranks, our algorithms can be up to 80x faster than the sequential algorithm.
For rank of $O(\sqrt{n})$, our algorithm is still about 14x than the classic sequential algorithm.

\cref{fig:exp:activity}(b) shows the running time of all tested algorithms on different input sizes with a fixed rank of $r=45000$.
For other values of the rank, we see similar trends, so we just show $r=45000$ as an example.
Our algorithm scales well to large input sizes.
The sequential algorithm grows superlinearly with the input size $n$, which matches the theoretical cost $O(n\log n)$.
As shown in \cref{fig:exp:activity}(b), our algorithm grows much slower with $n$ (almost linearly).
This is because when $n$ increases, the number of activities to be processed within each round also increases, which overall enables better parallelism.
We believe that this indicates the potential of our algorithm to scale to even larger data and more threads.

On all tests, Type 1 algorithm outperforms the Type 2 algorithm by up to 35\%. This is because in the Type 2 algorithm, we need to first find
pivots for all activities, while in the Type 1 algorithm, we directly start with processing all ready activities, and use a range query to find the frontier.
Nevertheless, the two algorithms still have very similar performance and both outperform the classic sequential algorithm for reasonably large input rank.

\hide{
\zheqi{
\subsubsection{Activity selection}

The input of activity selection consist of $n$ tuples where each tuple $(s_i,e_i,w_i)$ denotes the start time, end time, and weight of $i$-th activity, respectively. First, we generate $s_i\forall i$ randomly in the uniform distribution. Next, we get $e_i$ by deciding the activity duration $e_i-s_i$. In real life, activities such as taking classes usually have a minimum duration but may be extended in some cases. Thus, we set a minimum duration $b$ for all activities and generate an extended length in certain distribution for each one. In our experiments, such length are in truncated normal distribution and the end time is $e_i=s_i+b+|\mathcal{N}(0,\sigma^2)|$. The activity weights are randomly generated since the order is based on comparison of start and end time and not affect by the weights.

\hide{Activity selection.
In our experiments, the input data is generated in this way. For input of weighted activity selection problem, we first specify the range of task time and partition the whole timeline into even $N$ segments. In each segment, supposing starting at $t_s$ and ending at $t_e$, we will generates a randomized integral
}

}
}
\iffullversion{
\subsection{Huffman Tree}
\input{huffmantree-exp}
}
\fi
\subsection{Parallel SSSP}
\myparagraph{Background.} SSSP is a challenging problem in the parallel setting.
As mentioned, Dijkstra is work-efficient but hard to parallelize.
Bellman-Ford has better parallelism but significantly more work.
Even so, parallel Bellman-Ford has good performance on low-diameter graphs such as social networks~\cite{rhostepping,shun2013ligra} due to better parallelism.
In practice, many heuristics were proposed aiming to achieve a tradeoff between work and parallelism.
For example, the $\Delta$-stepping algorithm~\cite{meyer2003delta} %is a hybrid of Dijkstra and Bellman-Ford.
determines the correct shortest distances of vertices in increments of $\Delta$ of the tentative distances.
In step $i$, the algorithm will find and settle down all the vertices with distances in $[i\Delta, (i+1)\Delta]$.
%Within each step, the algorithm runs Bellman-Ford as substeps, until no more vertices in the given distance range gets updated.
$\Delta$-stepping is highly practical and widely used. However, its performance is very sensitive to the parameter $\Delta$~\cite{rhostepping}.

\myparagraph{Our experiments.} As mentioned in \cref{sec:sssp}, our \phaseparallel{} algorithm settles down all vertices within
tentative distance $[iw^*,(i+1)w^*]$, where $w^*$ is the smallest edge weight. This is conceptually the same as using $\Delta=w^*$
in $\Delta$-stepping~\cite{meyer2003delta}. Therefore, we run experiments using the $\Delta$-stepping implementation by Dong et~al.~\cite{rhostepping} to test our idea.
We note that this is not exactly the same as our algorithm as the implementation does not use a tree-based data structure to extract the frontier\footnote{In fact,
almost none of the parallel SSSP implementation uses tree-based structures to maintain distances due to their worse cache locality than flat arrays.},
but is still ``work-efficient'' w.r.t. the number of total relaxations. We note that empirically,
the I/O cost in processing and relaxing edges is usually the main cost in practical parallel SSSP algorithms.
Our results highly match our theory. We tested two graph benchmarks, Twitter~\cite{kwak2010twitter} and Friendster~\cite{yang2015defining}.
Both of them are real-world large-scale social networks with small diameters.

In our experiments, we fix the largest edge weight as $w_{\mathit{max}}=2^{23}$, vary the $w^*$ from $2^{17}$ to $2^{22}$, and set the weight uniformly at random in this range for each edge\footnote{This setting is similar to \emph{weighted BFS}~\cite{dhulipala2017} (which is trivially work-efficient). The difference is that after normalizing, the edge weight in our problem are not necessarily integers as in weighted BFS.}.
For each edge weight range, we run $\Delta$-stepping with $\Delta$ varying from $2^{16}$ to $2^{26}$. We show the running time in \cref{fig:exp:sssp}.
For both graphs, the best choice of $\Delta$ almost exactly matches $w^*$ (differ by at most 2x), when $w^*$ is close to $w_{\mathit{max}}$.
This verifies the importance of work-efficiency in practical SSSP algorithms. When $w^{*}$ gets even smaller, using $\Delta=w^*$ does not perform well, which
is also as expected---despite work-efficiency, using a small $\Delta$ limits the frontier size, hence we cannot fully exploit parallelism. This also reveals the work-parallelism (or work-round) tradeoff. In all, when $w_{\mathit{max}}/w^*$ is within 32, using our algorithm (i.e., $\Delta=w^*$) gives reasonably good performance.

It is worth noting that we also tried the same algorithm on large-diameter graphs, such as some road graphs~\cite{roadgraph}.
Probably not surprisingly, even on $w^*=w_{\mathit{max}}/2$, $\Delta=w^*$ did not give the best performance. This is because on such graphs,
the frontier size is usually small, and the performance is usually limited by the lack of parallelism.
In this case, avoiding extra work does not help much in improving performance (and even harms the performance since the parallelism gets worse).
Many state-of-the-art implementations~\cite{rhostepping,zhang2018graphit} optimize performance on such (large-diameter) graphs by sacrificing more work to get better parallelism.

\subsection{Parallel LIS}\label{sec:exp-lis}
\ifconference{
We also implement our LIS algorithm. Due to page limitation, we present a short discussion here and show our results and more details in \cref{app:exp:lis}.
We test input data of $10^8$ with different ranks (LIS sizes). We also implement a standard sequential algorithm based on \cref{eqn:lisdp}.
Our algorithm is faster than the sequential algorithm up to rank $r=100$ and performs competitive but slower for large ranks.
Note that our algorithm has an $O(\log^2 n)$ overhead in work.
When the parallelism is not sufficient to compensate for the overhead in work, the performance may drop.
It is worth noting that our algorithm still gets good parallelism---for $r=O(\sqrt{n})$, the self-speedup is over~40.
We believe that our algorithm could scale (and show advantage) to more cores in the future.
We are also interested in improving the work bound.
We believe that if we can shave off an $O(\log n)$ term for work, then our algorithm can likely be faster on most of the input cases.
}
\else{
\input{lis-exp}
\fi

%% file: huffmantree-exp.tex
We present our results of Huffman tree construction in \cref{fig:huffman}. We test different input distributions (uniform, Zipfian, and exponential) and input sizes ($10^5$ to $10^9$). We also implement a standard sequential algorithm for comparison. For the sequential algorithm, 
we implement the version which costs $O(n)$ work after sorting.
To guarantee that the root (the total) frequency is within 64 bits, the frequency is set in $[1,2^{32})$. We change the maximal frequency to control the rank (Huffman tree height) of the input.

We show how running time changes with the number of rounds in \cref{fig:huffman}(a), and how running time changes with the input size in \cref{fig:huffman}(b).
In \cref{fig:huffman}(a) we set $n=10^9$, and change parameters in uniform distribution and exponential distribution to get different number of rounds. Note that we used Zipfian distribution as a fixed distribution, and thus we cannot change its parameters to control the number of rounds. In \cref{fig:huffman}(b) we fix the maximum frequency to be $1000$ for uniform and exponential distributions, and also show the running time for the sequential algorithm for comparison.
As mentioned, the maximum total frequency is within 64-bit integer. This means the height of the tree (and thus the number of rounds) cannot be more than 64. Just as claimed in \ifconference{\cref{app:huffman}}\else{\cref{sec:huffman}}\fi, although the worst-case span can be $\Theta(n)$, in practice, the precision of machines will limit the actual tree height in a small value. This actually enables very good parallelism.
Indeed, we noticed that the performance on different distributions, and even with different number of rounds, are very similar. This is because for all of the tests, the number of rounds is a small value between 30 to 60. Therefore, in every round parallelism is fully exploited. On sufficiently large data, our implementation is 10-20x faster than the sequential algorithm.
\begin{figure}[t]
	\centering
	%\begin{minipage}[h]{0.48\textwidth}
   \includegraphics[width=\columnwidth]{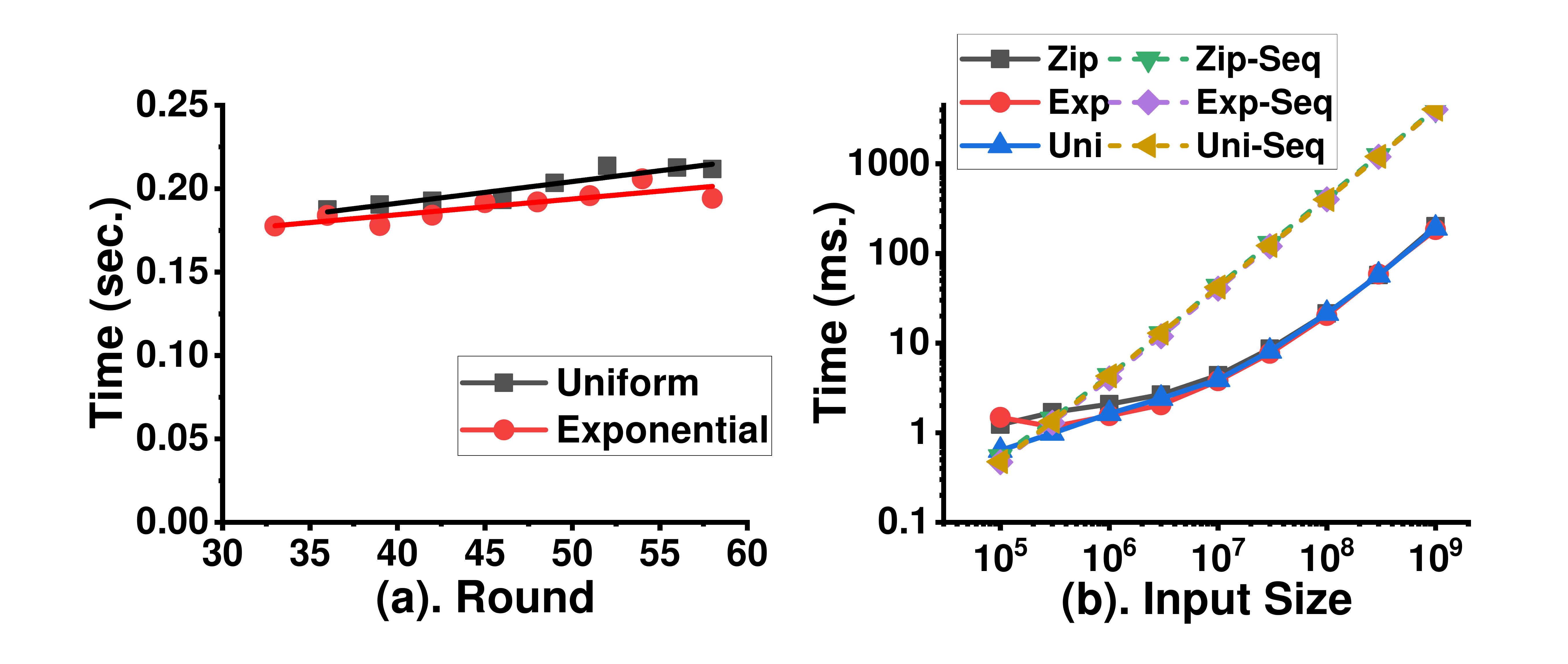}
	\caption{\small
   (a). Fixed input size as $n=10^9$, show the running time with the round varying from $33$ to $58$.
   (b). Fixed maximum frequency as $1000$, show the running time with the input size varying from $10^5$ to $10^9$.
   \label{fig:huffman}
   }
	%\end{minipage}
\end{figure} 

%% file: lis-exp.tex
We also test our LIS algorithm. To improve the performance, we use nested arrays to represent augmented range trees to improve locality. We also use a heuristic when choosing pivots---instead of choosing a uniformly random pivot, we choose the right-most unfinished point as the pivot. This is based on an intuition such that points to the right are more likely to be processed in later rounds. We also implemented a standard $O(n\log n)$ sequential version.

\begin{figure}[t]
	\centering
	\small
	\begin{minipage}{0.45\textwidth}
		\begin{tabular}{c@{}c}
			\includegraphics[width=.45\columnwidth]{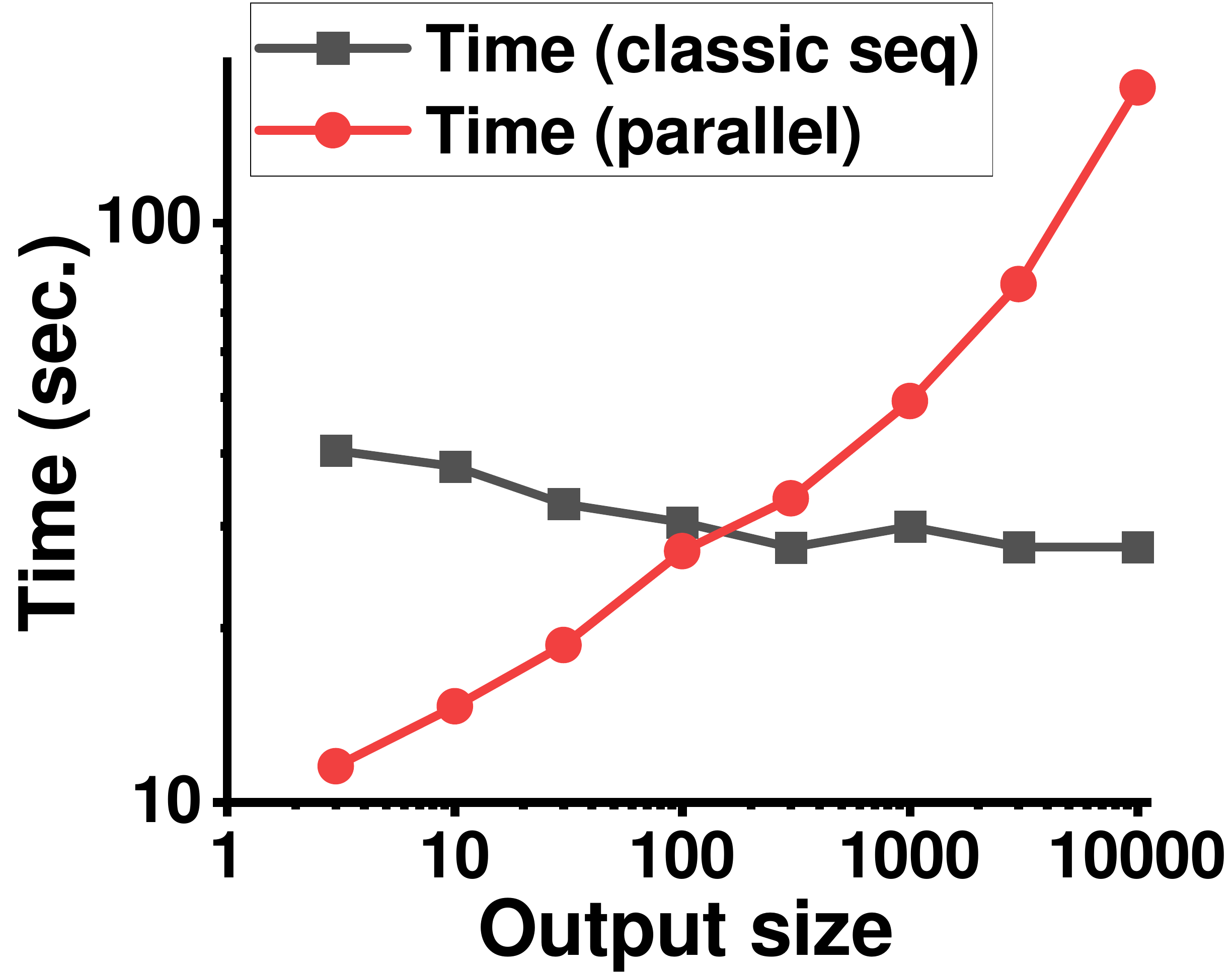} &
			\includegraphics[width=.5\columnwidth]{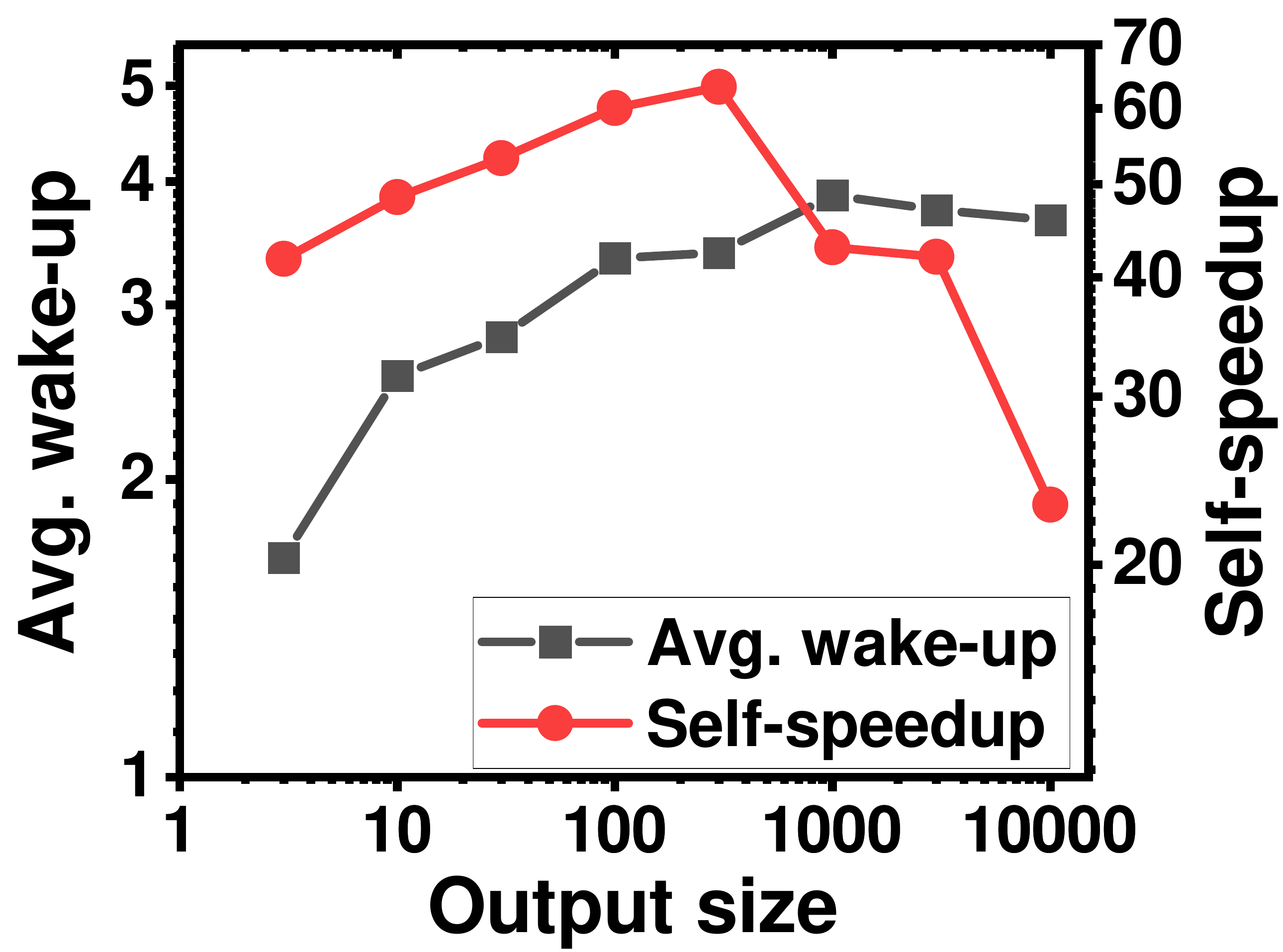}\\
			(a). Time & (b). Average wake-up and self-speedup  \\
		\end{tabular}
		\caption{\small\textbf{Experiments on LIS (segment).}
			We fix input size $n=10^8$ and use the \textit{segment pattern}, with varying the output size.
			``Classic seq'' is the classic sequential DP algorithm.
			\label{fig:exp:lis_segoffset}
		}
	\end{minipage}\hfill
	% \vspace{-.2in}
	\\
	\begin{minipage}{0.45\textwidth}
		\begin{tabular}{c@{}c}
			\includegraphics[width=.45\columnwidth]{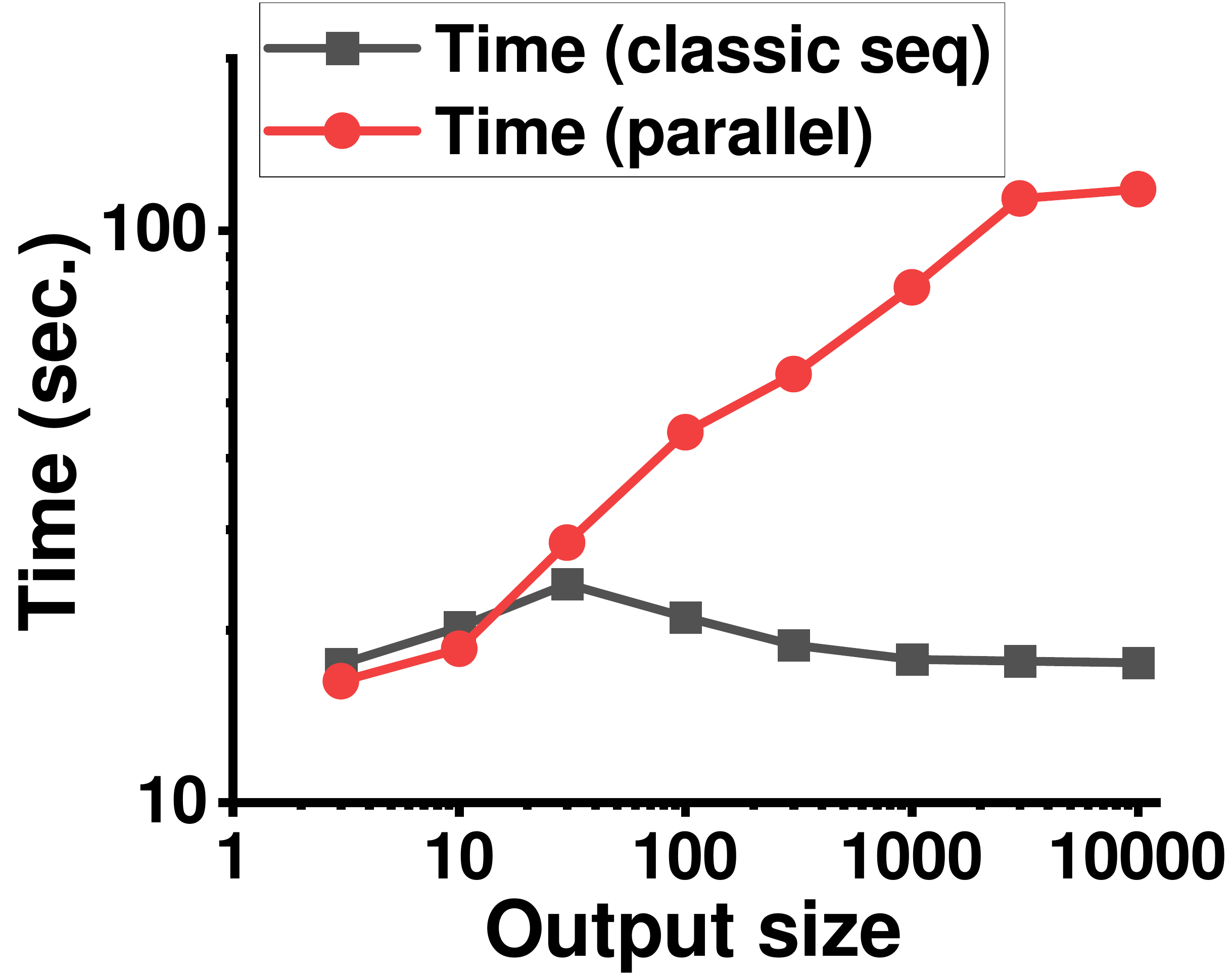} &
			\includegraphics[width=.5\columnwidth]{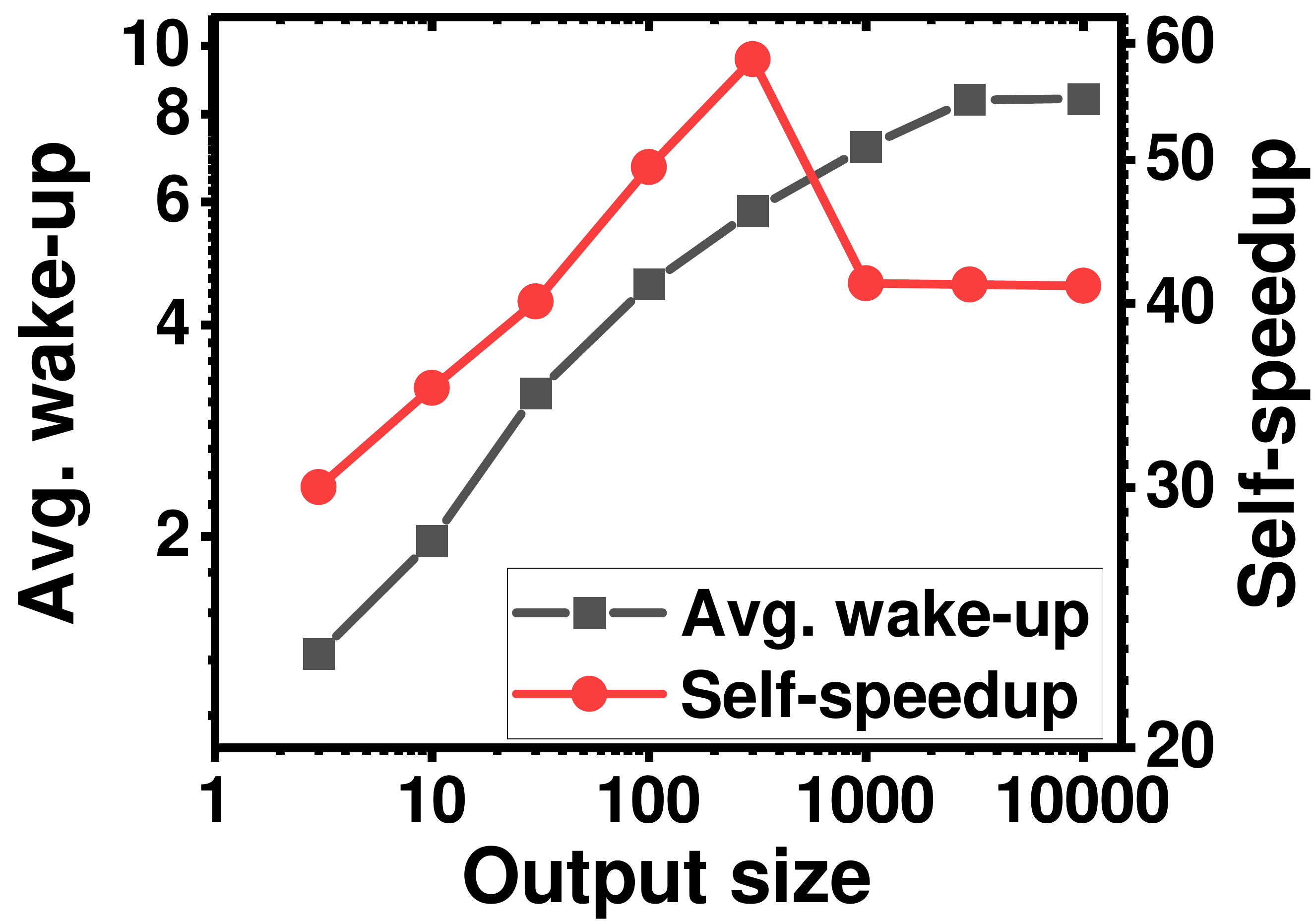}\\
			(a). Time & (b). Average wake-up and self-speedup  \\
		\end{tabular}
		\caption{\small\textbf{Experiments on LIS (line).}
			We fix input size $n=10^8$ and use the \textit{line pattern}, with varying the output size.
			``Classic seq'' is the classic sequential DP algorithm.
			\label{fig:exp:lis_line}
		}
	\end{minipage}\hfill
	%\vspace{-.2in}
\end{figure}

We present results for LIS implementations in \cref{fig:exp:lis_segoffset} and \cref{fig:exp:lis_line}. We also vary the input rank, which is the LIS size.
We use two different data patterns (see \cref{fig:lis_pattern}). The first one is roughly $k$ segments of data, and we call it the \emph{segment pattern}. Within each segment the data values are roughly decreasing (we also added some random noise), and across the segments, the values are roughly increasing. The LIS size is about $k$. The other pattern is generated by using an increasing line $a_i=t\times i+b_i$, where $b_i$ is a random variable choosing from a uniform distribution. We call it the \emph{line pattern}. By changing the slope $t$ and the distribution of $b$, we can also control the rank of the input data.

We show our running time as well as the average number of wake-up attempts for all objects. On the two different data patterns, our algorithm is faster up to when the rank $r=300$, and perform worse than the sequential algorithm afterwards. Interestingly, similar to the activity selection, our algorithm is getting slower as the rank increases (which matches theory), but the standard sequential algorithm is getting faster. We believe the improvement in the sequential algorithm is also caused by better locality because the range query (in \cref{eqn:lisdp}) for an object $x$ will find an object close to $x$. The detailed experiment data is shown in \cref{tab:exp:lis}.

We note that our algorithm still show very good scalability---in most tests, our self-speedup is more than 40x. The poor performance when the rank is large comes from the work-inefficiency. Although the overhead is polylogarithmic, it can still be large ($\log^2 n$ is much more than the number of available processors on our machine). Indeed, the sequential running time (on one core) of our algorithm is much more than the standard sequential algorithm. Therefore, when $k$ is large, the parallelism cannot compensate the overhead due to work-inefficiency.

It is also worth noting that the average number of wake-ups is very small. In all our tests, the maximum value is 8 times, which is less than $\log n$ shown in \cref{lem:lislogn}. This is partially enabled by our heuristic. Especially for
the segment pattern, when the pivot is chosen as the right-most unfinished object, it is almost always the last blocking object to wait.

We believe our algorithm are scalable to more cores, but we are also interested in improving the work bound to closer to work-efficient. Reducing work-bound should be promising to improve practical performance.

\begin{table}[t]
\centering
\small
\begin{tabular}{rrrrrc}
\multicolumn{6}{c}{\textbf{The segment pattern}}\\
\hline
\multicolumn{1}{c}{Output}  &  \multicolumn{1}{c}{Classic} & \multicolumn{1}{c}{Ours} & \multicolumn{1}{c}{Ours} & \multicolumn{1}{c}{Spd.} & \multicolumn{1}{c}{Average \# of}\\
\multicolumn{1}{c}{Size}  &  \multicolumn{1}{c}{seq.}  & \multicolumn{1}{c}{seq.} & \multicolumn{1}{c}{par.} & \multicolumn{1}{c}{(Our $\frac{\text{seq.}}{\text{par.}}$)} & \multicolumn{1}{c}{Wake-ups}\\
\hline
3 &  40.49  & 482  &  11.54  & 41.78 & 1.67\\
%\hline
10  &   37.96 & 709 & 14.64 & 48.47 & 2.55\\
%\hline
30  &   32.78 & 994 & 18.67 & 53.26 & 2.79\\
%\hline
100 &  30.42  & 1630 & 27.13 & 60.08 & 3.35\\
%\hline
300 &  27.52  & 2116 & 33.49 & 63.20 & 3.39\\
%\hline
1000  &   29.96 & 2850 & 49.24 & 42.98 & 3.87\\
%\hline
3000  &   27.58 & 3292 & 78.44 & 41.97 & 3.74\\
%\hline
10000 &  27.59  & 3960 & 171.25  &  23.12  & 3.66\\
\hline
\end{tabular}
\\
\begin{tabular}{rrrrrc}
\multicolumn{6}{c}{\textbf{The line pattern}}\\
\hline
\multicolumn{1}{c}{Output}  &  \multicolumn{1}{c}{Classic} & \multicolumn{1}{c}{Ours} & \multicolumn{1}{c}{Ours} & \multicolumn{1}{c}{Spd.} & \multicolumn{1}{c}{Average \# of}\\
\multicolumn{1}{c}{Size}  &  \multicolumn{1}{c}{seq.}  & \multicolumn{1}{c}{seq.} & \multicolumn{1}{c}{par.} & \multicolumn{1}{c}{(Our $\frac{\text{seq.}}{\text{par.}}$)} & \multicolumn{1}{c}{Wake-ups}\\
\hline
3 & 17.46 & 489  &  16.30  & 30.03 & 1.36\\
%\hline
10  &  20.27  & 650  &  18.57  & 35.04 & 1.97\\
%\hline
30  &  24.11  & 1141 & 28.44  &  40.11  & 3.20\\
%\hline
100 & 21.10  &  2194  &  44.37 &  49.44  & 4.58\\
%\hline
300 & 18.84 & 3281 & 56.07  &  58.51  & 5.82\\
%\hline
1000  &  17.81  & 4053 &  79.50 &  41.27 &  7.18\\
%\hline
3000  &  17.68  & 4683 & 113.72 & 41.18 & 8.39\\
%\hline
10000  &  17.56  & 4852 & 118.06 & 41.10 & 8.41\\
\hline
\end{tabular}
\caption{\label{tab:exp:lis} \small
\textbf{The detailed experiment data on LIS.}
All times are in seconds.
}
\end{table}

\begin{figure*}[ht]
 \centering
\small
 \begin{minipage}{0.92\textwidth}
   \begin{tabular}{cccc}
     \includegraphics[width=.23\columnwidth]{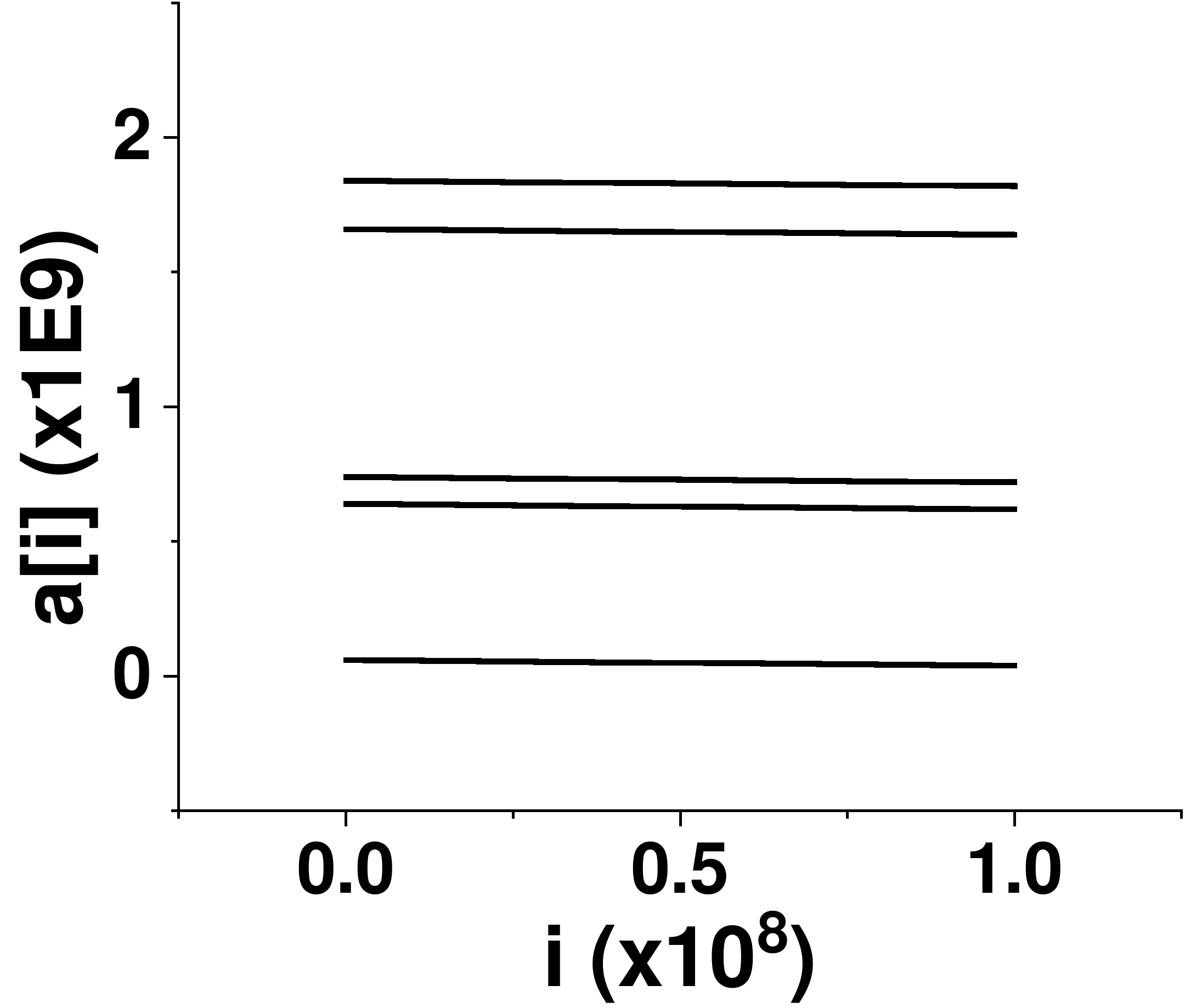} &
      \includegraphics[width=.223\columnwidth]{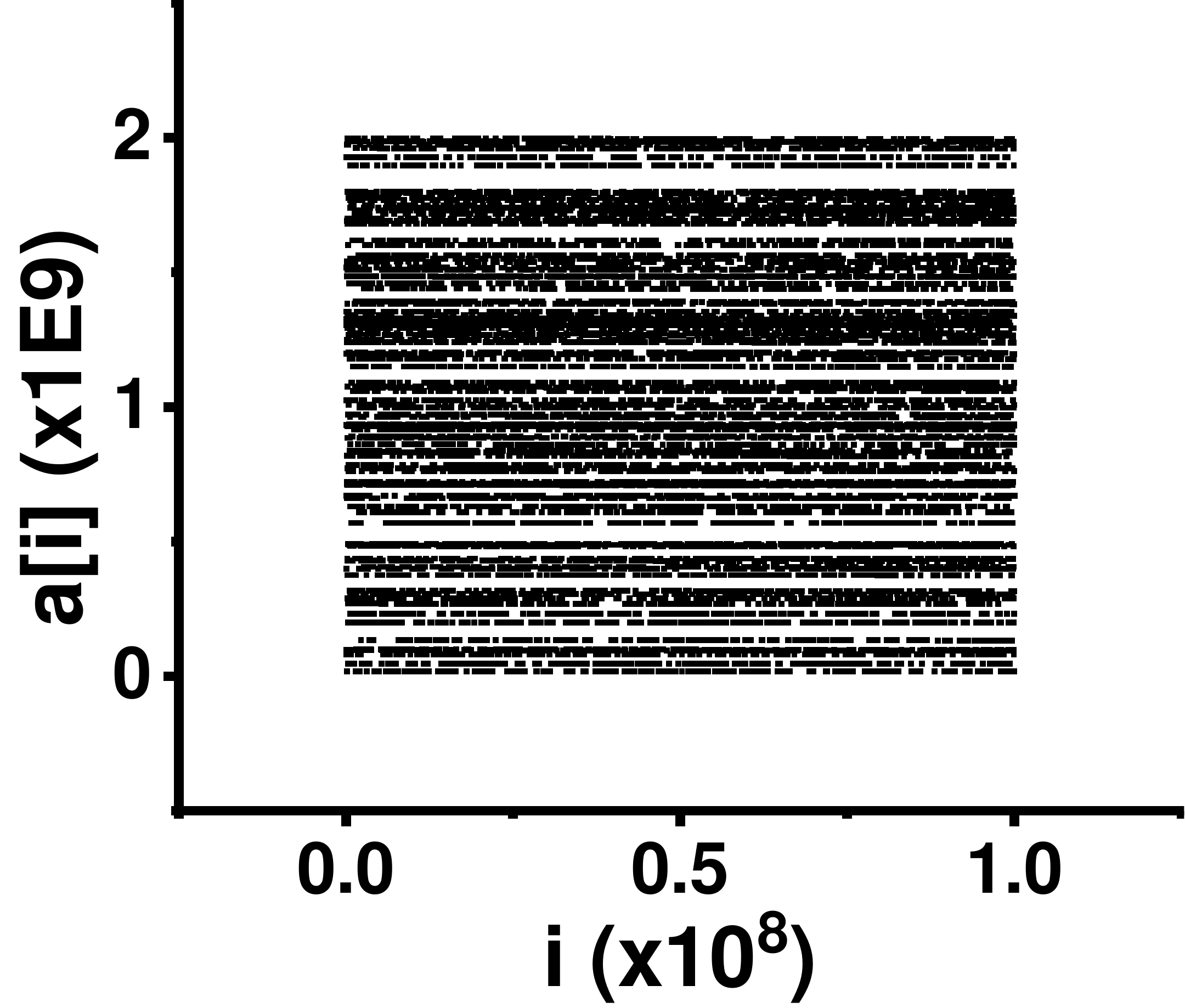}&
     \includegraphics[width=.27\columnwidth]{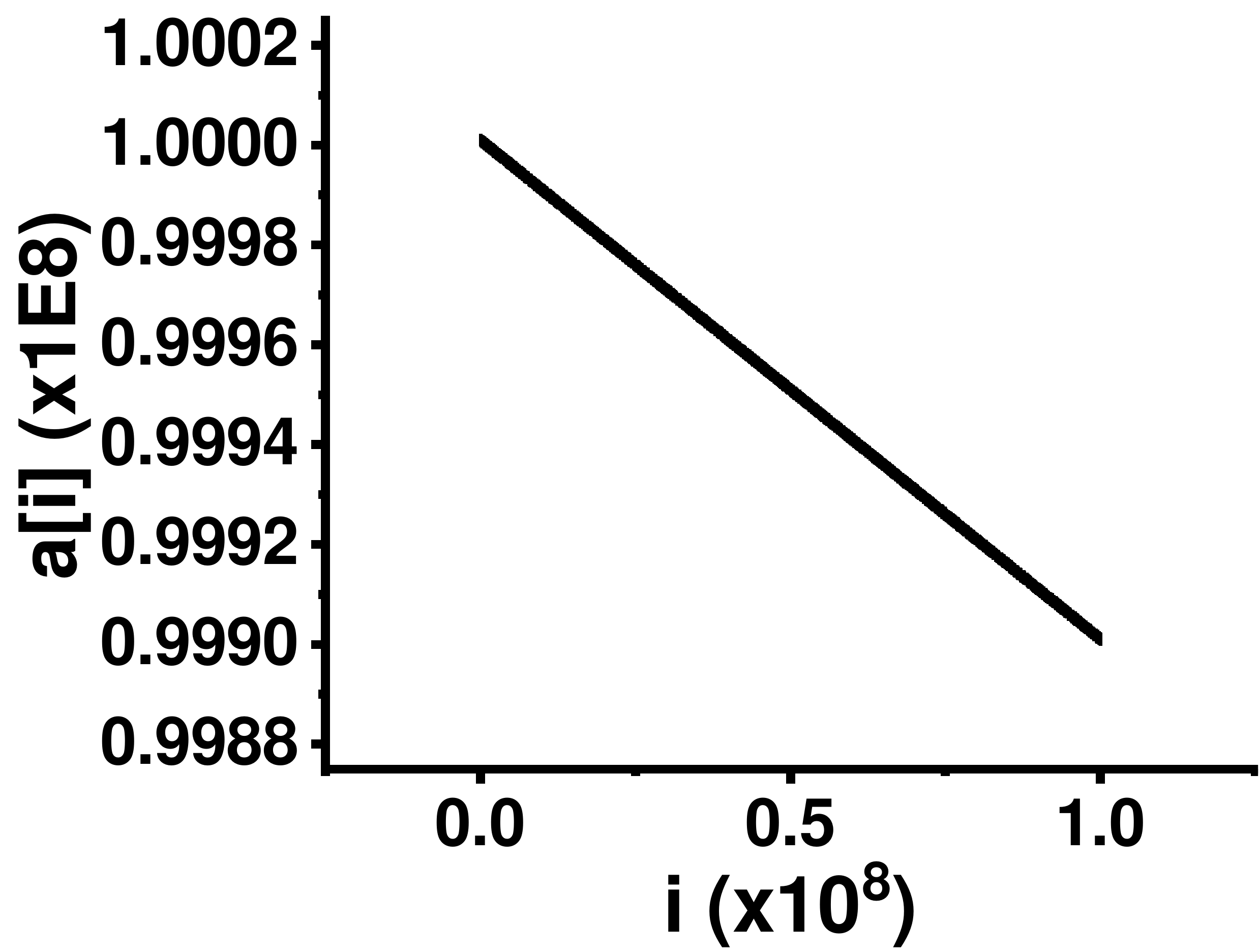} &
     \includegraphics[width=.27\columnwidth]{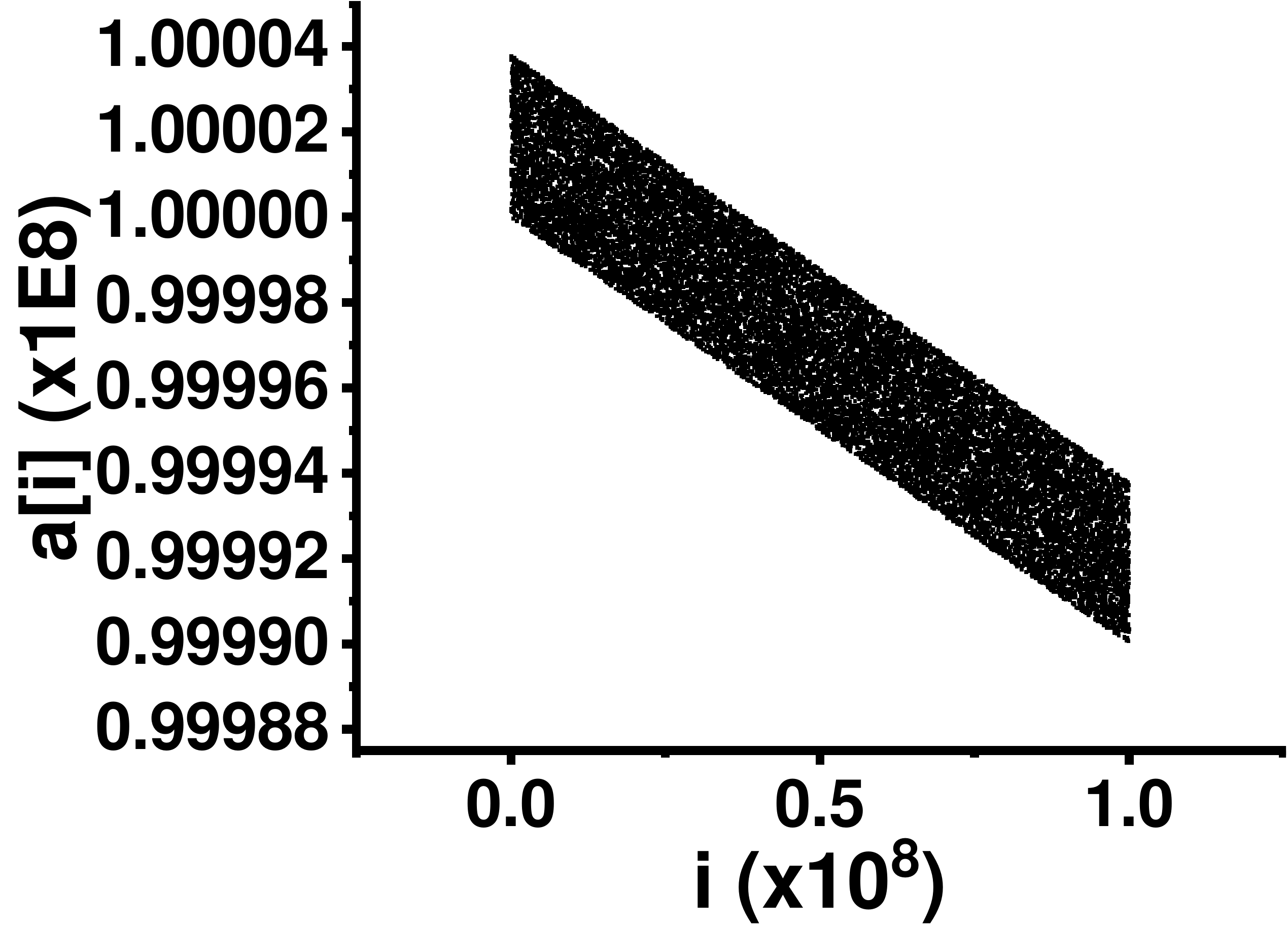} \\
     (a) & (b) & (c) & (d)  \\
   \end{tabular}
   \caption{\small\textbf{Examples to show input data patterns of LIS.} We fix input size $n=10^8$.
   (a). \textit{segment pattern}, with output size of $10$.
   (b). \textit{segment pattern}, with output size of $300$.
   (c). \textit{line pattern} with output size of $1000$.
   (d). \textit{line pattern} with output size of $3000$.
   \label{fig:lis_pattern}
   }
 \end{minipage}\hfill
\end{figure*} 

%% file: appendix.tex
\section{More Details for the Data Structures}
\label{app:datastructure}
\myparagraph{Parallel Augmented BST (\pabst{}) with Batch Operations.}
The base data structure we use for range queries is parallel augmented balanced binary search trees (\pabst{}). In particular, a PA-BST stores an ordered map of \element{s}, with \emph{keys} and \emph{values}. The \element{s} are sorted in the tree by the keys.
In addition, each tree node also maintains an \emph{augmented value}, which is some aggregate information about all \element{s} in its subtree.
The augmented value can be viewed as an abstract sum of the \element{s} in the subtree based on any associative operations.
This is maintained to support fast range sum queries. A simple example of the augmented value is to maintain the sum (or min/max) of values in the subtree.
In this case, reporting the sum of values (or min/max) in a certain key range can be performed in $O(\log n)$ time by observing at most $O(\log n)$ relevant
tree nodes and subtree augmented values.

To formally define the augmentation, we use the idea proposed in Sun et al.~\cite{sun2018pam}.
Suppose the map maintains key-value pair of type $K\times V$, we define the augmented value of type $A$ by two functions and an identity augmented value $I_A$.

\begin{itemize}
  \item A base function $g:K\times V\mapsto A$, which computes the augmented value of a single \element{} by its key and value.
  \item A combine function $f:A\times A\mapsto A$, which combines (adds) two augmented values into a new augmented value. It has to be associative.
  \item The identity $I_A\in A$ of $f$ on $A$.
\end{itemize}

In other words, $(A,f,I_A)$ is a monoid. In the previous value-sum example, the based function $g=(k,v)\mapsto v$, the combine function $f=(a_1,a_2)\mapsto a_1+a_2$, and $I_A=0$.

We can also perform parallel batch operations on \pabst{}~\cite{blelloch2016just,sun2018pam} (or general BST without augmentation). It has been shown that one can construct a \pabst{} of size $n$ in $O(n \log n)$ work and $O(\log^2 n)$ span\footnote{This is dominated by the parallel sorting on $n$ elements. Using the sorting algorithm in \cite{blelloch2020optimal}, we can also achieve $O(n \log n)$ work in expectation and $O(\log n)$ span \whp{}.}, and flatten a \pabst{} into a sorted array in $O(n)$ work and $O(\log n)$ span.
One can also perform batch operations including multi-find, multi-insert, multi-delete, and multi-update in $O(m\log (n/m))$ work and $O(\log n\log m)$ span on a tree of size $n$ and sorted batch of size $m\le n$~\cite{blelloch2016just,sun2018pam}.
One can also combine (compute the union of) two \pabst{} of size $n$ and $m\le n$ in $O(m\log (n/m))$ work and $O(\log n)$ span.
All these algorithms use a simple divide-and-conquer scheme to parallelize the operations, and the details can be found in~\cite{blelloch2020optimal}.
The algorithms also properly update the augmented values if needed.

Using a \pabst{}, we can report the augmented value of any key range in $O(\log n)$ work.

\myparagraph{Parallel Nested BSTs.} We can also nest the parallel BSTs to build a two-level index as a multi-map. %In particular, for a map of \element{s} with \emph{keys} and \emph{values}, the elements will be ordered first by keys and organize all elements with the same key together and sorted by values.

The elements will be ordered by keys and multiple entries with the same key will all be kept in the data structure.
In this case, we will first build a parallel BST, referred to as the \emph{primary tree}, indexed on the key. All \element{s} with the same key will be organized as another BST (indexed on the values), referred to as the \emph{secondary tree} associating with the corresponding key in the outer tree. Both the inner tree and the outer tree can be further augmented. The space needed to store $n$ \element{s} in a nested BST is $O(n)$. Note that this requires the value type to support a comparison function. This is true for the applications used in this paper.

We can extend the parallel algorithms on BSTs to the nested BSTs using a similar divide-and-conquer algorithm. A batch of multi-find, multi-insert, multi-delete, and multi-update can be performed in $O(m\log (n/m))$ work and $O(\log n\log m)$ span on a tree of size $n$ and sorted batch of size $m$~\cite{blelloch2016just,sun2018pam}.

\myparagraph{Parallel Augmented 2D Range Trees.} Some algorithms in this paper needs a \defn{2D (augmented) range query}. In particular, for a set of points on a 2D planar, the 2D augmented range query asks for the ``augmented value'' of all points in a given rectangle, where the ``augmented values'' are defined, similarly to \pabst{}, based on any associative operation.
This query can be answered by a \defn{2D Range Tree} with an appropriate augmentation.
A 2D range tree is a 2-level tree where the \emph{outer tree} is an index of the x-coordinates of the points.
Each tree node maintains an \emph{inner tree}, which stores the same set of points in its subtree, but is an index of the y-coordinates of these points.
This essentially can be understood as a \pabst{} where the inner trees are the augmented values, and thus the combine function is a parallel union on the two inner trees~\cite{sun2019parallel}. Each inner tree can be further augmented based on the application.
For example, if the inner tree is augmented with the maximum value of each subtree, then the range tree could answer any rectangle-max query (report the maximum value for all points in a given rectangle) in polylogarithmic time. This is exactly the query needed in our parallel LIS algorithm.

To formally define a range tree, one needs to specify the comparison function for the two dimensions ($<_x$ and $<_y$), respectively. If augmentation is needed, the corresponding \emph{base} and \emph{combine} functions also need to be defined to indicate how the augmented values of points should be combined.
The augmented value of any rectangle can be reported in $O(\log^2 n)$ time, for a dataset of $n$ points.

Note that the 2D range tree is \emph{not} a 2-level nested tree as defined above. The space usage of a range tree on $n$ points is $O(n\log n)$.

\ifconference{
\section{Additional Proofs}
% \section{Proof of the Activity Selection Algorithm}
\subsection{Proof of Lemma \ref{lem:activity} for the Activity Selection Algorithm}
\label{app:activityproof}

\input{act-proof.tex}
% \section{Proof of the LIS Algorithm}
\subsection{Proof of Lemma \ref{lem:backwardlogn} for the LIS Algorithm}
\label{app:lis}
\par

\input{lis-proof}
% \section{Proofs for the MIS Algorithm}
\subsection{Proof of Theorem \ref{thm:miscost} for the MIS Algorithm}
\label{app:mis:proof}
\par
\input{mis-proof}
}
\fi

\ifconference{
\section{Background of Parallel SSSP}
\label{app:sssp}
\input{sssp-bg.tex}
}
\fi

\ifconference{
\section{Unlimited Knapsack}
\label{app:knapsack}
\input{knapsack}
}
\fi

\ifconference{
\section{Huffman Tree}
\label{app:huffman}
\input{huffmantree}
}
\fi

\section{Whac-A-Mole Problem}
\label{app:mole}
Whac-A-Mole problem is a similar problem to the LIS problem. The setting is about a line of holes\footnote{A more commonly-seen setting used in real games is to have a 2D grid of holes. Here for simplicity we use 1D number line as an example. We note that our idea also apply to the 2D setting.}, from which moles can pop up out over time.
Assume there are $n$ moles sorted by their popping-up time, and the $i$-th mole will appear at time $t_i$ and coordinate $p_i$.
The moles only stay for a unit of time, and will disappear at time $t_i+1$.
The player will control a hammer and whack the mole when its above ground.
The hammer can struck a mole $i$ when it is also at $p_i$ at time $t_i$. Moving the hammer from position $p$ to $p'$ takes $|p-p'|$ time.
The goal is to hit the maximum number of moles.

A simple solution to this problem is to use a sequential iterative algorithm using dynamic programming. Let $dp[i]$, called the \dpvalue{} at mole $i$, be the number of moles one can whack up to the $i$-th mole, and must hit the $i$-th one. The algorithm will process all moles based the order of $t_i$. The DP recurrence of this problem is

\begin{equation}\label{eqn:moledp}
  D[i]=\max_{j<i,|p_j-p_i|\le |t_j-t_i|}D[j]+1
\end{equation}

We can define a later mole $y$ is compatible with mole $x$ when one can hit $y$ after hitting $x$. We can easily verify this problem is \phaseparallel{}.
In fact, the structure is similar to the LIS problem. The rank of a mole $x$ is the maximum number of moles one can hit ending with $x$.

Note that the condition in \cref{eqn:moledp} $|p_j-p_i|\le |t_j-t_i|$ can be transformed to

\begin{align}
  t_j+p_j &< t_i+p_i \\
  t_j-p_j &< t_i-p_i
\end{align}

If we view each mole $i$ as a point on 2D planar with coordinate $(t_i+p_i,t_i-p_i)$, then all the other moles that a mole $x$ relies on are in a 2D rectangle range.
This again is the same as LIS problem. Using the similar pivoting idea, we can achieve $O(n\log^3 n)$ work and $O(\rank(S)\log^2 n)$ span, where $\rank(S)$ is the rank of all input objects, which is also the maximum number of moles one can hit for the entire input.

When extending the setting to 2D grid instead on 1D number line, the problem requires a 3D range query, which adds up an extra $O(\log n)$ factor to both work and span.

\ifconference{
\section{Graph Coloring, Matching and Other Algorithms}
\label{app:graphcoloring}
\input{graphcoloring}
}
\fi

\ifconference{
\section{More Experimental Results}
\label{app:exp}

\subsection{Huffman Tree}
\label{app:exp:huffman}

\input{huffmantree-exp}
\subsection{Experiments on LIS}
\label{app:exp:lis}
\myparagraph{Implementation Details.}
\input{lis-exp}
}
\fi
\hide{
\begin{table}[t]
\centering
\small
\begin{tabular}{rrrrrc}
\multicolumn{6}{c}{\textbf{In the segment pattern}}\\
\hline
\multicolumn{1}{c}{Output}  &  \multicolumn{1}{c}{Classic} & \multicolumn{1}{c}{Ours} & \multicolumn{1}{c}{Ours} & \multicolumn{1}{c}{Spd.} & \multicolumn{1}{c}{Average \# of}\\
\multicolumn{1}{c}{Size}  &  \multicolumn{1}{c}{seq.}  & \multicolumn{1}{c}{seq.} & \multicolumn{1}{c}{par.} & \multicolumn{1}{c}{(Our $\frac{\text{seq.}}{\text{par.}}$)} & \multicolumn{1}{c}{Wake-ups}\\
\hline
3 &  40.49  & 482  &  11.54  & 41.78 & 1.67\\
\hline
10  &   37.96 & 709 & 14.64 & 48.47 & 2.55\\
\hline
30  &   32.78 & 994 & 18.67 & 53.26 & 2.79\\
\hline
100 &  30.42  & 1630 & 27.13 & 60.08 & 3.35\\
\hline
300 &  27.52  & 2116 & 33.49 & 63.20 & 3.39\\
\hline
1000  &   29.96 & 2850 & 49.24 & 42.98 & 3.87\\
\hline
3000  &   27.58 & 3292 & 78.44 & 41.97 & 3.74\\
\hline
10000 &  27.59  & 3960 & 171.25  &  23.12  & 3.66\\
\hline
\end{tabular}
\\
\begin{tabular}{rrrrrc}
\multicolumn{6}{c}{\textbf{In the line pattern}}\\
\hline
\multicolumn{1}{c}{Output}  &  \multicolumn{1}{c}{Classic} & \multicolumn{1}{c}{Ours} & \multicolumn{1}{c}{Ours} & \multicolumn{1}{c}{Spd.} & \multicolumn{1}{c}{Average \# of}\\
\multicolumn{1}{c}{Size}  &  \multicolumn{1}{c}{seq.}  & \multicolumn{1}{c}{seq.} & \multicolumn{1}{c}{par.} & \multicolumn{1}{c}{(Our $\frac{\text{seq.}}{\text{par.}}$)} & \multicolumn{1}{c}{Wake-ups}\\
\hline
3 & 17.46 & 489  &  16.30  & 30.03 & 1.36\\
\hline
10  &  20.27  & 650  &  18.57  & 35.04 & 1.97\\
\hline
30  &  24.11  & 1141 & 28.44  &  40.11  & 3.20\\
\hline
100 & 21.10  &  2194  &  44.37 &  49.44  & 4.58\\
\hline
300 & 18.84 & 3281 & 56.07  &  58.51  & 5.82\\
\hline
1000  &  17.81  & 4053 &  79.50 &  41.27 &  7.18\\
\hline
3000  &  17.68  & 4683 & 113.72 & 41.18 & 8.39\\
\hline
10000  &  17.56  & 4852 & 118.06 & 41.10 & 8.41\\
\hline
\end{tabular}
\vspace{-.1in}
\caption{\label{tab:exp:lis} \small
\textbf{The detailed experiment data on LIS.}
All times are in seconds.
}
\end{table}
}

\balance